\documentclass[a4paper]{article}
\usepackage[top=1.5cm]{geometry}
\usepackage{amsmath,amsthm,amssymb}
\usepackage{caption,subcaption,graphicx}
\usepackage{mathtools} % for coloneqq
\usepackage{authblk}
\usepackage{algorithm,algpseudocode}

\usepackage{tabularx,multirow,calc}
\newcolumntype{Y}{>{\centering\arraybackslash}X}
\usepackage{thmtools}

\usepackage{hyperref}
\usepackage[nameinlink]{cleveref}

\usepackage{float}
\usepackage{enumitem}
\usepackage{nicematrix}

\newcommand{\ket}[1]{|{#1}\rangle}
\newcommand{\bra}[1]{\langle{#1}|}
\newcommand{\kb}[1]{\mathinner{|{#1}\rangle\langle{#1}|}}
\newcommand{\ketbra}[2]{\mathinner{|{#1}\rangle\langle{#2}|}}
\newcommand{\braket}[2]{\mathinner{\langle{#1}|{#2}\rangle}}
\newcommand{\norm}[1]{\left\lVert#1\right\rVert}
\DeclarePairedDelimiter{\ceil}{\lceil}{\rceil}
\DeclareMathOperator{\erf}{erf}

\let\originalleft\left
\let\originalright\right
\renewcommand{\left}{\mathopen{}\mathclose\bgroup\originalleft}
\renewcommand{\right}{\aftergroup\egroup\originalright}

\newcommand{\ee}[1]{\mathbb{E}\left[{#1}\right]}
\newcommand{\pp}[1]{\mathbb{P}\left({#1}\right)}

\newcommand*\diff{\mathop{}\!\mathrm{d}}
\newcommand{\independent}{\perp\!\!\!\perp}

\newcommand*\polylog{\mathop{}\!\mathrm{polylog}}

\makeatletter
\newcommand{\doublewidetilde}[1]{{
  \mathpalette\double@widetilde{#1}
}}
\newcommand{\double@widetilde}[2]{
  \sbox\z@{$\m@th#1\widetilde{#2}$}
  \ht\z@=.9\ht\z@
  \widetilde{\box\z@}
}
\makeatother

\numberwithin{equation}{section}
\newcommand{\tagaligneq}{\refstepcounter{equation}\tag{\theequation}}

\theoremstyle{plain}
\newtheorem{thm}{Theorem}[section]
\newtheorem{theorem}[thm]{Theorem}
\newtheorem{proposition}[thm]{Proposition}
\newtheorem{lemma}[thm]{Lemma}
\newtheorem{corollary}[thm]{Corollary}

\theoremstyle{definition}
\newtheorem{definition}[thm]{Definition}
\newtheorem{example}[thm]{Example}
\newtheorem{assumption}[thm]{Assumption}
\newtheorem{notation}[thm]{Notation}
\newtheorem{convention}[thm]{Convention}

\usepackage{refcount} % for \getrefnumber
\newcounter{assumptionplusdummy}
\makeatletter
\newenvironment{assumptionplus}[1]{%
  \refstepcounter{assumptionplusdummy} % unique anchor for hyperlink
  \phantomsection
  \protected@edef\@currentlabel{\getrefnumber{#1}$^{+}$}%
  \protected@edef\@currentlabelname{Assumption \getrefnumber{#1}$^{+}$}%
  \def\Hy@tempa{assumptionplusdummy}%
  \expandafter\hyper@makecurrent\expandafter{\Hy@tempa}%
  \begin{trivlist}
  \item[\hskip\labelsep \bfseries Assumption \getrefnumber{#1}$^{+}$.]%
}{%
  \end{trivlist}
}
\makeatother

\theoremstyle{remark}
\newtheorem{remark}[thm]{Remark}

\title{Quantum analog-encoding for correlated Gaussian vectors and their exponentiation with application to rough volatility}
\date{\today}
\author[1]{Tassa Thaksakronwong}
\author[2]{Koichi Miyamoto}
\affil[1]{Graduate School of Engineering Science, the University of Osaka, Toyonaka, Osaka, Japan}
\affil[2]{Center for Quantum Information and Quantum Biology, the University of Osaka, Toyonaka, Osaka, Japan}

\begin{document}
\maketitle

\begin{abstract}
Quantum computing may speed up numerical problems involving large matrices that are demanding for classical computers, and active research on this possibility is ongoing.
In this paper, we propose quantum algorithms for the exact simulation of a normalised correlated Gaussian random vector $\ket{x}=\vec{x}/\norm{\vec{x}}$, $\vec{x}\sim\mathcal{N}(0,\Sigma)$, and its exponentiation $\ket{e^{\vec{x}}}= e^{\vec{x}}/\lVert e^{\vec{x}}\rVert$. 
Assuming that an $O(\polylog N)$-gate-depth quantum data loader for the covariance matrix $\Sigma\in\mathbb{R}^{N\times N}$ is available, the gate-depth complexities for preparing $\ket{x}$ and $\ket{e^{\vec{x}}}$ are $\widetilde{O}\left(\frac{\norm{\Sigma}_F}{\lambda_{\max}}\kappa^{1.5}\right)$ and $\widetilde{O}\left(\norm{\vec{x}}\frac{\norm{\Sigma}_F}{\lambda_{\max}}\kappa^{1.5}\right)$ respectively, where $\norm{\Sigma}_F$, $\lambda_{\max}$, $\kappa$ denote the Frobenius norm, maximal eigenvalue, and
condition number of $\Sigma$.
Motivated by financial applications, we provide an end-to-end resource analysis when $\vec{x}$ represents a sample path of a Riemann-Liouville or standard fractional Brownian motion, or of a stationary fractional Ornstein-Uhlenbeck process.
As a concrete example, we construct the quantum state encoding the rough Bergomi variance process and analyse the extraction of its time integral (related to the realised variance) via quantum amplitude estimation.
Depending on the behaviour of $\norm{\Sigma}_F/\lambda_{\max}$ and $\kappa$, the overall gate-depth complexity may range from sublinear to supercubic in $N$. Under specific conditions, subcubic complexity is achieved, indicating a quantum advantage over classical Cholesky-based sampling methods.
To our knowledge, this constitutes the first quantum algorithmic framework for the amplitude encoding of exponentiated Gaussian processes, providing foundational primitives for quantum-enhanced financial modelling.
\end{abstract}

\section{Introduction}
The multivariate Gaussian distribution is a fundamental model for correlated real-valued random vectors arising across statistical science and engineering, including genetics, biology, agriculture, medicine, and economics \cite{And03}. 
Exact simulation of a Gaussian vector $X\sim\mathcal{N}(0,\Sigma)$
with a dense covariance matrix $\Sigma\in\mathbb{R}^{N\times N}$ and no exploitable structure
typically requires $O(N^3)$ time to compute a square-root factorisation $\Sigma=SS^T$, most commonly via the Cholesky decomposition, and $O(N^2)$ time to generate a single sample through matrix-vector multiplication $X\stackrel{d}{=}SZ$, $Z\sim\mathcal{N}(0,I_N)$, where $I_N$ is the $N \times N$ identity matrix.
Approximate and FFT-based methods can reduce cost by exploiting structural assumptions such as near low rank, stationarity, or Toeplitz form, but typically at the expense of exactness or universality \cite[Appendix 3]{BMRS19}.

Practically important instances of this problem arise in mathematical finance: rough volatility models.
Early works in quantum finance largely focused on the Black-Scholes-Merton framework, in which the volatility of an asset is assumed to be a constant \cite{RGB18,SES+20,HGL+23}.
However, it has become increasingly clear in the mathematical finance literature that realistic modelling of volatility smiles and related market stylised facts requires fractional stochastic volatility models driven by stochastic processes that are ``rougher'' \cite{GJR18,Fuk23,BFF+23}, or ``smoother'' \cite{CR98}, than the standard Brownian motion.
Among such models, the family of log-normal fractional stochastic volatility models, including the rough Bergomi model \cite{BFG16,BDM21,BFF+23}, are defined through the exponentiation of auto-correlated Gaussian processes. 
The ability to prepare and manipulate exponentiated Gaussian sample paths therefore constitutes an essential algorithmic foundation for quantum approaches to rough volatility.

Recent quantum algorithms for fractional processes include \cite{BDP23}, which prepares quantum states that encode sample paths of some kinds of processes in their amplitudes (analog encoding): fractional generalisations of the Brownian bridge via the discrete sine transform, and L\'{e}vy processes via discretised stochastic integrals, both implemented by the quantum Fourier transform (QFT).
The former's use of DST provides substantial quantum speedup, but, in effect, imposes boundary constriants $Y_0=Y_T=0$ and restricts state preparation to uniform grids. 
Both constructions rely on either truncated Karhunen-Lo\`{e}ve expansions or discretisations and hence produce sample paths with covariances that approximate, rather than equal, the target covariance.
Therefore, an exact, structure-free quantum analogue of the covariance square-root factorisation-based simulation, in particular the exact simulation of fractional Brownian motion sample paths, has not yet been achieved.
Similarly, analog encoding of exponentiated fractional Gaussian processes, which need to be computed in some models,  also remains an unresolved issue in quantum finance \cite{BDP23,PSC+24}.

In this work, we propose quantum algorithms for the analog encoding of correlated Gaussian random vectors with exact covariances and their exponentiation. 
Our first main result is the preparation of the normalised state $\ket{x}=\vec{x}/\norm{\vec{x}}$, where $\vec{x}$ is a realisation of $X\sim\mathcal{N}(0,\Sigma)$.
Under the assumption that $O(\polylog N)$-gate-depth data loading is available for the covariance matrix $\Sigma\in\mathbb{R}^{N\times N}$,
the gate-depth complexity is given by $\widetilde{O}\left(\frac{\norm{\Sigma}_F}{\lambda_{\max}}\kappa^{1.5}\right)$,
where $\norm{\Sigma}_F$ is the Frobenius norm,
$\lambda_{\max}$ is the maximal eigenvalue, and
$\kappa$ is the condition number of $\Sigma$.
When $\Sigma$ is well-conditioned, i.e.\ $\kappa=O(1)$, it holds that the ratio ${\norm{\Sigma}_F}/{\lambda_{\max}}$ scales as $O(\sqrt{N})$ and our method's gate-depth complexity has a sub-linear dependence on the dimension $N$.

Our second main result is the preparation of the normalised state of the form
$\ket{\vec{f} \odot e^{c\vec{x}}}=(\vec{f} \odot e^{c\vec{x}})/\lVert \vec{f} \odot e^{c\vec{x}} \rVert$, where $\vec{f}\in\mathbb{R}^N\backslash\{\vec{0}\}$, $c\in\mathbb{R}\backslash\{0\}$, and $\odot$ and $e^{c\vec{x}}$ denote the component-wise product and exponentiation, respectively.
This form includes exponentiated Gaussian processes and log-normal fractional stochastic volatility models as special cases.
The corresponding complexity has an additional multiplicative factor $\norm{\vec{x}}$ to the preparation cost of $\ket{x}$, namely $\widetilde{O}\left(\norm{\vec{x}}\frac{\norm{\Sigma}_F}{\lambda_{\max}}\kappa^{1.5}\right)$.
This framework applies to general Gaussian processes with positive definite covariance matrices,
including Riemann-Liouville fractional Brownian motion, standard fractional Brownian motion, and stationary fractional Ornstein-Uhlenbeck processes for any Hurst parameter $H\in(0,1)$.
For these concrete classes, we provide a numerical analysis to determine the growth of
${\norm{\Sigma}_F}/{\lambda_{\max}}$ and $\kappa$ with respect to $N$, where the power-law behaviour is observed and the empirically estimated exponents are used to determine the overall complexity of our algorithm in terms of $N$.

As a concrete application, we provide an end-to-end construction of the quantum state embedding the rough Bergomi variance process and analyse the extraction of its time integral, also known in a financial context as the continuous limit of the realised variance, using quantum amplitude estimation (QAE) applied to the corresponding Riemann sum approximation.
The resulting overall gate-depth complexity dependence on $N$ is sub-cubic but super-quadratic for all $H\in(0,1)$, which shows the advantage compared to the classical Cholesky decomposition-based approach with a cubic complexity. 
For rough Bergomi-type models driven by more general Gaussian processes whose covariance matrices are well-conditioned, the gate-depth complexity can be reduced further: super-linear but sub-quadratic in $N$. 
To our knowledge, this constitutes the first exact and structure-free quantum solution for the analog encoding of exponentiated fractional Gaussian processes, establishing a foundational primitive for quantum stochastic simulation and quantum models of rough volatility.

The organisation of this paper is as follows. 
\Cref{prelim-sec} summarises conventions, assumptions, and underlying principles used to develop the results presented in this work.
\Cref{preparing-ket-x-section} shows how to prepare $\ket{x}$ and estimate $\norm{\vec{x}}$.
\Cref{numerical-section-label} presents numerical experiments conducted to determine the exponents of the power law behaviour of ${\norm{\Sigma}_F}/{\lambda_{\max}}$ and $\kappa$ on $N$, for the case of a Riemann-Liouville fractional Brownian motion, a standard fractional Brownian motion, and a stationary fractional Ornstein-Uhlenbeck process on a uniform grid.
\Cref{exp-non-lin-section} shows how to adopt the non-linear transformation technique of \cite{RR23} to prepare the state of the form $\ket{\vec{f} \odot e^{c\vec{x}}}$, including $\ket{e^{\vec{x}}}$ and rough Bergomi variance $\ket{V}=V/\norm{V}$ as special cases.
Together, we show how to output the discrete sum of the form $\frac{1}{N}\sum_{i=1}^Nf_ie^{cx_i}$, including the sum $\frac{T}{N}\sum_{i=1}^N V_{t_i}$, which approximates the rough Bergomi integrated variance $\int_0^T V_t \diff t$, as a special case.

\section{Preliminaries}\label{prelim-sec}
Throughout this work, all assumptions are assumed hold, while those with the name `Assumption plus' are assumed hold only if explicitly stated at the moment.
\begin{definition}[Bachmann–Landau notations]\label{bachmann-landau-notn-defn-1}
Let $f,g$ be eventually non-negative functions defined over positive real numbers.
We write $f(x)=O(g(x))$ if there exist constants $C>0$ and $x_0>0$ such that $0\leq f(x)\leq Cg(x)$ for all $x\geq x_0$,
$f(x)=o(g(x))$ if
$\lim_{x\to\infty}{f(x)}/{g(x)}=0$,
and $f(x)=\Theta(g(x))$ if both $f(x)=O(g(x))$ and $g(x)=O(f(x))$ hold.
\end{definition}
\begin{definition}[Soft Bachmann–Landau notations]
Let $f,g$ be eventually non-negative functions defined over positive real numbers. For $\vartheta\in\{O,o\}$ (see \autoref{bachmann-landau-notn-defn-1}), we write $f(x)=\widetilde{\vartheta}\left(g(x)\right)$ if there exists some $k\in\mathbb{R}$ such that $f(x)=\vartheta\left(g(x) \log^k(x)\right)$. We write $f(x)=\widetilde{\Theta}\left(g(x)\right)$ if there exist some $k_1,k_2\in\mathbb{R}$ such that both $f(x)=O\left(g(x) \log^{k_1}(x)\right)$ and $g(x)=O\left(f(x) \log^{k_2}(x)\right)$ hold. Note that the logarithmic factors are allowed to have negative powers.
\end{definition}
\subsection{Quantum computing preliminaries}
\begin{convention}[Quantum embedding of classical entries]
\label{convention-embed-classic-info}
In subsequent discussion, we define $n=\ceil{\log_2\left(N+1\right)}$ to be the number of system qubits required to encode $N$ classical data points. 
Given a classical vector $\vec{w}=(w_1,\dots,w_N)\in\mathbb{R}^N$, define $\ket{w}\coloneqq \vec{w}/\norm{\vec{w}}$ to be the corresponding (normalised) quantum state.
We adopt the convention that for $i=1,\dots,N$, the value $w_i/\norm{\vec{w}}$ is embedded in the amplitude of the computational basis state $\ket{i}$, which implies that the basis state $\ket{0}$ is unused for storing classical information. Accordingly, to multiply a matrix to a quantum state with correct indices, an $N\times N$ matrix is also embedded using bases $\ketbra{i}{j}$ whose indices start from $1$, i.e.\ $i,j=1,\dots,N$.
\end{convention}
\begin{convention}[Gate set]
We use the term `elementary gates' to mean arbitrary single-qubit and two-qubit gates\footnote{Since any two-qubit unitary can be decomposed into at most 3 CNOT gates and 15 single-qubit gates \cite{VW04}, this essentially reduces to the same asymptotic cost when using only single-qubit gates and CNOT gates.}.
\end{convention}
\begin{convention}[Complexity unit]\label{convention-gate-count-vs-depth}
The implementation cost of any quantum circuit presented in this work is measured in terms of either
elementary gates (gate counts) required, or elementary gate depth required.
Note that since any quantum algorithm requiring $O(d)$ gates can have at most $O(d)$ gate depth,
an algorithm whose complexity is given in gate-count unit may be referred to in gate-depth unit without affecting correctness.
\end{convention}
\begin{assumption}\label{assump-vector-loader}
For a real vector $\vec{w}\in\mathbb{R}^N$ whose entries are known, there exists a state-preparation unitary $U_{\ket{w}}$ such that
\[
U_{\ket{w}}\ket{0}^{\otimes n} 
= \ket{w}
\coloneqq
\frac{\vec{w}}{\norm{\vec{w}}}=\frac{1}{\norm{\vec{w}}}\sum_{i=1}^N w_i \ket{i},
\]
where the amplitudes of $\ket{0}$ and of $\ket{j}$ for $N+1\leq j \leq 2^n-1$ (if exist) are understood to be zero. We further assume that $U_{\ket{w}}$ can be implemented using $T_{U_{\ket{w}}}$ elementary gates. 
\end{assumption}
Loading a classical vector into unary-based quantum amplitudes using $O(N)$ qubits and $O(\log N)$ gate depth 
has been realised on actual hardware \cite{JDM+21}.
To create the circuit description of the corresponding state-preparation unitary, their method also takes $O(N \polylog N)$ classical computation time, but this is only needed once.
Since our algorithms operate via the standard, binary-based, amplitude encoding\footnote{Some works in the literature use the term `binary encoding' to mean `basis encoding', aka digital encoding, but what we mean here is the amplitude encoding, aka analog encoding, with the usual usage of $n=\ceil{\log_2 (N+1)}$ (or $n=\ceil{\log_2 N}$) qubits to encode $N$ classical data, as opposed to the unary-based basis usage which uses $N$ qubits to encode $N$ classical data.},
an additional conversion step from unary-based to binary-based amplitude encoding is required. This procedure, described in \cite[Proposition C.1]{BDP23}, uses $O(N)$ ancillary qubits and $O(N\log N)$ gates, but with depth $O(\log^2 N)$.
We formalise this into the following assumption.
\begin{assumptionplus}{assump-vector-loader}\label{assump-plus-TU-ket-y-log2-N}
We assume that $U_{\ket{w}}$ from \autoref{assump-vector-loader} can be implemented using $O(N)$ ancillary qubits and 
\[
T_{U_{\ket{w}}}=O(\log^2 N)
\]
elementary gate depth, achievable via parallel arrangement of $O(N\log N)$ elementary gates, together with a one-time $O(N \polylog N)$ classical computation to create the circuit description of $U_{\ket{w}}$. The unitaries $U_{\ket{w}}$ and $U_{\ket{w}}^{\dagger}$ are reusable without requiring classical re-computation. 
\end{assumptionplus}

\begin{definition}[Block-encoding] \label{block-encoding-def}
Given a matrix $A\in\mathbb{C}^{N \times N}$, any unitary matrix $U_A$ is called an $(\alpha,a,\varepsilon)$-block-encoding of $A$ if
$
\norm{A-\alpha\Big(\bra{0}^{\otimes a}\otimes I_N\Big) U_A \Big(\ket{0}^{\otimes a}\otimes I_N\Big)}\leq \varepsilon
$
holds, where $\alpha>0$, $a\in\mathbb{N}\cup\{0\}$, and $\varepsilon\geq 0$.
\end{definition}
\begin{assumption}\label{assump-abstract-block-encoding}
Suppose we have access to an $(\alpha_{U_\Sigma},a_{U_\Sigma},\varepsilon_{U_\Sigma})$-block-encoding $U_{\Sigma}$ of $\Sigma$ for some $\alpha_{U_\Sigma}>0$, $a_{U_\Sigma}\in\mathbb{N}\cup\{0\}$, and $\varepsilon_{U_\Sigma}\geq 0$. Suppose that $U_{\Sigma}$ can be implemented using $T_{U_\Sigma}$ elementary gates.
\end{assumption}
In developments of quantum linear system solver, various works in the literature (e.g., \cite[Theorem 5.1 and Theorem A.1 in the full version]{KP16}, \cite[Lemma 1]{WZP18}, and \cite[Theorem 4]{CGJ19}\footnote{Theorem 1 of the full version: arXiv:1804.01973.}) assume state-preparation circuit in a controlled fashion that allows for the embedding of 
dense matrices, which covariance matrices in most interesting applications generally are, into quantum algorithms\footnote{See how \autoref{assump-QROM-like-unitary} implies \autoref{block-encoding-of-sigma} in the subsequent discussion.}.
Quantum random access memory (QRAM) \cite{GLM08} is known to provide a sufficient condition for the realisation of such a data-loading circuit.
However, we do not require the ability to rewrite or update the data structure, but only to embed a fixed matrix. 
In this sense, assuming QRAM may be excessive. 
Rather, we may instead refer to this as a data loader with polylogarithmic depth and formalise it as follows.
\begin{assumptionplus}{assump-abstract-block-encoding}\label{assump-QROM-like-unitary}
Let $A\in\mathbb{R}^{N\times N}$ be a matrix whose entries are known and let $a=\ceil{\log_2 (N+1)}$ be the number of qubits in the ancillary register so that it can represent $\{\ket{0},\ket{1},\dots,\ket{N}\}$ states. We assume the ability to construct the following unitary operators 
\begin{align*}
&U_{\mathcal{R}}:\ket{0}^{\otimes a}\ket{i}\mapsto \ket{\psi_i}\coloneqq\sum_{k=1}^{N} \frac{A_{i,k}}{\norm{A_{i,\cdot}}} \ket{i,k}, \\
&U_{\mathcal{L}}:\ket{0}^{\otimes a}\ket{j}\mapsto \ket{\phi_j}\coloneqq\sum_{\ell=1}^N\frac{\norm{A_{\ell,\cdot}}}{\norm{A}_F}\ket{\ell,j}
\end{align*}
that hold for all $i,j\in\{1,\dots,N\}$, by using $O(\polylog N)$ elementary gate depth. Here, the second register is assumed to be of $n=\ceil{\log_2 (N+1)}$ qubits, $A_{i,\cdot}$ denotes the $i$-th row vector of $A$, and $\norm{\cdot}_F$ denotes the Frobenius norm.
\end{assumptionplus}
\begin{remark}
Note the register swap between the control and ancilla registers in the operation of $U_{\mathcal{R}}$, which was not originally performed in
the statement of \cite{KP16,WZP18,CGJ19}.
We need this modification to ensure that the resulting block-encoding complies with \autoref{block-encoding-def}, i.e.\ the embedded matrix is indeed located in the upper-left corner of its block-encoding, cf.\ the proof of \autoref{block-encoding-of-sigma}. 
Since both registers consist of $a=\ceil{\log_2{(N+1)}}=n$ qubits, this modification can be done simply by applying SWAP gates qubit-by-qubit after the data-loading procedure, using at most $\ceil{\log_2{(N+1)}}$ SWAP gates. Hence, this modification does not affect the assumed implementation cost of $O(\polylog N)$ gate depth.
\end{remark}
The following proposition shows that if \autoref{assump-QROM-like-unitary} holds, we can construct a block-encoding of a dense matrix using $U_{\mathcal{L}}$ and $U_{\mathcal{R}}$.
Similar results were given in \cite{GSLW19}\footnote{Theorem 50 of the full version: arXiv:1806.01838.} and \cite[Lemma 6]{CGJ19}\footnote{Lemma 25 of the full version: arXiv:1804.01973.}, but with a block-encoding error $\varepsilon$. However, in this work, we assume that the preparations of the states specified in \autoref{assump-QROM-like-unitary} via $U_{\mathcal{L}}$ and $U_{\mathcal{R}}$ are exact. Under this assumption, the resulting block-encoding has zero error, cf.\ the proof of \autoref{block-encoding-of-sigma}.
\begin{proposition}\label{block-encoding-of-sigma}
Under \autoref{assump-QROM-like-unitary}, the unitary $U_A\coloneqq U_{\mathcal{R}}^{\dagger}U_{\mathcal{L}}$ is a $(\norm{A}_F,\ceil{\log_2 (N+1)},0)$-block-encoding of $A$, and can be implemented using $T_{U_A}=O(\polylog N)$ elementary gate depth.
\end{proposition}
\begin{proof}
Note that for $i,j=1,\dots,N$, the states $\ket{\psi_i},\ket{\phi_j}$ preparable from $U_{\mathcal{R}},U_{\mathcal{L}}$ satisfy
\begin{align*}
\braket{\psi_i}{\phi_j}
=\sum_{k=1}^N\sum_{\ell=1}^N \frac{A_{i,k}}{\norm{A_{i,\cdot}}} \frac{\norm{A_{\ell,\cdot}}}{\norm{A}_F} \braket{i}{\ell}\braket{k}{j}
=\frac{A_{i,j}}{{\norm{A_{i,\cdot}}}} \frac{{\norm{A_{i,\cdot}}}}{\norm{A}_F}
=\frac{A_{i,j}}{\norm{A}_F}. \tagaligneq \label{inner-prod-psi-i-phi-j}
\end{align*}
Since the identity operator $I_N$ can be written in the computational basis\footnote{Recall that our convention uses indices starting from $\ket{1}$, not from $\ket{0}$, cf.\ \autoref{convention-embed-classic-info}.} as $I_N=\sum_{r=1}^N\kb{r}$,
 and $\ket{0}^{\otimes a} \otimes \kb{r}=\big(\ket{0}^{\otimes a}\ket{r}\big)\bra{r}$
by a standard property\footnote{More specifically, the so-called mixed-product property, which states $(A\otimes B)(C \otimes D)=AC\otimes BD$, provided that each matrix product makes sense. Note that scalars can be regarded as matrices of size $1\times 1$.} of tensor products, we have
\begin{align*}
&U_{\mathcal{R}}\left(\ket{0}^{\otimes a}\otimes I_N\right)
=U_{\mathcal{R}}\left(\ket{0}^{\otimes a}\otimes \sum_{r=1}^N\kb{r}\right)
=\sum_{r=1}^N U_{\mathcal{R}}\Big(\ket{0}^{\otimes a}\ket{r}\Big)\bra{r}
=\sum_{r=1}^N \ket{\psi_r}\bra{r}.
\end{align*}
In the same manner, $U_{\mathcal{L}}\left(\ket{0}^{\otimes a}\otimes I_N\right)
=\sum_{s=1}^N \ket{\phi_s}\bra{s}$.
Thus,
\begin{align*}
\Big(\bra{0}^{\otimes a}\otimes I_N\Big) U_{\mathcal{R}}^{\dagger}U_{\mathcal{L}} \Big(\ket{0}^{\otimes a}\otimes I_N\Big)
=\sum_{r=1}^N\sum_{s=1}^N \ket{r}
\braket{\psi_r}{\phi_s}
\bra{s}
\stackrel{\eqref{inner-prod-psi-i-phi-j}}{=}\frac{1}{\norm{A}_F}\sum_{r=1}^N\sum_{s=1}^N A_{r,s} \ketbra{r}{s}
=\frac{A}{\norm{A}_F},
\end{align*}
where the last equality follows according to \autoref{convention-embed-classic-info}. Consequently, we have that
\begin{align*}
&\norm{A-\norm{A}_F\Big(\bra{0}^{\otimes a}\otimes I_N\Big) U_{\mathcal{R}}^{\dagger}U_{\mathcal{L}} \Big(\ket{0}^{\otimes a}\otimes I_N\Big)}
=\norm{A-{\norm{A}_F}\frac{A}{{\norm{A}_F}}}
=0,
\end{align*}
where $a=\ceil{\log_2 (N+1)}$ under \autoref{assump-QROM-like-unitary}.
Following \autoref{block-encoding-def}, we have shown that $U_{\mathcal{R}}^{\dagger}U_{\mathcal{L}}$ is a $(\norm{A}_F,\ceil{\log_2 (N+1)},0)$-block-encoding of $A$. The total gate depth required follows directly as we use exactly one $U_{\mathcal{L}}$ and one $U_{\mathcal{R}}^{\dagger}$.
\end{proof}
\subsection{Mathematical preliminaries}
In accordance with the usual settings of the related classical algorithms,
we restrict our focus to only the simulation of a correlated Gaussian random vector with mean zero, $X\sim\mathcal{N}(0,\Sigma)$.
The letter $X$ is used solely for expressing a centered correlated Gaussian vector, with $\vec{x}\in\mathbb{R}^N$ denoting a realisation of $X$, while the letter $Z$ always denotes $Z\sim\mathcal{N}(0,I_N)$, a standard Gaussian vector, with $\vec{z}\in\mathbb{R}^N$ denoting a realisation of $Z$.
A Gaussian process $G=(G_t)_{t\in[0,T]}$ for $T>0$ is a stochastic process whose finite-dimensional distributions are Gaussian distributions. Similarly, we restrict our attention to only centered Gaussian processes. That is, for any finite index set $\mathcal{P}=(t_i)_{i=1}^N\subset [0,T]$, $(G_{t_1},\dots,G_{t_N})\sim\mathcal{N}(0,\Sigma)$ holds for some $\Sigma$.
Given a covariance matrix $\Sigma\in\mathbb{R}^{N\times N}$,
the simulation of a realisation $\vec{x}$ of $X\sim\mathcal{N}(0,\Sigma)$ is obtained by taking the product $\vec{x}\equiv S\vec{z}$ for $S$ being any matrix that satisfies $\Sigma = SS^T$.
When $\Sigma$ is given by the auto-covariance of a Gaussian process $G$ on some index set $\mathcal{P}$, $\vec{x}$ represents exact\footnote{It is called exact because the random vector $X$ has a mathematically identical law (distribution) to the law of $(G_{t_1},\dots,G_{t_N})$, since a centered Gaussian distribution is characterised solely by the covariance matrix.} values of a sample path of $G$ on that index set.
It can be shown easily that there exist inifinitely many such matrices $S$. 
The classical Cholesky-based simulation takes this $S$ to be the lower triangular matrix $L$ obtained from the Cholesky decomposition $\Sigma=LL^T$. 
In this work, we use the unique symmetric matrix square root $\Sigma^{1/2}$, as defined below.
\begin{proposition}[Uniqueness of the positive semi-definite square root]\label{horn-johnson-uniqueness-lemma}
Let $A\in\mathbb{R}^{N\times N}$ be a real symmetric positive semi-definite matrix.
Then there exists a unique real symmetric positive semi-definite matrix $B\in\mathbb{R}^{N\times N}$ such that $A = B^2$.
\end{proposition}

\begin{proof}
See \cite[Theorem~7.2.6]{HJ13} with $k=2$.
\end{proof}

\begin{definition}[Positive semi-definite square root]\label{defn-sigma12}
Given a real symmetric positive semi-definite matrix
$\Sigma\in\mathbb{R}^{N\times N}$, we denote by $\Sigma^{1/2}$ its unique real symmetric positive semi-definite square root from \autoref{horn-johnson-uniqueness-lemma} satisfying $\Sigma = (\Sigma^{1/2})^2$.
\end{definition}

\begin{proposition}[Spectral form of the positive semi-definite square root]\label{spectral-form-Sigma12-prop}
Let $\Sigma\in\mathbb{R}^{N\times N}$ be a real symmetric positive semi-definite matrix.
For any spectral decomposition $\Sigma=\sum_{k=1}^N \lambda_k \ket{v_k}\bra{v_k}$,
where $\{\lambda_k\}_{k=1}^N$ are the eigenvalues of $\Sigma$ and $\{\ket{v_k}\}_{k=1}^N$ is an orthonormal eigenbasis,
the positive semi-definite square root $\Sigma^{1/2}$ is given by
\[
\Sigma^{1/2} = \sum_{k=1}^N \sqrt{\lambda_k}\ket{v_k}\bra{v_k}. \tagaligneq \label{eigen-decom-express-sigma12}
\]
This expression holds for all choices of orthonormal eigenbasis.
\end{proposition}
\begin{proof}
The facts that $\Sigma$ must admit a spectral decomposition and that ${}^\forall k$, $\lambda_k\geq 0$ are standard.
If $\Sigma=\sum_{k=1}^N \lambda_k \ket{v_k}\bra{v_k}$ for some orthonormal eigenbasis $\{\ket{v_k}\}_{k=1}^N$, define $B=\sum_{k=1}^N \sqrt{\lambda_k} \ket{v_k}\bra{v_k}$.
Then $B$ is real symmetric positive semi-definite and satisfies $B^2=\Sigma$. By uniqueness (\autoref{horn-johnson-uniqueness-lemma}), we must have $B=\Sigma^{1/2}$.
Since the same argument applies to any other choice of orthonormal eigenbasis, the expression \eqref{eigen-decom-express-sigma12} holds regardless of that choice.
\end{proof}
\begin{remark}
Although a covariance matrix $\Sigma$ can have zero eigenvalues, 
it is well-known that $\det \Sigma=0$ is equivalent to the existence of an almost surely linear dependence between the components of $X\sim\mathcal{N}(0,\Sigma)$.
Since there are no particular incentives for generating linearly dependent Gaussian variables even in the classical settings, we restrict our focus to only the strictly positive definite covariance matrices.
\end{remark}
We denote by $\lambda_{\min}$, $\lambda_{\max}$, and $\kappa\equiv \lambda_{\max}/\lambda_{\min}$ the minimal, maximal eigenvalues, and the condition number of $\Sigma$.
In the context of path simulations, usually the problem size $N$ is the number of blocks in the time discretisation $\mathcal{P}^N_{[0,T]}=(t_i)_{i=0}^N$ of an interval $[0,T]$ one wants to sample a Gaussian path from. 
In such cases, the $(i,j)$-entries of the covariance matrix $\Sigma$ are computed on the grid $\mathcal{P}^N_{[0,T]}$ and therefore are functions of $N$, usually of $(t_i,t_j,N)$. 
This means the parameters  $\lambda_{\min}$, $\lambda_{\max}$, $\kappa$, and also $\norm{\Sigma}_F$, the Frobenius norm of $\Sigma$, are inherently functions of $N$.
Unlike $\norm{\Sigma}_F$ whose value can readily be calculated from the already known entries of $\Sigma$, 
the quantities $\lambda_{\max}$ and $\kappa$ are often unknown,
but are required in order to use the pre-existing quantum algorithm results, e.g.\ \autoref{prop-for-Hc2}.
A full spectral decomposition costs as much as $O(N^3)$ classically, destroying the purpose of developing a quantum algorithm.
Therefore, we do not assume knowledge of the exact values of $\lambda_{\max}$ and $\kappa$, but we do assume some form of estimates, which we formulate in what follows.
Our numerical analysis in \autoref{num-results-sub-sec} supports the plausibility of this assumption.
\begin{assumption}\label{assmp-estimate-lamb-max-kappa}
Suppose that we can estimate $\lambda_{\max}$ and $\kappa$ by $\widetilde{\lambda}_{\max}$ and $\widetilde{\kappa}$ satisfying 
\begin{equation*}
\begin{gathered}
\lambda_{\max}\leq \widetilde{\lambda}_{\max} \leq L \lambda_{\max}, \\
\kappa\leq \widetilde{\kappa} \leq K \kappa, 
\end{gathered}\tagaligneq \label{main-thm-col-assump2}
\end{equation*}
for some $L,K\geq 1$ that may depend on $N$. The values $\widetilde{\lambda}_{\max},\widetilde{\kappa},L,K$ are assumed known and available for algorithmic use.
\end{assumption}
\begin{assumptionplus}{assmp-estimate-lamb-max-kappa}\label{assump-LK-plus}
In addition to \autoref{assmp-estimate-lamb-max-kappa}, suppose that $L$ and $K$ are $\widetilde{\Theta}(1)$.
The values $\widetilde{\lambda}_{\max},\widetilde{\kappa},L,K$ are assumed known and available for algorithmic use.
\end{assumptionplus}
\begin{remark}\label{rem-abt-LK-assmp-plus}
\autoref{assump-LK-plus} implies that $L$ and $K$ may depend on $N$, but are at most polylogarithmic in $N$, and, essentially, 
our estimates are assumed to satisfy
$\widetilde{\lambda}_{\max}=\widetilde{\Theta}(\lambda_{\max})$ and $\widetilde{\kappa}=\widetilde{\Theta}(\kappa)$.
Empirical results show that this can be achieved in practical settings, see \autoref{remark-emp-result-support-assump}.
\end{remark}
In quantum computation, the accessible objects are normalised states such as $Z/\norm{Z}$ and $X/\norm{X}$. 
From a geometric perspective, it is a well-known result that $Z/\norm{Z}$ is uniformly distributed on the unit sphere 
$\mathcal{S}^{N-1}=\{\widehat{u}\in\mathbb{R}^N \mid \norm{\widehat{u}}=1\}$, 
while $\norm{Z}$ follows a chi-distribution. 
More importantly, the random variables $Z/\norm{Z}$ and $\norm{Z}$ are independent (see e.g.\ \cite[Exercise 3.3.7]{Ver18}). 
This independence has an important consequence that,
given a realisation $\ket{z}=\vec{z}/\norm{\vec{z}}$, where $\vec{z}$ is a realisation of a standard Gaussian vector $Z$, 
the transformed vector $\Sigma^{1/2}\ket{z}$, its norm $\lVert \Sigma^{1/2}\ket{z}\rVert$, and the normalised state
$
\ket{x}=\frac{\Sigma^{1/2}\ket{z}}{\lVert \Sigma^{1/2}\ket{z}\rVert}
$
are all in fact independent of $\norm{\vec{z}}$, since $\ket{z}\independent \norm{\vec{z}}$. 
Therefore, once we prepare $\ket{x}$ by an algorithm proprosed in \autoref{preparing-ket-x-section}
and we want to recover the corresponding $\norm{\vec{x}}$ such that $\norm{\vec{x}}\ket{x}\stackrel{d}{=} \vec{x}$ becomes a realisation of $X$ with correct covariance matrix $\Sigma$,
we can use this independence property to create $\norm{\vec{x}}$ from multiplying $\norm{\vec{z}}$ with $\lVert \Sigma^{1/2} \ket{z} \rVert$ since $\norm{\vec{z}}\ket{z}\stackrel{d}{=} \vec{z}$.
This means that only the normalised state $\ket{z}$ (and subsequently $\ket{x}$) is manipulated throughout the quantum circuit, while the information of $\norm{\vec{z}}$ is kept entirely as classical data and used only when classical reconstruction is required. 
This observation will be used repeatedly in our treatment of normalisation factors in the subsequent sections.

Finally, we introduce a technical property and formulate a corresponding assumption on bounds for the components of $\vec{x}$, particularly when $\vec{x}$ represents a sample path of a Gaussian process on a finite grid. 
This assumption (\autoref{assump-inf-norm-bound-for-vec-x}) is required and used only for the exponentiation procedure, to be developed in \autoref{exp-non-lin-section}.
\begin{definition}[$\ell_{\infty}$-norm]
For any $\vec{v}=(v_1,\dots,v_N)\in\mathbb{R}^N$, denote $\norm{\vec{v}}_\infty\coloneqq \max_{i=1,\dots,N}|v_i|$.
\end{definition}
It is widely known that the running supremum of a centered and continuous Gaussian process $G$ on $[0,1]$ has an exponential tail (see e.g.\ {\cite[Theorem 4.2]{Nou12}}, aka the Borell-TIS inequality).
This means that, with high probability $\beta\in[0,1)$, 
there exists a constant $\Xi>0$ such that $\pp{\sup_{u\in[0,1]}|G_u|\leq \Xi} \geq \beta$ (see \autoref{bound-xi-on-norm-x-infty}).
When $\vec{x}$ is generated by the covariance matrix of $G$ on a finite grid, as it is obtained from an exact simulation,
$\norm{\vec{x}}_\infty \leq \sup_{u\in[0,1]}|G_u|\leq \Xi$ holds regardless of $N$ with probability at least $\beta$.
If $G$ is self-similar, the same principle extends naturally to $[0,T]$, for a fixed $T>0$.
This motivates the following assumption.
We give a supplementary discussion about how to determine the size of $\Xi$ given $\beta$ and other related characteristics of $G$ in \autoref{appndx-inf-norm-bnd}.
\begin{assumption}\label{assump-inf-norm-bound-for-vec-x}
We restrict our attention to the realisations $\vec{x}$ of $X\sim\mathcal{N}(0,\Sigma)$ that satisfy
\[
{}^\forall N\in\mathbb{N},\quad\norm{\vec{x}}_\infty \leq \Xi,
\]
for some constant $\Xi>0$ independent of $N$. In other words, we assume $\Xi=\Theta(1)$.
\end{assumption}
\begin{remark}
Note that $\vec{x}$ depends on $N$ through both $\Sigma$ and the sampled $\vec{z}$, and \autoref{assump-inf-norm-bound-for-vec-x} does not hold for an arbitrary covariance matrix. 
However, it holds for (but not limited to) the cases of centered continuous Gaussian processes, which suffices our current purpose of preparing quantum states embedding exponentiated Gaussian processes.
\end{remark}
\begin{remark}\label{gaussian-sample-path-2-norm-sqrt-N}
Let $\vec{w}=(G_{t_i})_{i=1}^N$ be a vector representing a sample path of a centered and continuous Gaussian process on arbitrary time discretisation $\mathcal{P}_{[0,1]}^N=(0=t_0<t_1<\dots<t_N=1)$ of $[0,1]$. 
By \autoref{bound-xi-on-norm-x-infty}, there exists a constant $\Xi>0$ such that $\sup_{t\in[0,1]}|G_t|\leq \Xi$ holds with high probability, so we have that
$\norm{\vec{w}}_2=\sqrt{\sum_{i=1}^N |w_i|^2}\leq \Xi \sqrt{N}=O(\sqrt{N})$ with high probability.
\end{remark}

\section{Preparation of a correlated Gaussian vector}\label{preparing-ket-x-section}
\begin{proposition}[{\cite{CGJ19} Lemma 10 of the full version: arXiv:1804.01973.}]\label{prop-for-Hc2}
Let $\widetilde{\varepsilon}>0$ and $\widehat{\kappa}\geq2$. Let $H\in\mathbb{C}^{N\times N}$ be a Hermitian matrix whose eigenvalues lie in the interval $[1/\widehat{\kappa},1]$. Suppose there exists a unitary $U_H$ that is an $(\alpha,a,\delta)$-block-encoding of $H$, where $\delta=o(\widetilde{\varepsilon}/\left(\widehat{\kappa}\log^3(\widehat{\kappa}/\widetilde{\varepsilon})\right))$, and that $U_H$ can be implemented using $T_{U_H}$ elementary gates. Then, for any $c\in[0,1]$, there exists a unitary $U_{H^c}$ that is a $(2,a+O(\log\log(1/\widetilde{\varepsilon})),\widetilde{\varepsilon})$-block-encoding of $H^c$, and $U_{H^c}$ can be implemented using
$O\left(\alpha\widehat{\kappa}(a+T_{U_H})\log^2(\widehat{\kappa}/\widetilde{\varepsilon})\right)$
elementary gates.
\end{proposition}
\begin{definition}[Orthogonal projectors]
A square matrix $\Pi\in\mathbb{C}^{N\times N}$ is called an orthogonal projector if $\Pi^2=\Pi$ and $\Pi=\Pi^{\dagger}$, where $\Pi^{\dagger}$ denotes the conjugate transpose of $\Pi$.
\end{definition}
\begin{definition}[Multi-controlled $X$ gate / $q$-qubit Toffoli gate]
The multi-controlled-$X$ gate, denoted by $\mathrm{MC}X$, or $\mathrm{MC}_{\Pi}X$, when $\Pi\in\mathbb{C}^{2^q\times 2^q}$ is an orthogonal projector for some $q\in\mathbb{N}$, is defined as the unitary
\[
\mathrm{MC}_{\Pi}X \coloneqq X\otimes \Pi + I_2\otimes (I_{2^q}-\Pi),
\]
acting on the Hilbert space $\mathbb{C}^2\otimes\mathbb{C}^{2^q}$. 
This operator applies an $X$ gate to a target qubit if the state of a control register lies in the image of $\Pi$, and acts as the identity otherwise.
When $\Pi$ is of the form $\Pi=(\kb{0})^{\otimes q}$, $\mathrm{MC}_{\Pi}X$ is also known as a $q$-qubit Toffoli gate.
If the control register has a name (e.g.\ register $\mathrm{Q}$) and the control condition is clear from context
(e.g., applying an $X$ gate if the state of register $\mathrm{Q}$ is $\ket{0}^{\otimes q}$), the notation $\mathrm{MC}_{\mathrm{Q}}X$ may be used.
\end{definition}

\begin{lemma}[Implementation complexity of $q$-qubit Toffoli gate]\label{implement-q-qubit-toffoli}
Let $q\in\mathbb{N}$.
A $q$-qubit Toffoli gate can be implemented using $O(q)$ elementary gates and an additional single ancillary qubit.
\end{lemma}
\begin{proof}
As suggested in \cite{GSLW19} under the section of block-encoding quoting \cite{HLZ+17}, the cost to implement $q$-qubit Toffoli gate is $O(q)$ uses of two-qubit gates, with an additional single ancillary qubit.
\end{proof}

\begin{lemma}[Fixed-point amplitude amplification {\cite[Theorem 27 from full paper: arXiv:1806.01838]{GSLW19}}]\label{fixed-point-qaa-from-gslw}
Let $U$ be a unitary, $\Pi$ be an orthogonal projector, and $\ket{\psi_0}$ be a normalised quantum state such that $\Pi U\ket{\psi_0}=\widetilde{a}\ket{\Psi_G}$ holds for a normalised quantum state $\ket{\Psi_G}$ and some $\widetilde{a}\in\mathbb{C}$. Suppose that there exists $a_\ell\in\mathbb{R}$ such that $|\widetilde{a}|\geq a_\ell >0$. Then, for a given $\varepsilon\in(0,1]$, one can construct a unitary circuit $\widetilde{U}_{G}$ such that $\norm{\widetilde{U}_{G}\ket{\psi_0}-\ket{\Psi_G}}_2\leq \varepsilon$, by using a single ancilla qubit and $O\left(\frac{1}{{a_\ell}}\log\left(\frac{1}{\varepsilon}\right)\right)$ calls to $U$, $U^{\dagger},\mathrm{MC}_{\Pi}X,\mathrm{MC}_{\kb{\psi_0}}X$, and $e^{i\phi \sigma_z}$ gates. 
\end{lemma}
\begin{proposition}[Preparation of the normalised desired part of a state obtained from applying a block-encoding via QAA]\label{normalising-block-encoding-prep-state-via-qaa}
Let $\ket{\chi}=\frac{A\ket{b}}{\norm{A\ket{b}}}$ for some matrix $A$ and an $n$-qubit normalised quantum state $\ket{b}$.
Suppose we have access to $U_A$, an $(\alpha,a,0)$-block-encoding of $A$, and a state preparation unitary $U_{\ket{b}}:\ket{0}^{\otimes n}\mapsto \ket{b}$.
If there exists $a_\ell>0$ such that $\norm{A\ket{b}}/\alpha \geq a_\ell$, then for ${\varepsilon_{\chi}} \in (0,1]$ we can implement a unitary $U_{\ket{\widetilde{\varphi}},\varepsilon_\chi}$ that prepares an $(a+n)$-qubit state $\ket{\widetilde{\varphi}}\coloneqq U_{\ket{\widetilde{\varphi}},\varepsilon_\chi}\ket{0}^{\otimes (a+n)}$ satisfying $\norm{\ket{\widetilde{\varphi}}-\ket{0}^{\otimes a}\ket{\chi}}_2\leq {\varepsilon_{\chi}}$, by using $O\left(\frac{1}{{a_\ell}}\log\left(\frac{1}{{\varepsilon_{\chi}}}\right)\right)$ calls to $U_{A} \left(I^{\otimes a}\otimes U_{\ket{b}}\right)$ and its inverse, $O\left(\frac{(a+n)}{{a_\ell}}\log\left(\frac{1}{{\varepsilon_{\chi}}}\right)\right)$ other single and two-qubit gates, and two additional single ancillary qubits. In particular, if it takes $T_{U_A}$ and $T_{U_{\ket{b}}}$ elementary gates to implement $U_A$ and $U_{\ket{b}}$ respectively, then the overall cost to prepare $\ket{\widetilde{\varphi}}$ becomes $O\left(\frac{1}{{a_\ell}}\log\left(\frac{1}{{\varepsilon_{\chi}}}\right)\left(a+n+T_{U_A}+T_{U_{\ket{b}}}\right)\right)$ elementary gates.
\end{proposition}
\begin{proof}
Denote register $\mathrm{Q}$ as the $a$-qubit ancillary register for working with the block-encoding $U_A$, and $\mathrm{I}$ as the system register. By definition of a block-encoding,
\begin{align*}
U_{A} \left(I^{\otimes a}\otimes U_{\ket{b}}\right) \ket{0}_{\mathrm{Q}}\ket{0}_{\mathrm{I}}
&=\frac{1}{\alpha}\ket{0}_{\mathrm{Q}} A\ket{b}_{\mathrm{I}} + \ket{\varphi^{\perp}}_{\mathrm{Q}\mathrm{I}}
\eqqcolon \frac{\norm{A\ket{b}}}{\alpha}\ket{0}_{\mathrm{Q}} \ket{\chi}_{\mathrm{I}} + \norm{\ket{\varphi^{\perp}}}\ket{\Psi_B}_{\mathrm{Q}\mathrm{I}},
\end{align*}
where $\ket{\varphi^{\perp}}_{\mathrm{Q}\mathrm{I}}$ is an unnormalized quantum state such that $\Pi \ket{\varphi^{\perp}}_{\mathrm{Q}\mathrm{I}}=0$ with $\Pi\coloneqq \ket{0}_{\mathrm{Q}} \bra{0}_{\mathrm{Q}} \otimes I^{\otimes n}$, and $\ket{\Psi_B}_{\mathrm{Q}\mathrm{I}}=\ket{\varphi^{\perp}}_{\mathrm{Q}\mathrm{I}}/\norm{\ket{\varphi^{\perp}}}$.
Then, by considering 
$\ket{\psi_0}\coloneqq \ket{0}_{\mathrm{Q}}\ket{0}_{\mathrm{I}}$, 
$U\coloneqq U_{A} \left(I^{\otimes a}\otimes U_{\ket{b}}\right)$,
$\ket{\Psi_G}\coloneqq \ket{0}_{\mathrm{Q}} \ket{\chi}_{\mathrm{I}}$, 
we can utilise \autoref{fixed-point-qaa-from-gslw} to construct $\widetilde{U}_G \equiv U_{\ket{\widetilde{\varphi}},\varepsilon_\chi}$ that maps $\ket{0}_{\mathrm{Q}}\ket{0}_{\mathrm{I}} \mapsto \ket{\widetilde{\varphi}}_{\mathrm{Q}\mathrm{I}}$ with property $\norm{\ket{\widetilde{\varphi}}-\ket{0}^{\otimes a}\ket{\chi}}_2\leq {\varepsilon_{\chi}}$, satisfying the claim. Note that $\Pi \left(U_{A} \left(I^{\otimes a}\otimes U_{\ket{b}}\right)\right) \ket{0}_{\mathrm{Q}}\ket{0}_{\mathrm{I}}=\frac{\norm{A\ket{b}}}{\alpha}\ket{\Psi_G}$, so the parameter $\widetilde{a}$ in \autoref{fixed-point-qaa-from-gslw} becomes $\norm{A\ket{b}}/\alpha$, and hence the number of calls required to create QAA circuit depends on the lower bound $a_\ell$ of this quantity. By definitions of $\Pi$ and $\ket{\psi_0}$, $\mathrm{MC}_{\Pi}X$ and $\mathrm{MC}_{\kb{\psi_0}}X$ become $a$-qubit and $(a+n)$-qubit multi-controlled-NOT (Toffoli) gates respectively. By \autoref{implement-q-qubit-toffoli}, the cost to implement $a$-qubit and $(a+n)$-qubit Toffoli gates are $O(a)$ and $O(a+n)$ elementary gates respectively, with an additional single ancillary qubit each\footnote{However, as suggested in \cite{HLZ+17}, which is the version of the $q$-qubit Toffoli gate from \autoref{implement-q-qubit-toffoli}, the ancillary qubit is recyclable and hence only one is required.}. 
Combining this with the one single ancillary qubit requirement from \autoref{fixed-point-qaa-from-gslw}, we need a total of two single ancillary qubits.
As we need $O\left(\frac{1}{{a_\ell}}\log\left(\frac{1}{{\varepsilon_{\chi}}}\right)\right)$ copies of each of these gates, the total cost for implementing $\widetilde{U}_G$ follows straightforwardly by addition. 
In particular, if $U_A$ and $U_{\ket{b}}$ take $T_{U_A}$ and $T_{U_{\ket{b}}}$ elementary gates to implement, then the implementation of $U_{A} \left(I^{\otimes a}\otimes U_{\ket{b}}\right)$ costs $T_{U_A}+T_{U_{\ket{b}}}$ elementary gates, and hence the claim follows.
\end{proof}
Suppose we want to prepare a normalised version of a vector $\vec{v}$, i.e.\ $\widehat{v}=\vec{v}/\norm{\vec{v}}$, but in practice we can only prepare a normalised version of a vector $\vec{\omega}$, i.e.\ $\widehat{\omega}=\vec{\omega}/\norm{\vec{\omega}}$, where $\vec{\omega}$ approximates $\vec{v}$ in the $\ell_2$-norm. The following lemma quantifies how an additive approximation error in the unnormalised vectors translates into an approximation error between their normalised versions, measured with respect to the intended target $\widehat{v}$.
\begin{lemma}[Approximation error relative to a target vector] \label{lem-for-unit-vector-approx2}
Let $\vec{v},\vec{\omega}\in\mathbb{C}^N\backslash\{0\}$ satisfy $\norm{\vec{v}-\vec{\omega}}\leq \varepsilon$ for some $\varepsilon>0$. Denote $\widehat{v}=\vec{v}/\norm{\vec{v}}$ (the target unit vector) and $\widehat{\omega}=\vec{\omega}/\norm{\vec{\omega}}$ (the approximating unit vector). Then,
\[
\norm{\widehat{v}-\widehat{\omega}}\equiv\norm{\frac{\vec{v}}{\norm{\vec{v}}}-\frac{\vec{\omega}}{\norm{\vec{\omega}}}}\leq
\frac{2\varepsilon}{\norm{\vec{v}}}.
\]
\end{lemma}
\begin{proof}
Observe that
\[
\widehat{v}-\widehat{\omega}
= \frac{\vec{v}-\vec{\omega}}{\norm{\vec{v}}}
+ \vec{\omega} \left(\frac{1}{\norm{\vec{v}}}-\frac{1}{\norm{\vec{\omega}}}\right).
\]
Taking norms and using the triangle inequality gives
\[
\norm{\widehat{v}-\widehat{\omega}}
\leq \frac{\norm{\vec{v}-\vec{\omega}}}{\norm{\vec{v}}}
+ \norm{\vec{\omega}}
\frac{\bigl|\norm{\vec{\omega}}-\norm{\vec{v}}\bigr|}{\norm{\vec{v}} \norm{\vec{\omega}}}.
\]
By the reverse triangle inequality, $\bigl|\lVert\vec{v}\rVert-\lVert\vec{\omega}\rVert\bigr|\leq \lVert\vec{v}-\vec{\omega}\rVert\leq\varepsilon$. Therefore, both terms on the RHS are bounded by $\varepsilon/\norm{\vec{v}}$, and hence the claim follows.
\end{proof}
\begin{lemma}\label{norm-sqrt-sigma-times-ket-z}
Let $\Sigma\in\mathbb{R}^{N\times N}$ be a real symmetric positive definite matrix, whose eigenvalues lie in $\left[a,b\right]$ for some $0<a<b$. 
Then, for an arbitrary quantum state $\ket{z}=\vec{z}/\norm{\vec{z}}$, where $\vec{z}\in\mathbb{R}^N\backslash\{0\}$, we have
\[
\sqrt{a}
\leq\norm{\Sigma^{1/2}\ket{z}} 
\leq \sqrt{b}. \tagaligneq \label{bounds-for-norm-sigma12-z}
\]
In particular, for $\vec{x}=\Sigma^{1/2}\vec{z}$ and $\lambda_{\min},\lambda_{\max}$ denoting the minimal and the maximal eigenvalues of $\Sigma$, it holds that
\[
\sqrt{\lambda_{\min}}
\leq 
\frac{\norm{\vec{x}}}{\norm{\vec{z}}}
\leq
\sqrt{\lambda_{\max}}. \tagaligneq \label{bounds-for-norm-x-div-by-norm-z}
\]
\end{lemma}
\begin{proof}
The existence and uniqueness of $\Sigma^{1/2}$ (see \autoref{defn-sigma12}) is guaranteed by \autoref{horn-johnson-uniqueness-lemma}. Similarly to \autoref{spectral-form-Sigma12-prop}, we fix one orthonormal eigenbasis $\{\ket{v_k}\}_{k=1}^N$
and we have that any unit vector $\ket{z}$ admits an expression with respect to this orthonormal basis, as $\ket{z}=\sum_i z_i \ket{i}=\sum_k \xi_k \ket{v_k}$, for some $(\xi_k)_{k=1}^N\in\mathbb{R}^N$. This means $\norm{\Sigma^{1/2}\ket{z}}^2=\norm{\sum_k \sqrt{\lambda_k}\xi_k \ket{v_k}}^2=\sum_k \lambda_k |\xi_k|^2$. Since ${}^\forall k$, $a\leq \lambda_k\leq b$ by assumption and $\sum_k|\xi_k|^2 = \norm{\ket{z}}^2=1$, \eqref{bounds-for-norm-sigma12-z} follows. Additionally, \eqref{bounds-for-norm-x-div-by-norm-z} follows from the fact that ${}^\forall k$, $\lambda_k\in\left[\lambda_{\min},\lambda_{\max}\right]\subseteq [a,b]$ and the definition $\ket{z}=\vec{z}/\norm{\vec{z}}$.
\end{proof}
\subsection{Quantum algorithm for preparing $\ket{x}$}
\begin{proposition}\label{main-thm2}
Let $\varepsilon\in\left(0,2\right]$ and
$\Sigma\in\mathbb{R}^{N\times N}$ be a positive definite covariance matrix whose eigenvalues lie in $\left[\frac{1}{\widehat{\kappa}},1\right]$ for some known constant $\widehat{\kappa}\geq 2$.
Suppose we can implement an $(\alpha_{U_\Sigma},a_{U_\Sigma},\varepsilon_{U_\Sigma})$-block-encoding $U_{\Sigma}$ of $\Sigma$ and a state-preparation unitary $U_{\ket{z}}:\ket{0}\mapsto\ket{z}$, where $\ket{z}=\vec{z}/\norm{\vec{z}}$ and $\vec{z}$ is a realisation of the random vector $Z\sim\mathcal{N}(0,I_N)$, by using $T_{U_\Sigma}$ and $T_{U_{\ket{z}}}$ elementary gates respectively.
If the block-encoding error $\varepsilon_{U_\Sigma}$ of $U_\Sigma$ satisfies
\[
\varepsilon_{U_\Sigma}=o\left(\varepsilon\cdot \widehat{\kappa}^{-1.5} \cdot \log^{-3}\left(\frac{\widehat{\kappa}}{\varepsilon}\right)\right),
\tagaligneq \label{condn-for-eps-sig-main-thm}
\]
then we can construct a unitary $U_{\ket{x},\varepsilon}$ that prepares a quantum state $\ket{\widetilde{\varphi}}_{\mathrm{Q}\mathrm{I}}=U_{\ket{x},\varepsilon}\ket{0}_{\mathrm{Q}\mathrm{I}}$ satisfying $\norm{\ket{\widetilde{\varphi}}_{\mathrm{Q}\mathrm{I}}-\ket{0}_{\mathrm{Q}}\ket{x}_{\mathrm{I}}}_2\leq\varepsilon$. 
Here, $\mathrm{I}$ denotes the system register of $n=\ceil{\log_2(N+1)}$ qubits, $\mathrm{Q}$ denotes an ancillary register of $\widetilde{q}=a_{U_\Sigma}+O\left(\log\log(\widehat{\kappa}/{\varepsilon})\right)$ qubits, and $\ket{x}=\vec{x}/\norm{\vec{x}}$, where $\vec{x}$ is a realisation of the random vector $X\sim\mathcal{N}(0,\Sigma)$.
Implementing $U_{\ket{x},\varepsilon}$ takes
\[
T_{U_{\ket{x},\varepsilon}}
=O\left(\sqrt{\widehat{\kappa}}
\cdot
\log\left(\frac{1}{{\varepsilon}}\right)
\cdot
\Biggl(\widetilde{q}+n+T_{U_{\ket{z}}}+\alpha_{U_\Sigma}\,\widehat{\kappa}\,\Big(a_{U_\Sigma}+T_{U_\Sigma}\Big)\log^2\left(\frac{\widehat{\kappa}}{\varepsilon}\right)\Biggr)
\right)
\]
elementary gates. 
In particular, if \autoref{assump-plus-TU-ket-y-log2-N} and \autoref{assump-QROM-like-unitary} hold, \autoref{block-encoding-of-sigma} implies $\varepsilon_{U_{\Sigma}}=0$, and condition \eqref{condn-for-eps-sig-main-thm} is fulfilled automatically. Under such a circumstance, the implementation of $U_{\ket{x},\varepsilon}$ requires
\[
T_{U_{\ket{x},\varepsilon}}
=
O\left(
\norm{\Sigma}_F\,\widehat{\kappa}^{1.5}\,\polylog(N)\log^3\left(\frac{\widehat{\kappa}}{\varepsilon}\right)
\right) \tagaligneq \label{main-thm-cost-eq-2}
\]
elementary gate depth and an ancillary register $\mathrm{Q}$ of $\widetilde{q}=O\left(\log N+\log\log(\widehat{\kappa}/{\varepsilon})\right)$ qubits.
\end{proposition}
\begin{proof}
Let $\widetilde{\varepsilon}\coloneqq \frac{1}{2\sqrt{\widehat{\kappa}}}\frac{\varepsilon}{2}$.
In order to use \autoref{prop-for-Hc2} to implement a block-encoding of $\Sigma^{1/2}$ with parameter $\widetilde{\varepsilon}$, we need that $\varepsilon_{U_\Sigma}=o\left(\widetilde{\varepsilon}/(\widehat{\kappa}\log^{3}\left(\widehat{\kappa}/\widetilde{\varepsilon}\right)\right)$ holds. This, together with other requirements for using \autoref{prop-for-Hc2}, is already satisfied by conditions assumed in the theorem statement.
Hence, via \autoref{prop-for-Hc2} we can construct $U_{\Sigma^{1/2}}$, a $(2,\widetilde{q},\widetilde{\varepsilon})$-block-encoding of $\Sigma^{1/2}$, where $\mathbb{Z}^+\ni\widetilde{q}=a_{U_\Sigma}+O\left(\log\log(1/\widetilde{\varepsilon})\right)$, using 
$$
T_{U_{\Sigma^{1/2}}}=O\left(\alpha_{U_\Sigma}\,\widehat{\kappa}\,\left(a_{U_\Sigma}+T_{U_\Sigma}\right)\log^2\left(\frac{\widehat{\kappa}}{\widetilde{\varepsilon}}\right)\right)
$$ 
elementary gates. By definition of block-encodings, we have
\begin{align*}
U_{\Sigma^{1/2}}\left(I^{\otimes \widetilde{q}}\otimes U_{\ket{z}}\right)\ket{0}_{\mathrm{Q}}\ket{0}_{\mathrm{I}}
&=\frac{1}{2}\ket{0}_{\mathrm{Q}}f(\Sigma)\ket{z}_{\mathrm{I}}+\ket{\widetilde{z}^{\perp}}_{\mathrm{Q}\mathrm{I}}, \tagaligneq \label{state-after-apply-sigma12-main-thm}
\end{align*}
where $\norm{f(\Sigma)-\Sigma^{1/2}}\leq\widetilde{\varepsilon}$ and $\ket{\widetilde{z}^{\perp}}_{\mathrm{Q}\mathrm{I}}$ is an unnormalised state whose $\mathrm{Q}$ register is orthogonal to the state $\ket{0}_{\mathrm{Q}}$.
From this, it holds that
\begin{align*}
\frac12\norm{f(\Sigma)\ket{z}} 
\,\geq\, \frac12\norm{\Sigma^{1/2}\ket{z}}-\frac{\widetilde{\varepsilon}}{2} 
\,\stackrel{(\dagger)}{\geq}\,
\frac{1}{2\sqrt{\widehat{\kappa}}}-\frac{\widetilde{\varepsilon}}{2}
\,\stackrel{(\ddagger)}{\geq}\,
\frac{1}{4\sqrt{\widehat{\kappa}}}\eqqcolon a_\ell, \tagaligneq \label{find-psucc-lower-bound-main-thm}
\end{align*}
where $(\dagger)$ follows from \autoref{norm-sqrt-sigma-times-ket-z} and $(\ddagger)$ holds since $\widetilde{\varepsilon}\in\left(0,\frac{1}{2\sqrt{\widehat{\kappa}}}\right]$ by definition.
Then, by letting $\varepsilon_\chi\coloneqq \frac{\varepsilon}{2}$ and denoting $\ket{\chi}\coloneqq \frac{f(\Sigma)\ket{z}}{\norm{f(\Sigma)\ket{z}}}$,
we can use \autoref{normalising-block-encoding-prep-state-via-qaa} to construct a unitary $U_{\ket{\widetilde{\varphi}},\varepsilon_\chi}\eqqcolon U_{\ket{x},\varepsilon}$ that maps $\ket{0}_{\mathrm{Q}\mathrm{I}}\mapsto\ket{\widetilde{\varphi}}_{\mathrm{Q}\mathrm{I}}$ satisfying
\[
\Big\lVert 
\ket{\widetilde{\varphi}}_{\mathrm{Q}\mathrm{I}}-\ket{0}_{\mathrm{Q}}\ket{\chi}_{\mathrm{I}}
\Big\rVert_2
\leq
\varepsilon_\chi = \frac{\varepsilon}{2}, \tagaligneq \label{main-thm-apprx-err-varphi-to-chi-1}
\]
by using 
\begin{align*}
&O\left(\frac{1}{{a_\ell}}\log\left(\frac{1}{{\varepsilon_{\chi}}}\right)\left(\widetilde{q}+n+T_{U_{\Sigma^{1/2}}}+T_{U_{\ket{z}}}\right)\right)\\
&=
O\left(\sqrt{\widehat{\kappa}}\log\left(\frac{1}{{\varepsilon}}\right)\left(\widetilde{q}+n+\alpha_{U_\Sigma}\,\widehat{\kappa}\,\Big(a_{U_\Sigma}+T_{U_\Sigma}\Big)\log^2\left(\frac{\widehat{\kappa}}{\varepsilon}\right)+T_{U_{\ket{z}}}\right)\right) \tagaligneq \label{main-thm-cost-eq-1}
\end{align*}
elementary gates.
Finally, since
\[
\norm{f(\Sigma)\ket{z}_{\mathrm{I}}-\Sigma^{1/2}\ket{z}_{\mathrm{I}}}\leq\norm{f(\Sigma)-\Sigma^{1/2}}\leq\widetilde{\varepsilon},
\] 
\autoref{lem-for-unit-vector-approx2} suggests that
\[
\norm{\ket{\chi}-\ket{x}}=\norm{\frac{f(\Sigma)\ket{z}_{\mathrm{I}}}{\norm{f(\Sigma)\ket{z}_{\mathrm{I}}}}-\frac{\Sigma^{1/2}\ket{z}_{\mathrm{I}}}{\norm{\Sigma^{1/2}\ket{z}_{\mathrm{I}}}}}
\leq 
\frac{2\widetilde{\varepsilon}}{\norm{\Sigma^{1/2}\ket{z}_{\mathrm{I}}}}= \frac{(\varepsilon/2)}{\sqrt{\widehat{\kappa}}\norm{\Sigma^{1/2}\ket{z}_{\mathrm{I}}}}\leq\frac{\varepsilon}{2},
\]
where the last inequality follows from $\sqrt{\widehat{\kappa}}\norm{\Sigma^{1/2}\ket{z}_{\mathrm{I}}}\geq 1$; cf.\ \autoref{norm-sqrt-sigma-times-ket-z}.
Combining this with \eqref{main-thm-apprx-err-varphi-to-chi-1}, we now know that the prepared quantum state $\ket{\widetilde{\varphi}}$ satisfies:
\begin{align*}
\Big\lVert
\ket{\widetilde{\varphi}}_{\mathrm{Q}\mathrm{I}}-\ket{0}_{\mathrm{Q}}\ket{x}_{\mathrm{I}}
\Big\rVert
&\leq
\Big\lVert
\ket{\widetilde{\varphi}}-\ket{0}_{\mathrm{Q}}\ket{\chi}_{\mathrm{I}}
\Big\rVert
+
\Big\lVert
\ket{0}_{\mathrm{Q}}\ket{\chi}_{\mathrm{I}}-\ket{0}_{\mathrm{Q}}\ket{x}_{\mathrm{I}}
\Big\rVert
\leq \frac{\varepsilon}{2}+\frac{\varepsilon}{2}=\varepsilon.
\end{align*}
To show \eqref{main-thm-cost-eq-2}, plugging $\alpha_{U_\Sigma}=\lVert \Sigma \rVert_F$, $a_{U_\Sigma}=\ceil{\log_2(N+1)}$, $T_{U_\Sigma}=O(\polylog N)$ (see \autoref{block-encoding-of-sigma}), and $T_{U_{\ket{z}}}=O\left(\log^2 N\right)$ (see \autoref{assump-plus-TU-ket-y-log2-N}) into \eqref{main-thm-cost-eq-1} suffices\footnote{Note the change in the unit from ``gates'' (count) to ``gate \textit{depth}'' in \eqref{main-thm-cost-eq-2}, cf. \autoref{convention-gate-count-vs-depth}.}.
\end{proof}
\begin{theorem}\label{main-thm-for-gen-mat}
Let $\varepsilon\in\left(0,2\right]$ and $\Sigma\in\mathbb{R}^{N \times N}$ be a positive definite covariance matrix whose maximal eigenvalue and condition number are denoted by $\lambda_{\max}$ and $\kappa$ respectively.
Suppose \autoref{assmp-estimate-lamb-max-kappa} holds, so that we have the knowledge of the estimates $\widetilde{\lambda}_{\max}$ and $\widetilde{\kappa}$, and the multiplicative-error bounds $L$ and $K$, satisfying \eqref{main-thm-col-assump2}.
Suppose we can implement an $(\alpha_{U_\Sigma},a_{U_\Sigma},\varepsilon_{U_\Sigma})$-block-encoding $U_{\Sigma}$ of $\Sigma$ and a state-preparation unitary $U_{\ket{z}}:\ket{0}\mapsto\ket{z}$, where $\ket{z}=\vec{z}/\norm{\vec{z}}$ and $\vec{z}$ is a realisation of the random vector $Z\sim\mathcal{N}(0,I_N)$, by using $T_{U_\Sigma}$ and $T_{U_{\ket{z}}}$ elementary gates respectively.
If the block-encoding error $\varepsilon_{U_\Sigma}$ of $U_\Sigma$ satisfies
\[
\frac{\varepsilon_{U_\Sigma}}{\widetilde{\lambda}_{\max}}=o\left(\varepsilon\cdot (L\widetilde{\kappa})^{-1.5} \cdot \log^{-3}\left(\frac{L\widetilde{\kappa}}{\varepsilon}\right)\right), \tagaligneq \label{eps-U-sig-cond-main-thm-for-gen-mat}
\]
then we can construct a unitary $U_{\ket{x},\varepsilon}$ that prepares a quantum state $\ket{\widetilde{\varphi}}_{\mathrm{Q}\mathrm{I}}=U_{\ket{x},\varepsilon}\ket{0}_{\mathrm{Q}\mathrm{I}}$ satisfying $\norm{\ket{\widetilde{\varphi}}_{\mathrm{Q}\mathrm{I}}-\ket{0}_{\mathrm{Q}}\ket{x}_{\mathrm{I}}}_2\leq\varepsilon$. 
Here, 
$\mathrm{I}$ denotes the system register of $n=\ceil{\log_2(N+1)}$ qubits, 
$\mathrm{Q}$ denotes an ancillary register of $\widetilde{q}=a_{U_\Sigma}+O\left(\log\log(L\widetilde{\kappa}/{\varepsilon})\right)$ qubits, 
and $\ket{x}=\vec{x}/\norm{\vec{x}}$, where $\vec{x}$ is a realisation of the random vector $X\sim\mathcal{N}(0,\Sigma)$. 
Implementing $U_{\ket{x},\varepsilon}$ takes
\[
T_{U_{\ket{x},\varepsilon}}
=
O\left(\sqrt{L\widetilde{\kappa}}
\cdot
\log\left(\frac{1}{{\varepsilon}}\right)
\cdot
\Biggl(
\widetilde{q}
+
n+T_{U_{\ket{z}}}+\frac{\alpha_{U_\Sigma}}{\widetilde{\lambda}_{\max}}\,L\widetilde{\kappa}\,\Big(a_{U_\Sigma}+T_{U_{\Sigma}}\Big)\log^2\left(\frac{L\widetilde{\kappa}}{\varepsilon}\right)\Biggr)
\right)
\]
elementary gates.
In particular, if \autoref{assump-LK-plus} holds, implementing $U_{\ket{x},\varepsilon}$ takes
\[
T_{U_{\ket{x},\varepsilon}}
=
O\left(
\sqrt{\kappa}
\,
\log\left(\frac{1}{{\varepsilon}}\right)
\,
\polylog(N) \,
\Biggl(
\widetilde{q}
+
n+T_{U_{\ket{z}}}+\frac{\alpha_{U_\Sigma}}{\lambda_{\max}}\,\kappa\,\Big(a_{U_\Sigma}+T_{U_{\Sigma}}\Big)\log^2\left(\frac{\kappa}{\varepsilon}\right)\Biggr)
\right)
\]
elementary gates, where $\widetilde{q}=a_{U_\Sigma}+O\left(\log\log\log(N)+\log\log({\kappa}/{\varepsilon})\right)$ qubits.
\end{theorem}
\begin{proof}
Let us denote $\widetilde{\Sigma}\coloneqq \Sigma/\widetilde{\lambda}_{\max}$. 
Note that, as we have access to the block-encoding $U_{\Sigma}$ of $\Sigma$, we can use this unitary as the block-encoding for $\widetilde{\Sigma}$ as well, because
\begin{align*}
&\norm{\Sigma-\alpha_{U_\Sigma}\left(\bra{0}^{\otimes a_{U_\Sigma}}\otimes I_N\right)U_{\Sigma}\left(\ket{0}^{\otimes a_{U_\Sigma}}\otimes I_N\right)}
\leq
\varepsilon_{U_\Sigma}\\
\iff&
\norm{\widetilde{\Sigma}-\frac{\alpha_{U_\Sigma}}{\widetilde{\lambda}_{\max}}\left(\bra{0}^{\otimes a_{U_\Sigma}}\otimes I_N\right)U_{\Sigma}\left(\ket{0}^{\otimes a_{U_\Sigma}}\otimes I_N\right)}
\leq
\frac{\varepsilon_{U_\Sigma}}{\widetilde{\lambda}_{\max}},
\end{align*}
which means that $U_{\Sigma}$ is an $\left(\frac{\alpha_{U_\Sigma}}{\widetilde{\lambda}_{\max}},a_{U_\Sigma},\frac{\varepsilon_{U_\Sigma}}{\widetilde{\lambda}_{\max}}\right)$-block-encoding of $\widetilde{\Sigma}$.
Now, by definition of $\widetilde{\lambda}_{\max}$, we have that all eigenvalues of $\widetilde{\Sigma}$ lie in $\left[\frac{1}{L\kappa},1\right]\subseteq\left[\frac{1}{L\widetilde{\kappa}},1\right]$.
Therefore, we can apply \autoref{main-thm2} to $\widetilde{\Sigma}$ with choice $\widehat{\kappa}\equiv \max\{L\widetilde{\kappa},2\}$ as long as $\varepsilon_{U_\Sigma}/\widetilde{\lambda}_{\max}=o\left(\varepsilon\cdot \widehat{\kappa}^{-1.5} \cdot \log^{-3}\left(\frac{\widehat{\kappa}}{\varepsilon}\right)\right)$ holds, but this is already assumed in the claim (see \eqref{eps-U-sig-cond-main-thm-for-gen-mat}).
Doing so allows us to prepare the quantum state $\ket{\widetilde{\varphi}}$ satisfying
\[
\norm{\ket{\widetilde{\varphi}}_{\mathrm{Q}\mathrm{I}}-\ket{0}_{\mathrm{Q}}\frac{\widetilde{\Sigma}^{1/2}\ket{z}_{\mathrm{I}}}{\norm{\widetilde{\Sigma}^{1/2}\ket{z}_{\mathrm{I}}}}}
\stackrel{(\dagger)}{=}
\Big\lVert
\ket{\widetilde{\varphi}}_{\mathrm{Q}\mathrm{I}}-\ket{0}_{\mathrm{Q}}\ket{x}_{\mathrm{I}}
\Big\rVert
\leq
\varepsilon,
\]
where 
$\mathrm{Q}$ is an ancillary register of $\widetilde{q}=a_{U_\Sigma}+O\left(\log\log(L\widetilde{\kappa}/{\varepsilon})\right)$ qubits, 
and
$(\dagger)$ holds because $\frac{\widetilde{\Sigma}^{1/2}\ket{z}}{\norm{\widetilde{\Sigma}^{1/2}\ket{z}}}
=
\frac{{\Sigma^{1/2}\ket{z}}/{{\sqrt{\widetilde{\lambda}_{\max}}}}}{\norm{{\Sigma^{1/2}\ket{z}}/{{\sqrt{\widetilde{\lambda}_{\max}}}}}}
=\frac{\Sigma^{1/2}\ket{z}}{\norm{\Sigma^{1/2}\ket{z}}}=\ket{x}$.
According to \autoref{main-thm2}, the total gate complexity to prepare $\ket{\widetilde{\varphi}}$ is
\[
T_{U_{\ket{x},\varepsilon}}
=
O\left(\sqrt{L\widetilde{\kappa}}
\cdot
\log\left(\frac{1}{{\varepsilon}}\right)
\cdot
\Biggl(
\widetilde{q}
+
n+T_{U_{\ket{z}}}+\frac{\alpha_{U_\Sigma}}{\widetilde{\lambda}_{\max}}\,L\widetilde{\kappa}\,\Big(a_{U_\Sigma}+T_{U_{\Sigma}}\Big)\log^2\left(\frac{L\widetilde{\kappa}}{\varepsilon}\right)\Biggr)
\right).
\]
In particular, if \autoref{assump-LK-plus} holds, we have $L,K=\widetilde{\Theta}(1)$, which implies $\widetilde{\lambda}_{\max}=\widetilde{\Theta}(\lambda_{\max})$ and $\widetilde{\kappa}=\widetilde{\Theta}(\kappa)$. In such a case,
\[
T_{U_{\ket{x},\varepsilon}}
=
O\left(
\sqrt{\kappa}
\,
\log\left(\frac{1}{{\varepsilon}}\right)
\,
\polylog(N) \,
\Biggl(
\widetilde{q}
+
n+T_{U_{\ket{z}}}+\frac{\alpha_{U_\Sigma}}{\lambda_{\max}}\,\kappa\,\Big(a_{U_\Sigma}+T_{U_{\Sigma}}\Big)\log^2\left(\frac{\kappa}{\varepsilon}\right)\Biggr)
\right),
\]
where $\widetilde{q}=a_{U_\Sigma}+O\left(\log\log\log(N)+\log\log({\kappa}/{\varepsilon})\right)$ qubits.
\end{proof}
\begin{corollary}\label{main-thm-gen-mat-concrete-cost}
Under the setting of \autoref{main-thm-for-gen-mat}, if the assumptions about data loaders in polylogarithmic time (\autoref{assump-plus-TU-ket-y-log2-N} and \autoref{assump-QROM-like-unitary}), and the assumption about parameter estimation (\autoref{assump-LK-plus}) hold, then we can construct a unitary $U_{\ket{x},\varepsilon}$ that prepares a quantum state $\ket{\widetilde{\varphi}}_{\mathrm{Q}\mathrm{I}}=U_{\ket{x},\varepsilon}\ket{0}_{\mathrm{Q}\mathrm{I}}$ satisfying $\norm{\ket{\widetilde{\varphi}}_{\mathrm{Q}\mathrm{I}}-\ket{0}_{\mathrm{Q}}\ket{x}_{\mathrm{I}}}_2\leq\varepsilon$, where $\ket{x}=\vec{x}/\norm{\vec{x}}$, $\vec{x}$ is a realisation of the random vector $X\sim\mathcal{N}(0,\Sigma)$, and $\mathrm{Q}$ is an ancillary register of $\widetilde{q}=O\left(\log(N)+\log\log({\kappa}/{\varepsilon})\right)$ qubits. The implementation of $U_{\ket{x},\varepsilon}$ requires
\[
T_{U_{\ket{x},\varepsilon}}
=
{O}\left(
\frac{\lVert \Sigma \rVert_F}{\lambda_{\max}}\,\kappa^{1.5}\,\polylog(N)
\log^3\left(\frac{\kappa}{\varepsilon}\right)
\right),
\]
elementary gate depth.
\end{corollary}
\begin{proof}
Similarly to the proof of \eqref{main-thm-cost-eq-2}, 
\autoref{assump-QROM-like-unitary} implies \autoref{block-encoding-of-sigma}, which further implies $\varepsilon_{U_{\Sigma}}=0$, and hence condition \eqref{eps-U-sig-cond-main-thm-for-gen-mat} is satisfied automatically. So, the claimed complexity follows by applying \autoref{main-thm-for-gen-mat}
with parameters $\alpha_{U_\Sigma}=\lVert \Sigma \rVert_F$, $a_{U_\Sigma}=\ceil{\log_2(N+1)}$, $T_{U_\Sigma}=O(\polylog N)$ (see \autoref{block-encoding-of-sigma}), and $T_{U_{\ket{z}}}=O\left(\log^2 N\right)$ (see \autoref{assump-plus-TU-ket-y-log2-N}).
\end{proof}
\subsection{Cumulative sum of a correlated Gaussian vector}
When we want to sample a sample path of a Gaussian process $G$, 
sometimes preparing the vector of increments $(\Delta G_{t_i})_{i=0}^N$, where $\Delta G_{t_i}\equiv G_{t_i}-G_{t_{i-1}}$ with the convention $G_{t_{-1}}\equiv 0$, 
and then taking the cumulative sum $G_{t_j} = \sum_{i=0}^j \Delta G_{t_i}$ to recover the path values
can result in smaller overall complexity, since the covariance matrix of the increments $\Sigma^{\mathrm{ns}}\equiv (\Sigma^{\mathrm{ns}}_{ij}=\ee{\Delta G_{t_i} \Delta G_{t_j}})_{i,j}$ may be better conditioned (i.e.\ $\norm{\Sigma}_F/\lambda_{\max}$ and $\kappa$ may have smaller dependence on $N$) than the covariance matrix of the path values $\Sigma^{\mathrm{pv}}\equiv (\Sigma^{\mathrm{pv}}_{ij}=\ee{G_{t_i}G_{t_j}})_{i,j}$.
We give a more detailed discussion in \autoref{num-sec-setting-pv-vs-ns-cov-mat}.
In quantum computing, taking a cumulative sum of a vector can be achieved by applying the matrix $\mathcal{L}_N$, defined below, to a quantum state of interest. However, the affected normalising factor should also be taken into account.
\begin{notation}\label{defn-of-LN}
Denote a lower triangular matrix $\mathcal{L}_N \in \mathbb{R}^{N \times N}$ as
\[
\mathcal{L}_N \coloneqq
\begin{pNiceMatrix}
1 & & & \Block{3-2}<\huge>{\mathbf{0}}  \\
1 & 1 & & & \\
1 & 1 & \ddots &  & \\
\vdots & \vdots & \ddots & \ddots & \\
1 & 1 & \ldots & 1 & 1
\end{pNiceMatrix}.
\]
\end{notation}
\begin{lemma}\label{prop-for-LN-times-unit-vec2}
Let $N\in\mathbb{N}$.
Any $\widehat{u}\in\mathbb{R}^N$ with $\norm{\widehat{u}}=1$ satisfies
\[
\frac12 
\leq
\sigma_{\min}^N{(\mathcal{L}_N)}
\leq
\big\lVert\mathcal{L}_N\widehat{u}\big\rVert
\leq 
\sigma_{\max}^N{(\mathcal{L}_N)}
\leq
N,
\]
where $\sigma_{\min}^N{(\mathcal{L}_N)}$, $\sigma_{\max}^N{(\mathcal{L}_N)}$ denote the minimal and maximal singular values of $\mathcal{L}_N$ respectively.
\end{lemma}
\begin{proof}
We give the proof in \autoref{appndx-proof-of-LN-property}.
\end{proof}
The following corollary on the preparation of a quantum state encoding the cumulative sum of a correlated Gaussian vector
mirrors the proofs of \autoref{main-thm2} and \autoref{main-thm-for-gen-mat}. 
A more general statement with complexity expressed via the abstract block-encoding implementation cost $T_{\mathcal{L}_N}$ for $\mathcal{L}_N$ can be given similarly to \autoref{main-thm-for-gen-mat} by the same reasoning.
We therefore do not repeat this derivation here. Instead, we state only the result under \autoref{assump-QROM-like-unitary}, giving an algorithmic complexity expressed only in concrete parameters, which suffices for subsequent discussion regarding explicit comparison with classical complexity.
\begin{corollary}\label{main-thm-coro-concrete-cost-sum-of-inc}
Under the setting of \autoref{main-thm-for-gen-mat}, if the assumptions about data loaders in polylogarithmic time (\autoref{assump-plus-TU-ket-y-log2-N} and \autoref{assump-QROM-like-unitary}), and the assumption about parameter estimation (\autoref{assump-LK-plus}) hold, then we can construct a unitary $U_{\ket{y},\varepsilon}$ that prepares a quantum state $\ket{\widetilde{\varphi}}_{\mathrm{Q}\mathrm{I}}=U_{\ket{y},\varepsilon}\ket{0}_{\mathrm{Q}\mathrm{I}}$ satisfying $\norm{\ket{\widetilde{\varphi}}_{\mathrm{Q}\mathrm{I}}-\ket{0}_{\mathrm{Q}}\ket{y}_{\mathrm{I}}}_2\leq\varepsilon$, where $\ket{y}=\vec{y}/\norm{\vec{y}}$ for $\vec{y}\coloneqq \mathcal{L}_N\vec{x}$, $\vec{x}$ is a realisation of the random vector $X\sim\mathcal{N}(0,\Sigma)$, and $\mathrm{Q}$ is an ancillary register of $\widetilde{q}=O\left(\log(N)+\log\log({\kappa}/{\varepsilon})\right)$ qubits. The implementation of $U_{\ket{x},\varepsilon}$ requires
\[
T_{U_{\ket{x},\varepsilon}}
=
{O}\left(
\frac{\lVert \Sigma \rVert_F}{\lambda_{\max}}\,\kappa^{1.5}
\cdot
N\polylog(N)
\cdot
\log^3\left(\frac{\kappa}{\varepsilon}\right)
\right)
\]
elementary gate depth.
\end{corollary}
\begin{proof}
Under \autoref{assump-QROM-like-unitary}, we can construct the $(\norm{\Sigma}_F,\ceil{\log_2(N+1)},0)$-block-encoding $U_\Sigma$ of $\Sigma$ in $T_{U_\Sigma}=O(\polylog N)$ elementary gate depth by using \autoref{block-encoding-of-sigma}.
Following the proof of \autoref{main-thm-for-gen-mat} and \autoref{main-thm-gen-mat-concrete-cost}, we define $\widetilde{\Sigma}\coloneqq \Sigma/\widetilde{\lambda}_{\max}$ and use $U_\Sigma$ as the $\left(\frac{\alpha_{U_\Sigma}}{\widetilde{\lambda}_{\max}},a_{U_\Sigma},\frac{\varepsilon_{U_\Sigma}}{\widetilde{\lambda}_{\max}}\right)$-block-encoding of $\widetilde{\Sigma}$,
in which $\alpha_{U_\Sigma}=\norm{{\Sigma}}_F$, $a_{U_\Sigma}=\ceil{\log_2 (N+1)}$, and $\varepsilon_{U_\Sigma}=0$. Again, this makes the condition $\varepsilon_{U_\Sigma}/\widetilde{\lambda}_{\max}=o(\widetilde{\varepsilon}/\left(\widehat{\kappa}\log^3(\widehat{\kappa}/\widetilde{\varepsilon})\right))$ required to use \autoref{prop-for-Hc2} (similarly to \eqref{eps-U-sig-cond-main-thm-for-gen-mat}) hold trivially for any $\widetilde{\varepsilon}>0$ and $\widetilde{\kappa}\geq 2$. 
So,
by setting $\widehat{\kappa}\coloneqq \max\{L\widetilde{\kappa},2\}$ and $\widetilde{\varepsilon}\coloneqq \frac{1}{2N(2\sqrt{\widehat{\kappa}})}\frac{\varepsilon}{2}$ in \autoref{prop-for-Hc2},
we can create 
a $(2,a_{U_\Sigma}+O(\log\log(1/\widetilde{\varepsilon})),\widetilde{\varepsilon})$-block-encoding $U_{\widetilde{\Sigma}^{1/2}}$ of $\widetilde{\Sigma}^{1/2}$.
Implementing $U_{\widetilde{\Sigma}^{1/2}}$ takes
\[
T_{U_{\widetilde{\Sigma}^{1/2}}}
=O\left(\frac{\alpha_{U_\Sigma}}{\widetilde{\lambda}_{\max}}\,\widehat{\kappa}\,\left(a_{U_\Sigma}+T_{U_\Sigma}\right)\log^2\left(\frac{\widehat{\kappa}}{\widetilde{\varepsilon}}\right)\right)
=O\left(\frac{\norm{{\Sigma}}_F}{\widetilde{\lambda}_{\max}}\,\widehat{\kappa}\,\polylog(N)\log^2\left(\frac{\widehat{\kappa}}{\widetilde{\varepsilon}}\right)\right)
\]
elementary gate depth.
Define $\mathrm{Q}_1$ to be an ancillary register of $\widehat{q}=a_{U_\Sigma}+O\left(\log\log(1/\widetilde{\varepsilon})\right)=O\left(\log(N)+\log\log({N\kappa}/{\varepsilon})\right)=O\left(\log(N)+\log\log({\kappa}/{\varepsilon})\right)$ qubits and $\mathrm{Q}_2$ to be an ancillary register of $\widehat{a}\coloneqq\ceil{\log_2 (N+1)}$ qubits respectively.
Denote the overall ancillary register $\mathrm{Q}$ to include $\mathrm{Q}_1$ and $\mathrm{Q}_2$.
That is, $\mathrm{Q}$ is an ancillary register of 
$\widetilde{q}=\widehat{q}+\widehat{a}=O\left(\log(N)+\log\log({\kappa}/{\varepsilon})\right)$ qubits,
and
$\ket{0}_{\mathrm{Q}}\equiv \ket{0}_{\mathrm{Q}_2}\ket{0}_{\mathrm{Q}_1}$.
Similarly to \eqref{state-after-apply-sigma12-main-thm}, applying $U_{\widetilde{\Sigma}^{1/2}}$ gives
\begin{align*}
\left(I^{\otimes \widehat{a}}\otimes U_{\widetilde{\Sigma}^{1/2}}\right)\left(I^{\otimes \widetilde{q}}\otimes U_{\ket{z}}\right)\ket{0}_{\mathrm{Q}}\ket{0}_{\mathrm{I}}
&=
\ket{0}_{\mathrm{Q}_2}
\left(
\frac{1}{2}\ket{0}_{\mathrm{Q}_1}f(\widetilde{\Sigma})\ket{z}_{\mathrm{I}}
+\sum_{i=1}^{2^{\widehat{q}}-1}\ket{i}_{\mathrm{Q}_1}\ket{\widetilde{z}^{\perp}_i}_{\mathrm{I}}
\right)\\
&\eqqcolon \ket{0}_{\mathrm{Q}_2}\ket{\psi}_{\mathrm{Q}_1\mathrm{I}}, \tagaligneq \label{state-after-apply-sigma12-corollary-sum-of-inc}
\end{align*}
where $\big\lVert f(\widetilde{\Sigma})-\widetilde{\Sigma}^{1/2}\big\rVert\leq\widetilde{\varepsilon}$
and
$\{\ket{\widetilde{z}^{\perp}_i}\}_{i}$ are some unnormalised states. 
By using \autoref{block-encoding-of-sigma},
we can construct an $(\lVert\mathcal{L}_N\rVert_F,\ceil{\log_2 (N+1)},0)$-block-encoding $U_{\mathcal{L}_N}$ of $\mathcal{L}_N$ over the $\widehat{a}$-qubit ancillary register $\mathrm{Q}_2$, 
such that $U_{\mathcal{L}_N}$ applies $\mathcal{L}_N/\lVert\mathcal{L}_N\rVert_F$ to the state in register $\mathrm{I}$, conditioned on the state in register $\mathrm{Q}_2$ being $\ket{0}_{\mathrm{Q}_2}$. For ease of expression, we denote $\widetilde{U}_{\mathcal{L}_N}$ to be the unitary acting on all registers $\mathrm{Q}\mathrm{I}$, where it applies the identity operator $I^{\otimes \widehat{q}}$ to the register $\mathrm{Q}_1$ and applies $U_{\mathcal{L}_N}$ to the registers $\mathrm{Q}_2$ and $\mathrm{I}$ as described.
By applying $\widetilde{U}_{\mathcal{L}_N}$ to \eqref{state-after-apply-sigma12-corollary-sum-of-inc}, we obtain
\begin{align*}
\widetilde{U}_{\mathcal{L}_N}\ket{0}_{\mathrm{Q}_2}\ket{\psi}_{\mathrm{Q}_1\mathrm{I}}
&=
\ket{0}_{\mathrm{Q}_2}
\left(
\frac{1}{2\lVert\mathcal{L}_N\rVert_F}\ket{0}_{\mathrm{Q}_1}\mathcal{L}_Nf(\widetilde{\Sigma})\ket{z}_{\mathrm{I}}
+\sum_{i=1}^{2^{\widehat{q}}-1}\ket{i}_{\mathrm{Q}_1}\frac{\mathcal{L}_N}{\lVert\mathcal{L}_N\rVert_F}\ket{\widetilde{z}^{\perp}_i}_{\mathrm{I}}
\right)
+\ket{\widetilde{\psi}^\perp}_{\mathrm{Q}\mathrm{I}}\\
&\equiv
\frac{1}{2\lVert\mathcal{L}_N\rVert_F}\ket{0}_{\mathrm{Q}}\mathcal{L}_Nf(\widetilde{\Sigma})\ket{z}_{\mathrm{I}}
+
\ket{0}_{\mathrm{Q}_2}
\ket{\doublewidetilde{z}{}^{\perp}_i}_{\mathrm{Q}_1\mathrm{I}}
+\ket{\widetilde{\psi}^\perp}_{\mathrm{Q}\mathrm{I}}, \tagaligneq \label{state-after-sigma12-and-LN}
\end{align*}
where $\ket{\doublewidetilde{z}{}^{\perp}}_{\mathrm{Q}_1\mathrm{I}}$ 
and
$\ket{\widetilde{\psi}^\perp}_{\mathrm{Q}\mathrm{I}}$
are unnormalised states such that $\left(\ket{0}\bra{0}_{\mathrm{Q}_1} \otimes I^{\otimes n}\right)\ket{\doublewidetilde{z}{}^{\perp}}_{\mathrm{Q}_1\mathrm{I}}=0$ and $\left(\ket{0}\bra{0}_{\mathrm{Q}_2} \otimes I^{\otimes (\widehat{q}+n)}\right)\ket{\widetilde{\psi}^\perp}_{\mathrm{Q}\mathrm{I}}=0$, respectively.
This means that $\left(\ket{0}\bra{0}_{\mathrm{Q}}\otimes I^{\otimes n}\right)\ket{0}_{\mathrm{Q}_2}
\ket{\doublewidetilde{z}{}^{\perp}_i}_{\mathrm{Q}_1\mathrm{I}}=0$ and $\left(\ket{0}\bra{0}_{\mathrm{Q}}\otimes I^{\otimes n}\right)\ket{\widetilde{\psi}^\perp}_{\mathrm{Q}\mathrm{I}}=0$.
Similarly to \autoref{main-thm2}, we want to use fixed-point QAA procedure (\autoref{normalising-block-encoding-prep-state-via-qaa}) to construct a unitary $U_{\ket{\widetilde{\varphi}},\varepsilon_\chi}\eqqcolon U_{\ket{y},\varepsilon}$ that prepares a quantum state $U_{\ket{\widetilde{\varphi}},\varepsilon_\chi}\ket{0}_{\mathrm{Q}\mathrm{I}}\eqqcolon\ket{\widetilde{\varphi}}_{\mathrm{Q}\mathrm{I}}$ satisfying
\[
\Big\lVert 
\ket{\widetilde{\varphi}}_{\mathrm{Q}\mathrm{I}}-\ket{0}_{\mathrm{Q}}\ket{\chi}_{\mathrm{I}}
\Big\rVert_2
\leq
\varepsilon_\chi = \frac{\varepsilon}{2}, 
\tagaligneq \label{diff-varphi-chi-coro-sum-of-inc}
\]
where $\varepsilon_\chi\coloneqq \frac{\varepsilon}{2}$ and $\ket{\chi}\coloneqq \frac{\mathcal{L}_Nf(\widetilde{\Sigma})\ket{z}}{\lVert\mathcal{L}_Nf(\widetilde{\Sigma})\ket{z}\rVert}$. To do this, we need to lower bound the amplitude of the state $\ket{0}_{\mathrm{Q}}$ in the RHS of \eqref{state-after-sigma12-and-LN}, which is
\begin{align*}
\frac{1}{\lVert\mathcal{L}_N\rVert_F}
\cdot
\left\lVert
\mathcal{L}_N\frac{f(\widetilde{\Sigma})\ket{z}}{\left\lVert f(\widetilde{\Sigma})\ket{z} \right\rVert}
\right\rVert
\cdot
\frac12 \left\lVert f(\widetilde{\Sigma})\ket{z} \right\rVert
\geq 
\frac{1}{\sqrt{N}\sqrt{N+1}}\cdot\frac{1}{4\sqrt{2}\sqrt{\widehat{\kappa}}}
\eqqcolon a_\ell. 
\end{align*}
Here, in the left-hand side, 
$\lVert\mathcal{L}_N\rVert_F=\sqrt{N(N+1)/2}$, 
and the second term is lower bounded by $1/2$ (see \autoref{prop-for-LN-times-unit-vec2}). For the third term, $\varepsilon\in(0,2]$
implies 
$\widetilde{\varepsilon}\in\left(0,\frac{1}{2N(2\sqrt{\widehat{\kappa}})}\right]\subset\left(0,\frac{1}{2\sqrt{\widehat{\kappa}}}\right]$, 
so the same calculation as \eqref{find-psucc-lower-bound-main-thm} tells us that the third term is bounded below by $1/(4\sqrt{\widehat{\kappa}})$.
Then, according to \autoref{normalising-block-encoding-prep-state-via-qaa}, we can implement $U_{\ket{\widetilde{\varphi}},\varepsilon_\chi} \equiv U_{\ket{y},\varepsilon}$ by using $\widetilde{U}_P\coloneqq \widetilde{U}_{\mathcal{L}_N}\left(I^{\otimes \widehat{a}}\otimes U_{\widetilde{\Sigma}^{1/2}}\right) \left(I^{\otimes \widetilde{q}}\otimes U_{\ket{z}}\right)$ and additional elementary gates, which in total takes
\begin{align*}
T_{U_{\ket{y},\varepsilon}}
\equiv
T_{U_{\ket{\widetilde{\varphi}},\varepsilon_\chi}}
&=
O\left(\frac{1}{{a_\ell}}\log\left(\frac{1}{{\varepsilon_{\chi}}}\right)\left(\widetilde{q}+n+T_{\widetilde{U}_P}\right)\right)
\end{align*}
elementary gates.
Here, 
$\widetilde{U}_P$'s gate depth is 
\[
T_{\widetilde{U}_P}\equiv T_{\widetilde{U}_{\mathcal{L}_N}}+T_{U_{\widetilde{\Sigma}^{1/2}}}+T_{U_{\ket{z}}}
=O\left(\frac{\norm{{\Sigma}}_F}{\widetilde{\lambda}_{\max}}\,\widehat{\kappa}\,\polylog(N)\log^2\left(\frac{\widehat{\kappa}}{\widetilde{\varepsilon}}\right)\right),
\]
where we have used the fact that $T_{\widetilde{U}_{\mathcal{L}_N}}= T_{U_{\mathcal{L}_N}}=O(\polylog N)$ and $T_{U_{\ket{z}}}=O(\log^2 N)$ under \autoref{block-encoding-of-sigma} and \autoref{assump-plus-TU-ket-y-log2-N} respectively.
Plugging in all the parameters including the definitions of $\widetilde{\varepsilon}$ and $\widehat{\kappa}$, and $L=\widetilde{\Theta}(1)=O(\polylog N)$, $\widetilde{\lambda}_{\max}=\widetilde{\Theta}(\lambda_{\max})$, $\widetilde{\kappa}=\widetilde{\Theta}(\kappa)$ under \autoref{assump-LK-plus}, we then arrive at
\[
T_{U_{\ket{y},\varepsilon}}
=
{O}\left(
\frac{\lVert \Sigma \rVert_F}{\lambda_{\max}}\,\kappa^{1.5}
\cdot
N\polylog(N)
\cdot
\log^3\left(\frac{\kappa}{\varepsilon}\right)
\right).
\]
Note that the $N$ factor in $\widetilde{\varepsilon}=\Theta\big(\varepsilon/(\sqrt{\widehat{\kappa}}\cdot N) \big)$ is absorbed into $\polylog (N)$ term due to logarithmic dependence on $\widetilde{\varepsilon}$.
Finally, we are to prove that the prepared $\ket{\widetilde{\varphi}}_{\mathrm{Q}\mathrm{I}}$ is $\varepsilon$-close to the target state $\ket{0}_{\mathrm{Q}}\ket{y}_{\mathrm{I}}$, where $\ket{y}=\mathcal{L}_N\vec{x}/\norm{\mathcal{L}_N\vec{x}}$. To see this, note that since
\[
\norm{\mathcal{L}_Nf(\widetilde{\Sigma})\ket{z}-\mathcal{L}_N\widetilde{\Sigma}^{1/2}\ket{z}}
\leq
\big\lVert\mathcal{L}_N \big\rVert
\cdot
\left\lVert \left( f(\widetilde{\Sigma})-\widetilde{\Sigma}^{1/2}\right) \ket{z}
\right\rVert
\leq
N
\cdot
\widetilde{\varepsilon}
\] 
holds by
\autoref{prop-for-LN-times-unit-vec2},
we have that
\autoref{lem-for-unit-vector-approx2} suggests 
\begin{align*}
\Big\lVert\ket{\chi}-\ket{y}\Big\rVert
&=\norm{\frac{\mathcal{L}_Nf(\widetilde{\Sigma})\ket{z}}{\left\lVert\mathcal{L}_Nf(\widetilde{\Sigma})\ket{z}\right\rVert}-\frac{\mathcal{L}_N\widetilde{\Sigma}^{1/2}\ket{z}}{\norm{\mathcal{L}_N\widetilde{\Sigma}^{1/2}\ket{z}}}}
\leq 
\frac{2 N
\cdot
\widetilde{\varepsilon}}{\norm{\mathcal{L}_N\widetilde{\Sigma}^{1/2}\ket{z}}}
\eqqcolon (\star).
\end{align*}
Plugging in the definition of $\widetilde{\varepsilon}$ gives
\begin{align}
\Big\lVert\ket{\chi}-\ket{y}\Big\rVert
\leq
(\star)
=
\frac{2 N
\cdot
\frac{1}{2N}
\cdot
({\varepsilon}/{2})}
{
2\norm{\mathcal{L}_N\frac{\widetilde{\Sigma}^{1/2}\ket{z}}{\norm{\widetilde{\Sigma}^{1/2}\ket{z}}}}
\cdot 
\sqrt{\widehat{\kappa}}
\norm{\widetilde{\Sigma}^{1/2}\ket{z}}
}
\leq
\frac{\varepsilon}{2}.
\label{eq:chi_y_diff}
\end{align}
Here, the last inequality follows from 
$\norm{\mathcal{L}_N{\widetilde{\Sigma}^{1/2}\ket{z}}/{\big\lVert\widetilde{\Sigma}^{1/2}\ket{z}\big\rVert}}\geq 1/2$,
cf.\ \autoref{prop-for-LN-times-unit-vec2},
and
$\sqrt{\widehat{\kappa}}\norm{\widetilde{\Sigma}^{1/2}\ket{z}_{\mathrm{I}}}\geq 1$, which follows from \autoref{norm-sqrt-sigma-times-ket-z} and the fact that all the eigenvalues of $\widetilde{\Sigma}$ lie in $\left[\frac{1}{L\kappa},1\right]\subseteq\left[\frac{1}{L\widetilde{\kappa}},1\right]=\left[\frac{1}{\widehat{\kappa}},1\right]$, as seen in the proof of \autoref{main-thm-for-gen-mat}.
Combining \eqref{eq:chi_y_diff} with \eqref{diff-varphi-chi-coro-sum-of-inc}, we now know that the prepared quantum state $\ket{\widetilde{\varphi}}$ satisfies:
\begin{align*}
\Big\lVert
\ket{\widetilde{\varphi}}_{\mathrm{Q}\mathrm{I}}-\ket{0}_{\mathrm{Q}}\ket{y}_{\mathrm{I}}
\Big\rVert
&\leq
\Big\lVert
\ket{\widetilde{\varphi}}-\ket{0}_{\mathrm{Q}}\ket{\chi}_{\mathrm{I}}
\Big\rVert
+
\Big\lVert
\ket{0}_{\mathrm{Q}}\ket{\chi}_{\mathrm{I}}-\ket{0}_{\mathrm{Q}}\ket{y}_{\mathrm{I}}
\Big\rVert
\leq \frac{\varepsilon}{2}+\frac{\varepsilon}{2}=\varepsilon. \qedhere
\end{align*}
\end{proof}
\subsection{QAE and normalising factor estimation}
\begin{proposition}[{\cite[Theorem 12]{BHMT02}} restated]\label{lem-QAE}
Suppose we have access to a quantum circuit $U$ that maps $\ket{0}\ket{0}^{\otimes n}$ to the $(n+1)$-qubit state of the form
\[
\ket{s}=a\ket{1}\ket{\Psi_G}+b\ket{0}\ket{\Psi_B},
\]
where $\ket{\Psi_G},\ket{\Psi_B}$ are $n$-qubit normalised quantum states, not necessarily orthogonal.
Then, for $\varepsilon_{\mathrm{QAE}}>0$, one can construct a quantum circuit that can output either $\widetilde{a}$ or $\widetilde{p}$ satisfying
\begin{gather}
\big| \widetilde{a} - |a|\big| \leq \varepsilon_{\mathrm{QAE}}, \label{qae-prop-amp}\\
\big| \widetilde{p} - |a|^2\big| \leq \varepsilon_{\mathrm{QAE}} \label{qae-prop-prob},
\end{gather}
with probability at least
$0.99$\footnote{The original version of QAE was presented as a method with a constant success probability $8/\pi^2\approx 81\%$. We can enhance this to an arbitrary value $1-\delta$ by repeating the QAE operation $O\left(\log\left({1}/{\delta}\right)\right)$ times and taking the median as the estimate (the so-called powering lemma, see e.g.\ \cite[Lemma 2.1 (Lemma 1 of the full version: {arXiv:1504.06987})]{Mon15}).
In this paper, we fix $1-\delta$ to a specific value of 99\% and treat the overheads in gate count and gate depth arising from the repetition of QAE as $O(1)$.
}, by using $M=O(1/\varepsilon_{\mathrm{QAE}})$ calls to reflection operators $\mathcal{S}_s=U\left(2(\kb{0})^{\otimes (n+1)}-I^{\otimes (n+1)}\right)U^{\dagger}$ and $\mathcal{S}_G=(I-2\kb{1})\otimes I^{\otimes n}$, and $O(\log^2 M)$ rotation gates and Hadamard gates for implementing QFT and inverse QFT circuit. If implementing $U$ takes $T_U$ elementary gates, then the overall gate cost for outputting either $\widetilde{a}$ or $\widetilde{p}$ becomes $O\left(\frac{1}{\varepsilon_{\mathrm{QAE}}}\left(T_U + n\right)\right)$.
\end{proposition}
\begin{proof}
\cite[Theorem 12]{BHMT02} has already proved the error bound \eqref{qae-prop-prob}, but not \eqref{qae-prop-amp}. However, \eqref{qae-prop-amp} follows as straightforwardly. As already shown in the proof of \cite[Theorem 12]{BHMT02}, QAE circuit outputs a probabilistic $\widetilde{\theta}$ satisfying $|\widetilde{\theta}-\theta|\leq \frac{\pi}{M}\leq \varepsilon_{\mathrm{QAE}}$ for $\theta\in\{\theta_a,\pi-\theta_a\}$ such that $\sin^2(\theta_a)=|a|^2$ with probability at least $8/\pi^2$. 
Define $\widetilde{a}\coloneqq \lvert\sin(\widetilde{\theta})\rvert$ and $\widetilde{p}\coloneqq \sin^2(\widetilde{\theta})$.
Since both functions $f_1(x)\coloneqq \lvert\sin x\rvert$ and $f_2(x)\coloneqq \sin^2 x$ are Lipschitz continuous with Lipschitz constant equal to $1$, $|\widetilde{\theta}-\theta|\leq \varepsilon_{\mathrm{QAE}}$ implies both \eqref{qae-prop-amp} and \eqref{qae-prop-prob}.
To deduce the overall cost for constructing QAE circuit, note that 
$2(\kb{0})^{\otimes (n+1)}-I^{\otimes (n+1)}$ can be realised via $-X^{\otimes (n+1)} \mathrm{MC}Z X^{\otimes (n+1)}$ where $\mathrm{MC}Z$ (multi-controlled-$Z$) could be realised via $(n+1)$-qubit Toffoli gate costing $O(n)$ elementary gates (see \autoref{implement-q-qubit-toffoli}) and two single-qubit Hadamard gates\footnote{A similar argument can be observed in \cite[Definition 7]{RR23} quoting \cite{SP13}.}. Therefore, the reflection about initial state $\mathcal{S}_s$ costs $2T_U+O(n)$ elementary gates. The reflection about good state $\mathcal{S}_G$ is just a single-qubit $Z$ gate costing $O(1)$ each. Hence, the overall cost for a QAE operation becomes $O\left(\frac{1}{\varepsilon_{\mathrm{QAE}}}\left(T_U + n + 1\right)+\log^2\left(\frac1{\varepsilon_{\mathrm{QAE}}}\right)\right)=O\left(\frac{1}{\varepsilon_{\mathrm{QAE}}}\left(T_U + n\right)\right)$ elementary gates.
\end{proof}
\begin{proposition}\label{qae-for-main-thm}
Let $\varepsilon_{\mathrm{QAE}}>0$.
Let $\Sigma$ be a matrix satisfying the same assumption as \autoref{main-thm-for-gen-mat}, and suppose that we have access to the unitaries $U_{\Sigma}$ and $U_{\ket{z}}$ in \autoref{main-thm-for-gen-mat}.
If the block-encoding error $\varepsilon_{U_\Sigma}$ of $U_{\Sigma}$ satisfies
\[
\frac{\varepsilon_{U_\Sigma}}{\widetilde{\lambda}_{\max}}=o\left(\frac{\varepsilon_{\mathrm{QAE}}} {L\widetilde{\kappa}} \cdot \log^{-3}\left(\frac{L\widetilde{\kappa}}{\varepsilon_{\mathrm{QAE}}}\right)\right),
\tagaligneq \label{cond-for-using-chakraborty-et-al-qae-prop1}
\] 
then we can output the value of $\norm{\widetilde{\Sigma}^{1/2}\ket{z}}$, where $\widetilde{\Sigma}\coloneqq \Sigma/\widetilde{\lambda}_{\max}$, to within $\varepsilon_{\mathrm{QAE}}$ with probability $0.99$ by using 
\[
O
\left(
\frac{1}{\varepsilon_{\mathrm{QAE}}}
\left(
\widetilde{q}
+
n
+
\frac{\alpha_{U_\Sigma}}{\widetilde{\lambda}_{\max}}
\,L\widetilde{\kappa}\,\left(a_{U_\Sigma}+T_{U_\Sigma}\right)\log^2\left(\frac{L\widetilde{\kappa}}{\varepsilon_{\mathrm{QAE}}}\right)
+T_{U_{\ket{z}}}
\right)
\right) \tagaligneq \label{qae-cost-thm-1st-abstract-cost}
\]
elementary gates, where $\mathbb{Z}^{+}\ni\widetilde{q}=a_{U_\Sigma}+O\left(\log\log(1/{\varepsilon_{\mathrm{QAE}}})\right)$, and an ancillary register of $\widetilde{q}+2$ qubits.
In particular, if the assumptions about data loaders in polylogarithmic time (\autoref{assump-plus-TU-ket-y-log2-N} and \autoref{assump-QROM-like-unitary}), and the assumption about parameter estimation (\autoref{assump-LK-plus}) hold,
we have $\varepsilon_{U_\Sigma}=0$ and condition \eqref{cond-for-using-chakraborty-et-al-qae-prop1} is satisfied automatically. Under such a circumstance,
the quantum circuit to output $\norm{\widetilde{\Sigma}^{1/2}\ket{z}}$ to within $\varepsilon_{\mathrm{QAE}}$ with probability at least $0.99$ has
\[
O
\left(
\frac{1}{\varepsilon_{\mathrm{QAE}}}
\frac{\norm{\Sigma}_F}{{\lambda}_{\max}}
\,{\kappa}\,
\polylog(N)
\log^2\left(\frac{{\kappa}}{\varepsilon_{\mathrm{QAE}}}\right)
\right) \tagaligneq \label{qae-cost-thm-2nd-concrete-cost-0}
\]
elementary gate depth and uses an ancillary register of $\widetilde{q}+2=O\left(\log(N)+\log\log({1}/{\varepsilon_{\mathrm{QAE}}})\right)$ qubits.
\end{proposition}
\begin{proof}
Similarly to the proof of \autoref{main-thm-for-gen-mat}, letting $\widehat{\kappa}\equiv \max\{L\widetilde{\kappa},2\}$ 
and $\widetilde{\varepsilon}\equiv \varepsilon_{\mathrm{QAE}}/2$,
we apply \autoref{prop-for-Hc2} to $\widetilde{\Sigma}$, by using $U_\Sigma$ as the $\left(\frac{\alpha_{U_\Sigma}}{\widetilde{\lambda}_{\max}},a_{U_\Sigma},\frac{\varepsilon_{U_\Sigma}}{\widetilde{\lambda}_{\max}}\right)$-block-encoding of $\widetilde{\Sigma}$.
This is possible as long as $\frac{\varepsilon_{U_\Sigma}}{\widetilde{\lambda}_{\max}}=o(\widetilde{\varepsilon}/\left(\widehat{\kappa}\log^3(\widehat{\kappa}/\widetilde{\varepsilon})\right))$ holds, which is already assumed in \eqref{cond-for-using-chakraborty-et-al-qae-prop1}.
Doing so allows us to construct a unitary $U_{\widetilde{\Sigma}^{1/2}}\left(I^{\otimes \widetilde{q}}\otimes U_{\ket{z}}\right)$, which requires a block-encoding ancillary register $\mathrm{Q}$ of $\widetilde{q}=a_{U_\Sigma}+O\left(\log\log(1/{\widetilde{\varepsilon}})\right)$ qubits, and prepare the state of the form
\begin{align*}
U_{\widetilde{\Sigma}^{1/2}}\left(I^{\otimes \widetilde{q}}\otimes U_{\ket{z}}\right)\ket{0}_{\mathrm{Q}}\ket{0}_{\mathrm{I}}
&=\frac{1}{2}\ket{0}_{\mathrm{Q}}f(\widetilde{\Sigma})\ket{z}_{\mathrm{I}}+\ket{\widetilde{z}^{\perp}}_{\mathrm{Q}\mathrm{I}}, \tagaligneq \label{state-after-apply-tilde-sigma12-qae-prop}
\end{align*}
where $\ket{\widetilde{z}^{\perp}}_{\mathrm{Q}\mathrm{I}}$ is an unnormalised state such that $\left(\kb{0}_{\mathrm{Q}}\otimes I^{\otimes n}\right)\ket{\widetilde{z}^{\perp}}_{\mathrm{Q}\mathrm{I}}=0$,
mirroring \eqref{state-after-apply-sigma12-main-thm} and \eqref{state-after-apply-sigma12-corollary-sum-of-inc}.
Here, $U_{\widetilde{\Sigma}^{1/2}}$ exactly block-encodes the matrix $f(\widetilde{\Sigma})$ with property $\big\lVert f(\widetilde{\Sigma})-\widetilde{\Sigma}^{1/2}\big\rVert\leq\widetilde{\varepsilon}$.
By adding a single-qubit flag register to the state \eqref{state-after-apply-tilde-sigma12-qae-prop} and applying multi-controlled-$X$ (MC$X$) on the flag register, conditioned on register $\mathrm{Q}$ being in $\ket{0}^{\otimes \widetilde{q}}$, we can obtain a state in the form where QAE (\autoref{lem-QAE}) is applicable.
Hence, the unitary operator that QAE needs to query is $U\coloneqq (\mathrm{MC}_{\mathrm{Q}}X\otimes I^{\otimes n}) (I\otimes U_{\widetilde{\Sigma}^{1/2}})\left(I^{\otimes (1+ \widetilde{q})}\otimes U_{\ket{z}}\right)$. 
By \autoref{implement-q-qubit-toffoli}, $\mathrm{MC}X$ controlled on $\widetilde{q}$ qubits costs $O(\widetilde{q})$ two-qubit gates and a single-qubit ancillary register. 
Therefore, totally, the number of ancillary qubits required is $\widetilde{q}+2$, where we need $\widetilde{q}$ qubits for the block-encoding ancilla, another one for the flag register for QAE, and the last one for the implementation of $\mathrm{MC}X$ gate. 
Implementing $U$ takes
\[
T_U
=
O\left(
\widetilde{q}
+T_{U_{\widetilde{\Sigma}^{1/2}}}+T_{U_{\ket{z}}}
\right)
\]
elementary gates in total. Here, recall that by \autoref{prop-for-Hc2}, we have
\[
T_{U_{\widetilde{\Sigma}^{1/2}}}
=
O
\left(
\frac{\alpha_{U_\Sigma}}{\widetilde{\lambda}_{\max}}\, L\widetilde{\kappa}\,\left(a_{U_\Sigma}+T_{U_\Sigma}\right)\log^2\left(\frac{L\widetilde{\kappa}}{\widetilde{\varepsilon}}\right)
\right),
\]
and $\widetilde{\varepsilon}=\Theta(\varepsilon_{\mathrm{QAE}})$ by definition.
We then apply \autoref{lem-QAE} to output an estimate $\widetilde{\mathcal{X}}$ of $\mathcal{X}\coloneqq \norm{f(\widetilde{\Sigma})\ket{z}}/2$ to within $\varepsilon_{\mathcal{X}}\coloneqq \varepsilon_{\mathrm{QAE}}/4$ with probability at least $0.99$, by using $O\left(\frac{1}{\varepsilon_{\mathcal{X}}}\left(\widetilde{q}+n+T_U\right)\right)$ elementary gates, which equals to \eqref{qae-cost-thm-1st-abstract-cost} as claimed.
Lastly, the estimate $\widetilde{\mathcal{X}}$ satisfies
\begin{align*}
\bigg|
2\widetilde{\mathcal{X}} - \norm{\widetilde{\Sigma}^{1/2}\ket{z}}
\bigg|
&\leq 
2
\left|
\widetilde{\mathcal{X}} - \mathcal{X}
\right|
+
\bigg|
\norm{f(\widetilde{\Sigma})\ket{z}}
 - \norm{\widetilde{\Sigma}^{1/2}\ket{z}}
\bigg|
\leq
2\varepsilon_{\mathcal{X}}
+
\widetilde{\varepsilon}
\leq
\varepsilon_{\mathrm{QAE}}.
\end{align*}
The concrete complexity \eqref{qae-cost-thm-2nd-concrete-cost-0} under \autoref{assump-QROM-like-unitary}, \autoref{assump-plus-TU-ket-y-log2-N}, and \autoref{assump-LK-plus} follows straightforwardly, in the same manner as the proof of \autoref{main-thm-gen-mat-concrete-cost}.
\end{proof}
\begin{corollary}\label{qae-for-main-thm-sum-of-inc-case}
Let $\varepsilon_{\mathrm{QAE}}>0$.
Let $\Sigma$ be a matrix satisfying the same assumption as \autoref{main-thm-for-gen-mat} and $\ket{z}$ be as in \autoref{main-thm-for-gen-mat}.
Suppose that the assumptions about data loaders in polylogarithmic time (\autoref{assump-plus-TU-ket-y-log2-N} and \autoref{assump-QROM-like-unitary}), and the assumption about parameter estimation (\autoref{assump-LK-plus}) hold in the same way as
\autoref{main-thm-coro-concrete-cost-sum-of-inc}. 
Then, we can output the value of $\norm{\mathcal{L}_N\widetilde{\Sigma}^{1/2}\ket{z}}$ to within $\varepsilon_{\mathrm{QAE}}$, where $\widetilde{\Sigma}\coloneqq \Sigma/\widetilde{\lambda}_{\max}$, with probability at least $0.99$ by using a quantum circuit with
\[
O
\left(
\frac{1}{\varepsilon_{\mathrm{QAE}}}
\frac{\norm{\Sigma}_F}{{\lambda}_{\max}}
\,{\kappa}\,
N\polylog(N)
\log^2\left(\frac{{\kappa}}{\varepsilon_{\mathrm{QAE}}}\right)
\right) \tagaligneq \label{qae-cost-thm-2nd-concrete-cost}
\]
elementary gate depth and an ancillary register of $\widetilde{q}+2=O\left(\log(N)+\log\log({1}/{\varepsilon_{\mathrm{QAE}}})\right)$ qubits.
\end{corollary}
\begin{proof}
The proof goes in the same manner as that of \autoref{qae-for-main-thm}, 
but with 
$\widetilde{\varepsilon}\coloneqq \varepsilon_{\mathrm{QAE}}/(2N)$ and
$\varepsilon_{\mathcal{X}}\coloneqq \varepsilon_{\mathrm{QAE}}/(2\sqrt{2}(N+1))=\Theta\left(\varepsilon_{\mathrm{QAE}}/N\right)$
instead due to the error bound \eqref{QAE-sum-of-inc-error-bound-ineq} given below.
Note that the number of the required ancillary qubits, $\widetilde{q}$, remains of the same order even though $\widetilde{\varepsilon}$ now depends on $N$, because $O\left(\log(N)+\log\log\left(N/\varepsilon_{\mathrm{QAE}}\right)\right)=O\left(\log(N)+\log\log(1/\varepsilon_{\mathrm{QAE}})\right)$.
Again, we need to add a flag register and an MC$X$ gate, and perform QAE to the state \eqref{state-after-sigma12-and-LN}, similarly to what we do to \eqref{state-after-apply-tilde-sigma12-qae-prop} in \autoref{qae-for-main-thm}. The unitary operator that QAE needs to query becomes $U\coloneqq (\mathrm{MC}_{\mathrm{Q}}X\otimes I^{\otimes n}) (I\otimes\widetilde{U}_{\mathcal{L}_N})(I^{\otimes (\widehat{a}+1)}\otimes U_{\widetilde{\Sigma}^{1/2}})\left(I^{\otimes (\widetilde{q}+1)}\otimes U_{\ket{z}}\right)$, where $\widehat{a}=\ceil{\log_2(N+1)}$, and its implementation requires $T_U
=
O\left(
\widetilde{q}
+T_{U_{\mathcal{L}_N}}
+T_{U_{\widetilde{\Sigma}^{1/2}}}+T_{U_{\ket{z}}}
\right)$ elementary gates.
By using QAE (\autoref{lem-QAE}), we can output the estimate $\widetilde{\mathcal{X}}$ of $\mathcal{X}\coloneqq \big\lVert \mathcal{L}_N f(\widetilde{\Sigma}) \ket{z} \big\rVert / (2\norm{\mathcal{L}_N}_F)$ satisfying $\left| \widetilde{\mathcal{X}}-\mathcal{X} \right|\leq \varepsilon_{\mathcal{X}}$ with probability at least $0.99$.
Note that, under \autoref{assump-plus-TU-ket-y-log2-N}, \autoref{assump-QROM-like-unitary}, and \autoref{assump-LK-plus},
the $N$ factor in $\widetilde{\varepsilon}$ is absorbed into $\polylog (N)$ term due to $T_{U_{\widetilde{\Sigma}^{1/2}}}$'s logarithmic dependence on $\widetilde{\varepsilon}$, 
and the overall cost required becomes \eqref{qae-cost-thm-2nd-concrete-cost}.
Since $\norm{\mathcal{L}_N}_F=\sqrt{N(N+1)/2}\leq (N+1)/\sqrt{2}$ holds by definition and $\norm{\mathcal{L}_N}_2\leq N$ by \autoref{prop-for-LN-times-unit-vec2}, the estimate $\widetilde{\mathcal{X}}$ satisfies
\begin{align*}
\left\lvert 2\norm{\mathcal{L}_N}_F\widetilde{\mathcal{X}}-
\big\lVert\mathcal{L}_N\widetilde{\Sigma}^{1/2}\ket{z}\big\rVert
\right\rvert
&\leq
2\norm{\mathcal{L}_N}_F
\left\lvert \widetilde{\mathcal{X}}-\mathcal{X}\right\rvert
+
\left\lvert \big\lVert \mathcal{L}_N f(\widetilde{\Sigma}) \ket{z} \big\rVert-
\big\lVert\mathcal{L}_N\widetilde{\Sigma}^{1/2}\ket{z}\big\rVert
\right\rvert\\
&\leq
2\norm{\mathcal{L}_N}_F
\cdot
\varepsilon_{\mathcal{X}}
+
\norm{\mathcal{L}_N}_2
\cdot
\widetilde{\varepsilon} \tagaligneq \label{QAE-sum-of-inc-error-bound-ineq}
\\
&\leq
\varepsilon_{\mathrm{QAE}}. \qedhere
\end{align*}
\end{proof}
\begin{corollary}[Absolute-error estimate to $\norm{\vec{x}}$]\label{qae-w-abs-err-est}
Under the setting of \autoref{qae-for-main-thm},
the norm $\norm{\vec{x}}$ of $\vec{x}=\Sigma^{1/2}\vec{z}$ can be estimated to within $\widehat{\varepsilon}>0$ with probability at least $0.99$ by using 
\[
O
\left(
\frac{\norm{\vec{z}}_2\sqrt{\widetilde{\lambda}_{\max}}}{\widehat{\varepsilon}}
\left(
\widetilde{q}
+
n
+
\frac{\alpha_{U_\Sigma}}{\widetilde{\lambda}_{\max}}
\,L\widetilde{\kappa}\,\left(a_{U_\Sigma}+T_{U_\Sigma}\right)\log^2\left(\frac{\norm{\vec{z}}_2\sqrt{\widetilde{\lambda}_{\max}} \, L\widetilde{\kappa}}{\widehat{\varepsilon}}\right)
+T_{U_{\ket{z}}}
\right)
\right)
\]
elementary gates, where $\mathbb{Z}^{+}\ni\widetilde{q}=a_{U_\Sigma}+O\left(\log\log\left(\norm{\vec{z}}_2\sqrt{\widetilde{\lambda}_{\max}}\Big/{\widehat{\varepsilon}}\right)\right)$, 
and an ancillary register of $\widetilde{q}+2$ qubits.
If \autoref{assump-plus-TU-ket-y-log2-N}, \autoref{assump-QROM-like-unitary}, and \autoref{assump-LK-plus} hold,
the quantum circuit to output $\norm{\vec{x}}$ to within $\widehat{\varepsilon}$ with probability at least $0.99$ has
\[
O
\left(
\frac{\norm{\vec{z}}_2\sqrt{\lambda_{\max}}}{\widehat{\varepsilon}}
\frac{\norm{\Sigma}_F}{{\lambda}_{\max}}
\,{\kappa}\,
\polylog(N)
\log^2\left(\frac{\norm{\vec{z}}_2\sqrt{\lambda_{\max}}\cdot\kappa}{\widehat{\varepsilon}}\right)
\right)  \tagaligneq \label{qae-cost-abs-err-concrete-cost}
\]
elementary gate depth and an ancillary register of $\widetilde{q}+2=O\left(\log(N)+\log\log(\norm{\vec{z}}_2\sqrt{\lambda_{\max}}/{\widehat{\varepsilon}})\right)$ qubits.
\end{corollary}
\begin{proof}
Using \autoref{qae-for-main-thm} with $\varepsilon_{\mathrm{QAE}}\equiv \widehat{\varepsilon}\Big/\Big(\norm{\vec{z}}_2\sqrt{\widetilde{\lambda}_{\max}}\Big)$,
we get an estimate $\widetilde{\mathcal{X}}$ of $\mathcal{X}\coloneqq \lVert\widetilde{\Sigma}^{1/2}\ket{z}\rVert
$ such that $\left|\widetilde{\mathcal{X}}-\mathcal{X}\right|\leq \varepsilon_{\mathrm{QAE}}$ with probability at least $0.99$. 
Multiplying both sides of the inequality by $\norm{\vec{z}}_2\sqrt{\widetilde{\lambda}_{\max}}$, 
we have $\left|
\norm{\vec{z}}_2\sqrt{\widetilde{\lambda}_{\max}}
\cdot
\widetilde{\mathcal{X}}
-
\lVert\Sigma^{1/2}\vec{z}\rVert
\right|
\leq 
\widehat{\varepsilon}$,
as $\sqrt{\widetilde{\lambda}_{\max}}\widetilde{\Sigma}^{1/2}={\Sigma}^{1/2}$ and $\norm{\vec{z}}\ket{z}=\vec{z}$ by construction. This proves the error bound since $\norm{{\Sigma}^{1/2}\vec{z}}=\norm{\vec{x}}$.
The resource requirements follow directly by
plugging 
$\varepsilon_{\mathrm{QAE}}\equiv \widehat{\varepsilon}\Big/\Big(\norm{\vec{z}}_2\sqrt{\widetilde{\lambda}_{\max}}\Big)$
into
\eqref{qae-cost-thm-1st-abstract-cost} and \eqref{qae-cost-thm-2nd-concrete-cost-0} of \autoref{qae-for-main-thm}.
\end{proof}
\begin{lemma}\label{bounds-for-estimate-for-lamb-min} Let $\lambda_{\max},\lambda_{\min}>0$ and $\kappa=\lambda_{\max}/\lambda_{\min}$. 
If $\widetilde{\lambda}_{\max},\widetilde{\kappa}>0$ and $L,K\geq 1$ are such that \eqref{main-thm-col-assump2} in \autoref{assmp-estimate-lamb-max-kappa} holds, then
we have
\[
\frac{1}{K}\lambda_{\min}\leq \frac{\widetilde{\lambda}_{\max}}{\widetilde{\kappa}}\leq L\lambda_{\min}.
\]
\end{lemma}
\begin{proof}
From the bounds \eqref{main-thm-col-assump2} in \autoref{assmp-estimate-lamb-max-kappa}, it holds that
\[
\frac{\lambda_{\max}}{K\kappa}
\leq
\frac{\widetilde{\lambda}_{\max}}{\widetilde{\kappa}}
\leq
\frac{L\lambda_{\max}}{\kappa},
\]
and the proof follows since $\lambda_{\max}/\kappa=\lambda_{\min}$.
\end{proof}
\begin{lemma}[Relative-error bound in QAE estimation for $\norm{\vec{x}}$]\label{epsilon-QAE-implies-epsilon-hat}
Let ${C}>0$ be a constant and assume parameters under the setting of \autoref{qae-for-main-thm}.
Then,
$\varepsilon_{\mathrm{QAE}}\in \left(0,\frac{1}{{C}\sqrt{L\widetilde{\kappa}}}\right]$ 
implies 
$\widehat{\varepsilon}\equiv \norm{\vec{z}}_2\sqrt{\widetilde{\lambda}_{\max}}\ \varepsilon_{\mathrm{QAE}}\in\left(0,\frac{\norm{\vec{x}}_2}{{C}}\right]$.
\end{lemma}
\begin{proof}
The lower bound $\widehat{\varepsilon}>0$ is trivial. To prove the upper bound for $\widehat{\varepsilon}$, observe that
\begin{align*}
\widehat{\varepsilon}
&= \norm{\vec{z}}_2\sqrt{\widetilde{\lambda}_{\max}}\ \varepsilon_{\mathrm{QAE}}
\leq
\frac{\norm{\vec{z}}_2\sqrt{\widetilde{\lambda}_{\max}}}{{C}\sqrt{L\widetilde{\kappa}}}
=
\frac{\norm{\vec{z}}_2}{\norm{\vec{x}}_2}
\cdot
\frac{\norm{\vec{x}}_2}{{C}} 
\cdot
\frac{\sqrt{\widetilde{\lambda}_{\max}}}{\sqrt{L\widetilde{\kappa}}}
\leq 
\frac{\norm{\vec{x}}_2}{{C}}, \tagaligneq \label{widehat-eps-ineq-relative-err}
\end{align*}
where the last inequality follows since $\sqrt{\widetilde{\lambda}_{\max}/(L\widetilde{\kappa})}\leq \sqrt{\lambda_{\min}}$ by \autoref{bounds-for-estimate-for-lamb-min} and 
${\norm{\vec{z}}_2}/{\norm{\vec{x}}_2}\leq (1/\sqrt{\lambda_{\min}})$ by \eqref{bounds-for-norm-x-div-by-norm-z} in \autoref{norm-sqrt-sigma-times-ket-z}.
\end{proof}
\begin{corollary}[Relative-error estimate to $\norm{\vec{x}}$]\label{qae-w-rel-err-est}
Under the setting of \autoref{qae-for-main-thm} and for a constant ${C}>0$,
the norm $\norm{\vec{x}}$ of $\vec{x}=\Sigma^{1/2}\vec{z}$ can be estimated to within $\widehat{\varepsilon}\in\left(0,\frac{\norm{\vec{x}}_2}{{C}}\right]$ with probability at least $0.99$ by using 
\[
O
\left(
{C}{\sqrt{L\widetilde{\kappa}}}
\left(
\widetilde{q}
+
n
+
\frac{\alpha_{U_\Sigma}}{\widetilde{\lambda}_{\max}}
\,L\widetilde{\kappa}\,\left(a_{U_\Sigma}+T_{U_\Sigma}\right)\log^2\left({C}{L\widetilde{\kappa}}\right)
+T_{U_{\ket{z}}}
\right)
\right)
\]
elementary gates, where $\mathbb{Z}^{+}\ni\widetilde{q}=a_{U_\Sigma}+O\left(\log\log\left({C}L\widetilde{\kappa}\right)\right)$, 
and an ancillary register of $\widetilde{q}+2$ qubits.
If \autoref{assump-plus-TU-ket-y-log2-N}, \autoref{assump-QROM-like-unitary}, and \autoref{assump-LK-plus} hold,
the quantum circuit to output $\norm{\vec{x}}$ to within $\widehat{\varepsilon}\in\left(0,\frac{\norm{\vec{x}}_2}{{C}}\right]$  with probability at least $0.99$ has
\[
O
\left(
{C}
\frac{\norm{\Sigma}_F}{{\lambda}_{\max}}
\,{\kappa}^{1.5}\,
\polylog(N)
\log^2\left({C}{\kappa}\right)
\right) 
\]
elementary gate depth and an ancillary register of $\widetilde{q}+2=O\left(\log(N)+\log\log({C}\kappa)\right)$ qubits.
\end{corollary}
\begin{proof}
By letting $\varepsilon_{\mathrm{QAE}}\coloneqq \frac{1}{{C}\sqrt{L\widetilde{\kappa}}}$ in \autoref{qae-for-main-thm},
we can use \autoref{epsilon-QAE-implies-epsilon-hat} and
the error bound of the QAE estimate follows in the same manner as in the proof of \autoref{qae-w-abs-err-est} but with
$\widehat{\varepsilon}\leq 
\frac{\norm{\vec{x}}_2}{{C}}$. 
The resource requirements also follow directly by
plugging
$\varepsilon_{\mathrm{QAE}}= \frac{1}{{C}\sqrt{L\widetilde{\kappa}}}$
into
\eqref{qae-cost-thm-1st-abstract-cost} and \eqref{qae-cost-thm-2nd-concrete-cost-0} of \autoref{qae-for-main-thm}.
\end{proof}
\begin{corollary}[QAE estimation for $\norm{\vec{y}}=\norm{\mathcal{L}_N\Sigma^{1/2}\vec{z}}$]\label{qae-for-norm-y}
Under the setting of
\autoref{qae-for-main-thm-sum-of-inc-case},
the norm $\norm{\vec{y}}=\norm{\mathcal{L}_N\Sigma^{1/2}\vec{z}}$ can be estimated to within (absolute error) $\widehat{\varepsilon}>0$ with probability at least $0.99$ by use of a quantum circuit with
\[
O
\left(
\frac{\norm{\vec{z}}_2\sqrt{\lambda_{\max}}}{\widehat{\varepsilon}}
\frac{\norm{\Sigma}_F}{{\lambda}_{\max}}
\,{\kappa}\,
N\polylog(N)
\log^2\left(\frac{\norm{\vec{z}}_2\sqrt{\lambda_{\max}}\cdot\kappa}{\widehat{\varepsilon}}\right)
\right) \tagaligneq \label{qae-cost-sum-of-inc-abs}
\]
elementary gate depth and an ancillary register of $\widetilde{q}+2=O\left(\log(N)+\log\log(\norm{\vec{z}}_2\sqrt{\lambda_{\max}}/{\widehat{\varepsilon}})\right)$ qubits.
Alternatively,
for a constant ${C}>0$,
the norm $\norm{\vec{y}}=\norm{\mathcal{L}_N\Sigma^{1/2}\vec{z}}$ can be estimated to within (relative error) $\widehat{\varepsilon}\in\left(0,\frac{\norm{\vec{y}}_2}{{C}}\right]$ with probability at least $0.99$ by use of a quantum circuit with
\[
O
\left(
{C}
\frac{\norm{\Sigma}_F}{{\lambda}_{\max}}
\,{\kappa}^{1.5}\,
N\polylog(N)
\log^2\left({C}{\kappa}\right)
\right) \tagaligneq \label{qae-cost-sum-of-inc-rel}
\]
elementary gate depth and an ancillary register of $\widetilde{q}+2=O\left(\log(N)+\log\log({C}\kappa)\right)$ qubits.
\end{corollary}
\begin{proof}
Since $\norm{\vec{z}}_2\sqrt{\widetilde{\lambda}_{\max}}\norm{\mathcal{L}_N\widetilde{\Sigma}^{1/2}\ket{z}}=\norm{\vec{y}}$,
the absolute error case is proved similarly to the proof of \autoref{qae-w-abs-err-est} by using \autoref{qae-for-main-thm-sum-of-inc-case} with $\varepsilon_{\mathrm{QAE}}\equiv \widehat{\varepsilon}\Big/\Big(\norm{\vec{z}}_2\sqrt{\widetilde{\lambda}_{\max}}\Big)$.
The relative error case also follows similarly to \autoref{qae-w-rel-err-est}, where we use the bound similar to \autoref{epsilon-QAE-implies-epsilon-hat}, but with $\norm{\vec{y}}$ instead of $\norm{\vec{x}}$.
Here, note that by \autoref{prop-for-LN-times-unit-vec2} and \eqref{bounds-for-norm-sigma12-z} in \autoref{norm-sqrt-sigma-times-ket-z}, it holds that
\[
\norm{\vec{y}}
=
\norm{\mathcal{L}_N\Sigma^{1/2}\vec{z}}
=\norm{\mathcal{L}_N\frac{\Sigma^{1/2}\vec{z}}{\norm{\Sigma^{1/2}\vec{z}}}}
\norm{\Sigma^{1/2}\frac{\vec{z}}{\norm{\vec{z}}}}
\norm{\vec{z}}
\geq
\frac12
\norm{\Sigma^{1/2}\frac{\vec{z}}{\norm{\vec{z}}}}
\norm{\vec{z}}
\geq \frac{\sqrt{\lambda_{\min}}}{2}\norm{\vec{z}}, \tagaligneq \label{norm-y-norm-z-ratio-ineq}
\]
which implies $\norm{\vec{z}}/\norm{\vec{y}}\leq 2/\sqrt{\lambda_{\min}}$. By \eqref{widehat-eps-ineq-relative-err}, it holds that if $\varepsilon_{\mathrm{QAE}}\in \left(0,\frac{1}{{2C}\sqrt{L\widetilde{\kappa}}}\right]$,
then
$\widehat{\varepsilon}\equiv \norm{\vec{z}}_2\sqrt{\widetilde{\lambda}_{\max}}\ \varepsilon_{\mathrm{QAE}}\in\left(0,\frac{\norm{\vec{y}}_2}{{C}}\right]$, so setting $\varepsilon_{\mathrm{QAE}}\coloneqq \frac{1}{{2C}\sqrt{L\widetilde{\kappa}}}$ in \autoref{qae-for-main-thm-sum-of-inc-case} proves the claim.
\end{proof}
\subsection{Application and post-processing examples}\label{application-ex-and-post-process-subsec}
\autoref{main-thm-gen-mat-concrete-cost} says that our algorithm offers the state preparation of a correlated Gaussian vector with complexity $\widetilde{O}\left(\frac{\norm{\Sigma}_F}{\lambda_{\max}}\kappa^{1.5}\right)$.
Therefore, the more ill-conditioned the covariance matrix is, the poorer quantum speedup we could gain from using our algorithm.
On the other hand, when the covariance matrix is well-conditioned, i.e.\ $\kappa=\Theta(1)$, our algorithm may be advantageous.
In such a case, all eigenvalues are of the same order, $\lambda_i=\Theta(\lambda_{\max})$, so no small subset of principal components captures a dominant fraction of the variance, which means low-rank approximations do not work. 
Moreover, the condition number does not imply sparsity or decay of correlations; there exist infinitely many dense, well-conditioned covariance matrices\footnote{This follows from the eigendecomposition $\Sigma=Q\Lambda Q^T$, where $\kappa$ depends solely on $\Lambda$, while $Q$ may be dense.}. 
Dense and well-conditioned covariance matrices, in fact, constitute a particularly challenging regime in classical settings, where neither low-rank approximations nor structure-specific approximations (e.g., Toeplitz structures) can be exploited, and exact sampling of correlated Gaussian vectors requires a full Cholesky decomposition-based simulation, resulting in $O(N^3)$ complexity.
In such instances,
our algorithm has complexity $\widetilde{O}\left(\frac{\norm{\Sigma}_F}{\lambda_{\max}}\kappa^{1.5}\right)=\widetilde{O}(\sqrt{N})$\footnote{$\kappa=\Theta(1) \iff {}^\forall i$, $\lambda_i=\Theta(\lambda_{\max}) \implies \frac{\norm{\Sigma}_F}{\lambda_{\max}}=\sqrt{\sum_i (\frac{\lambda_i}{\lambda_{\max}})^2}=\Theta(\sqrt{N})$.}, hence offering a sextic speedup over the classical cubic-time complexity for correlated Gaussian state preparation.

If the dependence of $\kappa$ on $N$ is not too large (the case we refer to as mildly ill-conditioned), and in particular when the product\footnote{Unlike the case $\kappa=\Theta(1)$, the ratio $\frac{\norm{\Sigma}_F}{\lambda_{\max}}$ may be of strictly smaller order than $\Theta(\sqrt{N})$, due to the equivalence of matrix norms: $\lambda_{\max}=\norm{\Sigma}_2\leq \norm{\Sigma}_F \leq \sqrt{N} \norm{\Sigma}_2=\sqrt{N}\lambda_{\max}$.} $\widetilde{O}\left(\frac{\norm{\Sigma}_F}{\lambda_{\max}}\kappa^{1.5}\right)$ grows only mildly with $N$, we can still obtain a quantum speed-up, albeit weaker than in the well-conditioned case.
As a concrete application example, numerical results in \autoref{num-results-subsubsec-ns-cases} show that the covariance matrices of fractional Gaussian noises (fGNs; the increments of fBM) satisfy the condition $\frac{\norm{\Sigma}_F}{\lambda_{\max}}\kappa^{1.5}=\widetilde{\Theta}(N^{1-\varepsilon})$,
$\varepsilon\in(0,1/2]$, for $H\in(1/3,5/6)$\footnote{This corresponds to the pink line in \autoref{cost-comparison-graph-numericals} minus $1$, i.e., a preparation without taking the cumulative sum.}. 
This implies that quantum states embedding normalised fGNs can be prepared strictly faster than linear time, and our method may therefore offer a quantum advantage over the classical counterpart $O(N\log N)$ FFT-based methods.
As fractional Gaussian noises are used in the modelling of
hydrology, geophysics, climate \cite{GGWF17}, network traffic \cite{LTWW93}, among others,
our algorithm can be used for simulations of these phenomena.

We should note the issue of our algorithm that, as with any analog-encoding state-preparation quantum algorithms such as the Harrow-Hassidim-Lloyd (HHL) algorithm \cite{HHL09} and its variants, the components of the prepared state cannot be accessed directly. 
To make use of such a state, one must either feed it as input to subsequent quantum algorithms, or extract partial information or some summarised statistics of the state via measurements.
It is important to note that linear combinations of Gaussian variables are themselves Gaussian, and their distributions can be deduced straightforwardly from the covariance matrix alone. 
Therefore, if the objective is merely to evaluate linear functionals of the generated components, there is little incentive to perform simulations, even in classical settings.
However, classical outputs of a quantum circuit can be fairly non-linear in nature\footnote{Note that the normalisation procedure $\vec{x}\mapsto \vec{x}/\norm{\vec{x}}$ itself is already non-linear.}. For example, the quantum expectation with respect to an observable $A$, i.e.\ $\bra{x}A\ket{x}$, is a quadratic form in $\ket{x}$. 
In the correlated setting, the distribution of such a quantity can generally be non-trivial, and sampling may be necessary when these observables constitute the object of computational or statistical interest.
Another important direction is to apply a non-linear function $f$ to each component of the prepared state, using the non-linear transformation algorithm of \cite{RR23}, 
which we describe in more detail in \autoref{exp-non-lin-section}. 
In such a setting, one may want to compute Monte-Carlo estimates of the form
$\frac{1}{N}\sum_{i=1}^N f(\zeta_i)$ for $\zeta_i\equiv x_i/\norm{\vec{x}}$, or the $\norm{\vec{x}}$-compensated quantity $\frac{1}{N}\sum_{i=1}^N g(\zeta_i)=\frac{1}{N}\sum_{i=1}^Nf(x_i)$, where $g(\zeta)\coloneqq f(\norm{\vec{x}} \zeta)$.
These quantities are highly non-trivial, and it is usually the case that simulations are necessary.
An analysis of outputting such quantities is given in \autoref{qae-to-estimate-riemann-sum-time-int-subsec}.
Even though \autoref{exp-non-lin-section} is dedicated to the demonstration of the case of $\norm{\vec{x}}$-compensated exponentiation, i.e.\ $g(\zeta)=e^{\norm{\vec{x}}\zeta}$, the same technique developed there can be used for general functions that are well approximable by polynomials. 

\section{Numerical experiments} \label{numerical-section-label}
Although the algorithms for preparing a quantum state embedding a correlated Gaussian vector and for outputting the normalising factor proposed in \autoref{preparing-ket-x-section} have complexities depending explicitly on $N$ as $\widetilde{O}(1)$ or $\widetilde{O}(N)$, in practice, we usually have implicit dependence on $N$ through the parameters $\lVert\Sigma \rVert_F/\lambda_{\max}$ and $\kappa\equiv \lambda_{\max}/\lambda_{\min}$ as well. 
In this section, we want to investigate how these parameters scale with respect to $N$ for the cases of the major Gaussian processes used in financial modelling, which are the Riemann-Liouville fractional Brownian motion (RL-fBM), as being used in modelling rough Bergomi volatility (see e.g.\ \cite[Section 1.1]{MP18} or \cite[Section 2--4]{BFG16}), the standard Brownian motion (std-fBM), and the stationary fractional Ornstein-Uhlenbeck (fOU) process (see e.g.\ \cite[Section 3.1]{GJR18}),
by simulating the covariance matrix $\Sigma$ for each $N$ and calculating these characteristics numerically.
\subsection{Covariance structures}\label{subsec-cov-structures}
The following covariance structures of the processes under consideration are used for computing the entries of $\Sigma$.
\subsubsection{Riemann-Liouville fractional Brownian motion (RL-fBM)}
\begin{definition}
\label{defn-of-rl-fbm-bfg16-ver}
Let $H\in(0,1)$,
and $(B_u)_{u\geq 0}$ be a standard Brownian motion\footnote{A standard Brownian motion is an std-fBM of $H=1/2$, see \autoref{defn-std-fBM-num-sec}.}.
For $t\geq 0$,
define the process
$
\widetilde{W}^H_t = \sqrt{2H}\int_0^{t}(t-s)^{H-1/2}\diff B_s
$
to be a Riemann-Liouville fractional Brownian motion (RL-fBM) with Hurst index $H$.
\end{definition}
By definition, $\widetilde{W}^H$ is a centered Gaussian process.
The coefficient before the integral, here $\sqrt{2H}$, varies among different references. The current version is adopted from \cite[Section~4]{BFG16}, where the factor $\sqrt{2H}$ ensures that $\ee{\big(\widetilde{W}_1^H\big)^2}=1$.
\begin{proposition}[{\cite[Section~4]{BFG16}}] \label{cov-of-RL-fbm}
Let $v\geq u$.
The auto-covariance of $\widetilde{W}^H$ is given by 
\begin{gather*}
\ee{\widetilde{W}^H_u\,\widetilde{W}^H_v}=u^{2H}G\left(\frac vu\right),
\shortintertext{where for $x \geq 1$,}
G(x)=\frac{2H}{\frac12+H}\left(\frac{1}x\right)^{\frac12-H}{}_2F_1\left(\frac12-H,1;\frac32+H;\frac1x\right). \tagaligneq \label{eq-for-G-in-cov-of-RL-fBM}
\shortintertext{Here, ${}_2F_1$ is the hypergeometric function, defined by}
{}_2F_1(a,b;c;z)\coloneqq \sum_{n=0}^\infty \frac{(a)_n(b)_n}{(c)_n}\frac{z^n}{n!}, \tagaligneq \label{defn-of-2f1-inf-sum}
\end{gather*}
where for any $a>0$, $(a)_n\coloneqq {\Gamma(a+n)}/{\Gamma(a)}=a(a+1)\dots(a+n-1)$.
\end{proposition}
\begin{remark}
We have made a correction to \cite[Eq.~(4.1)]{BFG16} in \eqref{eq-for-G-in-cov-of-RL-fBM}, and provide a proof in the appendix, see \autoref{proof-of-cov-RL-fbm}.
\end{remark}
The RL-fBM is used in modelling the rough Bergomi variance and volatility processes, as defined below.
\begin{definition}[{\cite[Section 5.1]{FH21}} and {\cite[Section 4]{BFG16}}]
\label{rBer-S-V-defn}
The rough Bergomi (rBergomi) model asset price process $S$ and its variance process $V$ are given by
\begin{align*}
S_t&=S_0\exp\left(\int_0^t\sqrt{V_u}\diff W_u-\frac12\int_0^tV_u \diff u\right),\\
V_u&=\xi_0(u)\exp\left(
\eta
\int_0^u\frac{\sqrt{2H}}{(u-s)^{\frac12-H}}\diff B_s-\frac{\eta^2}{2}\int_0^u\frac{2H}{(u-s)^{1-2H}}\diff s\right),
\end{align*}
for constants $\eta,S_0>0$, $t\mapsto\xi_0(t)$ a continuous deterministic function defined by $\xi_0(t)\coloneqq\ee{V_t}$ called the forward variance curve, and  $W\coloneqq \rho B+\sqrt{1-\rho^2}B^{\perp}$, where $B,B^{\perp}$ are independent standard Brownian motions and $\rho\in[0,1]$ representing correlation between $W$ and $B$.
\end{definition}
\subsubsection{Standard fractional Brownian motion (std-fBM)}
\begin{definition}\label{defn-std-fBM-num-sec}
A standard fractional Brownian motion (std-fBM) with Hurst index $H\in(0,1)$ is a Guassian process $B^H$ that satisfies the following properties:
\begin{enumerate}
\item $\ee{B^H_t}=0$, for all $t\in\mathbb{R}$,
\item $\ee{B^H_tB^H_s}=\frac12(|t|^{2H}+|s|^{2H}-|t-s|^{2H})$, for all $t,s\in\mathbb{R}$.
\end{enumerate}
\end{definition}
Similar to the coefficient $\sqrt{2H}$ of \autoref{defn-of-rl-fbm-bfg16-ver}, the coefficient $1/2$ in the covariance structure can vary among references, cf.~\cite{Lim01} or \cite{Fuk23}. However, this choice ensures that $B^H$ has variance $1$ at $t=1$, and, more importantly, facilitates the connection to the definition of a stationary fractional Ornstein-Uhlenbeck process (\autoref{defn-stationary-fOU-num-sec}), where its integral representation uses this choice of std-fBM as its driver, cf.~\cite[Definition~1.1]{CKM03}.
For a more detailed treatment regarding properties of a standard fBM, refer to e.g.\ \cite[Section 1.1]{BMRS19} or \cite[Section 2.2]{Nou12}.
\subsubsection{Stationary fractional Ornstein-Uhlenbeck (fOU) process}
\begin{definition}[{\cite[Section 2]{CKM03}}]
\label{defn-stationary-fOU-num-sec}
Let $\lambda,\sigma>0$ and $H\in(0,1)$.
A stationary fractional Ornstein-Uhlenbeck (fOU) process is a Gaussian process $Y^H$ that admits a path-wise Riemann-Stieltjes integral representation of the form
\[
Y^H_t\coloneqq\sigma\int_{-\infty}^te^{-\lambda (t-u)}\diff B^H_u, \quad t\geq0,
\]
where $B^H$ is a standard fractional Brownian motion with Hurst index $H$ (see \autoref{defn-std-fBM-num-sec}).
\end{definition}
As a stationary process, the auto-covariance of $Y^H$ relies only on the difference $s$ in time, and is given as follows.
\begin{proposition}[{\cite[Remark 2.4]{CKM03}}] \label{cov-of-fou}
Let $\lambda,\sigma>0$ and $H\in(0,1)$.
For any $t,s\in\mathbb{R}$,
\[
\ee{Y^H_tY^H_{t+s}}=\sigma^2\frac{\Gamma(2H+1)\sin(\pi H)}{2\pi}\int_{-\infty}^{\infty} e^{isx}\frac{|x|^{1-2H}}{\lambda^2+x^2} \diff x.
\]
\end{proposition}
The integral on the RHS can be regarded as the Fourier transform of $f(x)=|x|^{1-2H}/(\lambda^2+x^2)$.
In order to perform numerical simulations,
we need this integral to be implementable using standard mathematical libraries.
The following lemma shows that this quantity admits a closed-form expression using the generalised hypergeometric function ${}_1F_2$.
We then implement the fOU covariance using the function \texttt{hyp1f2} from the Python library \texttt{mpmath}. 
\begin{lemma}\label{Fourier-trans-of-fOU-cov-int}
For any $\lambda>0$, $H\in(0,1)$, and $s\geq 0$, the following holds
\begin{align*}
\widehat{f}_{H,\lambda}(s)
&\coloneqq\int_{-\infty}^{\infty} e^{isx}\frac{|x|^{1-2H}}{\lambda^2+x^2} \diff x\\
&=
\frac{\pi}{\sin(\pi H)}
\left\{\lambda^{-2H} \cosh(\lambda s)
-\frac{s^{2H}}{\Gamma(2H+1)} {}_1F_2\Bigg(1;H+\frac{1}{2},H+1;\frac{\lambda^2 s^2}{4}\Bigg)\right\},
\end{align*}
where ${}_1F_2$ is the generalised hypergeometric function, defined as
\[
{}_1F_2(a;b_1,b_2;z)\coloneqq \sum_{n=0}^\infty \frac{(a)_n}{(b_1)_n(b_2)_n}\frac{z^n}{n!},
\]
where for any $a\in\mathbb{C}$ and $n\in\mathbb{N}\cup\{0\}$, $(a)_n\coloneqq {\Gamma(a+n)}/{\Gamma(a)}$.
In particular, if $H=1/2$, $\widehat{f}_{1/2,\lambda}$ reduces to the usual stationary Ornstein-Uhlenbeck covariance structure:
\[
\widehat{f}_{1/2,\lambda}(s)=\frac{\pi}{\lambda}e^{-\lambda s}.
\]
\end{lemma}
\begin{proof}
We give a detailed proof in \autoref{appndx-stn-fOU-cov-1f2}.
\end{proof}
\subsection{Setting and methodology}
\subsubsection{The covariance matrices $\Sigma^{\mathrm{pv}}$ (of path values) and $\Sigma^{\mathrm{ns}}$ (of noises)}\label{subsubsec-cov-mat-pv-ns-defn}\label{num-sec-setting-pv-vs-ns-cov-mat}
Let $T>0$ and $\mathcal{P}_{[0,T]}^N=(0=t_0<t_1<\dots<t_N=T)$ be a time discretisation of the interval $[0,T]$.
When we want to simulate a sample path of an RL-fBM or a std-fBM over this interval, it suffices to start from the time point $t_1$ since we have $\widetilde{W}^H_0=0$ and $B^H_0=0$ a.s.\ by definitions.
Therefore,
for $G$ being either $\widetilde{W}^H$ or $B^H$,
we define $\Sigma^{\mathrm{pv}}=(\Sigma^{\mathrm{pv}}_{ij})_{i,j=1}^N$, where $\Sigma^{\mathrm{pv}}_{ij}=\ee{G_{t_i}G_{t_j}}$, to be the covariance matrix of path values,
and
$\Sigma^{\mathrm{ns}}=(\Sigma^{\mathrm{ns}}_{ij})_{i,j=1}^N$, where $\Sigma^{\mathrm{ns}}_{ij}=\ee{\Delta G_{t_i}\Delta G_{t_j}}$ and $\Delta G_{t_i}=G_{t_i}-G_{t_{i-1}}$, to be the covariance matrix of noises (increments).
From this, if $Z\sim\mathcal{N}(0,I_N)$, we have $(\Sigma^{\mathrm{pv}})^{1/2} Z \stackrel{d}{=} (G_{t_1},\dots,G_{t_N})$ and similarly $(\Sigma^{\mathrm{ns}})^{1/2} Z \stackrel{d}{=} (G_{t_1},\Delta G_{t_2},\dots,\Delta G_{t_N})$, since $G_0=0$ implies $\Delta G_{t_1}=G_{t_1}$.
By taking a cumulative sum, we can recover the path values from the noises, i.e.\ $\mathcal{L}_N(\Sigma^{\mathrm{ns}})^{1/2} Z \stackrel{d}{=}(G_{t_1},\dots,G_{t_N})$.

For the case of a stationary fOU process, since $Y^H_0$ is not almost surely $0$, but instead is distributed as a Gaussian variable, we need to simulate $Y^H_0$ with a correct law as well.
Hence, for the case of a stationary fOU,
we define 
$\Sigma^{\mathrm{pv}}=(\Sigma^{\mathrm{pv}}_{ij})_{i,j=1}^{N+1}$, where $\Sigma^{\mathrm{pv}}_{ij}=\ee{Y^H_{t_{i-1}}Y^H_{t_{j-1}}}$,
and
$\Sigma^{\mathrm{ns}}=(\Sigma^{\mathrm{ns}}_{ij})_{i,j=1}^{N+1}$, where $\Sigma^{\mathrm{ns}}_{ij}=\ee{\Delta Y^H_{t_{i-1}}\Delta Y^H_{t_{j-1}}}$ and $\Delta Y^H_{t_i}=Y^H_{t_i}-Y^H_{t_{i-1}}$, with an additional convention $\Delta Y^H_{t_0}\equiv Y^H_{t_0}$.
If $Z\sim\mathcal{N}(0,I_N)$, we have $(\Sigma^{\mathrm{pv}})^{1/2} Z \stackrel{d}{=} (Y^H_{t_0},Y^H_{t_1},\dots,Y^H_{t_N})$, 
and similarly
$(\Sigma^{\mathrm{ns}})^{1/2} Z \stackrel{d}{=} (Y^H_{t_0},\Delta Y^H_{t_1},\dots,\Delta Y^H_{t_N})$.
By taking a cumulative sum, we can recover the path values from the noises, i.e.\ $\mathcal{L}_N(\Sigma^{\mathrm{ns}})^{1/2} Z \stackrel{d}{=}(Y^H_{t_0},Y^H_{t_1},\dots,Y^H_{t_N})$.
\subsubsection{Time discretisation}
We define the time discretisation for our simulation by $\widehat{\mathcal{P}}_{[0,1]}^N\coloneqq (t_i\equiv i/N)_{i=0}^N$, i.e.\ the uniform grid on $[0,1]$.
By the self-similarity property of the processes under consideration, the analysis of simulations over $[0,1]$ can be generalised to the cases of simulations over $[0,T]$ for arbitrary $T>0$, cf.\ \autoref{remark-self-sim-cost-regardless-of-T}.
\begin{remark}
Even though the algorithms proposed in \autoref{preparing-ket-x-section} can also be applied to non-uniform grids, such as $\big(t_i \equiv (i/N)^k\big)_{i=0}^N$ for $k>0$
(see \cite[Section 3.2.1]{HJT20} or \cite[Remark 7]{BDM21}),
where $0<k<1$ gives a grid clustered near $1$ and $k>1$ a grid clustered near $0$,
numerical experiments conducted for a limited selection of $k$ suggest that the behaviours of $\lambda_{\min}$, $\lambda_{\max}$, and $\lVert\Sigma \rVert_F$ also depend on the choice of $k$.
Hence, a thorough analysis of all non-uniform grids would be a non-trivial task, and we therefore restrict our attention to the uniform grid in the numerical results presented in this paper.
However, across the non-uniform grid cases examined, we observe that when the grid is given as a smooth function of $N$, the behaviours of $\lambda_{\min}$, $\lambda_{\max}$, and $\lVert\Sigma \rVert_F$ also exhibit patterns as smooth functions of $N$.
On this basis, \autoref{assump-LK-plus} appears reasonable in practice, while the correct patterns remain subject to case-by-case examination.
\end{remark}
\subsubsection{Methodology}\label{subsubsec-methodology}
The parameters $\lambda$ and $\sigma$ from \autoref{defn-stationary-fOU-num-sec} for the stationary fOU case are set to be $1$ for simplicity.
We fix increasing sequences of $N \in \mathbb{N}$ and $H \in (0,1)$.
For each fixed $H$, we numerically compute the entries of the covariance matrices $\Sigma^{\mathrm{pv}}$ and $\Sigma^{\mathrm{ns}}$ on the grid $\widehat{\mathcal{P}}_{[0,1]}^N$ for all $N$, using the covariance structures given in \autoref{subsec-cov-structures}.
The covariance between noises is obtained from the path-value covariances via the relation
$
\ee{\Delta Y_{t_i} \Delta Y_{t_j}}
=
\ee{Y_{t_i} Y_{t_j}}
- \ee{Y_{t_i} Y_{t_{j-1}}}
- \ee{Y_{t_{i-1}} Y_{t_j}}
+ \ee{Y_{t_{i-1}} Y_{t_{j-1}}}.
$
For each $N$, we compute the characteristics $\lambda_{\min}$, $\lambda_{\max}$, and $\lVert \Sigma \rVert_F$ of $\Sigma^{\mathrm{pv}}$ and $\Sigma^{\mathrm{ns}}$.
For fixed $H$, these quantities are plotted as functions of $N$.
As illustrated in \autoref{num-results-sub-sec}, each characteristic exhibits an asymptotic power-law dependence on $N$ of the form
$
y \sim A N^{p},
$
where $y$ denotes the characteristic ($\lambda_{\min}$, $\lambda_{\max}$, or $\lVert\Sigma \rVert_F$)
and the exponent $p = p(H)$ depends on $H$.
For each fixed $H$, the exponent $p$ is estimated empirically by selecting $p$ such that the graph of $y^{1/p}$ appears linear in $N$.
After gathering empirical estimates of $p$ for a range of values of $H$, we examine how $p$ depends on $H$ and deduce the relation $p=p(H)$.

\subsection{Numerical results}\label{num-results-sub-sec}
We test for $H=0.05,0.10,\dots,0.95$, and estimate $p$ using the procedure described in \autoref{subsubsec-methodology} for each characteristic. We summarise the estimated exponents for each characteristic of $\Sigma^{\mathrm{pv}}$ in \autoref{table-char-pv-case}, and of $\Sigma^{\mathrm{ns}}$ in \autoref{table-char-ns-case}.
As these exponents are obtained empirically and potential logarithmic corrections are difficult to verify over a limited range of $N$, we therefore report the observed asymptotic relations using the $\widetilde{\Theta}$ notation.

The behaviour of each characteristic relative to $N$ is plotted for the $\Sigma^{\mathrm{pv}}$ case in
\Cref{fig-pv-char-RL-fBM,fig-pv-char-std-fBM,fig-pv-char-fOU},
and for the $\Sigma^{\mathrm{ns}}$ case in
\Cref{fig-ns-char-RL-fBM,fig-ns-char-std-fBM,fig-ns-char-fOU}.
For clarity of presentation, only a representative subset of the values of $H$ is displayed in the figures.
The figures for the $\Sigma^{\mathrm{ns}}$ case are plotted over a finer range of $H$ to reflect more accurately the different behaviour of the exponent $p$ across different intervals of $H$, cf.\ \autoref{table-char-ns-case}.
The graphs are presented on a log-log scale, for which power-law behaviour corresponds to straight lines.
The solid lines represent the numerically computed values of each characteristic, while the dotted lines correspond to fitted power laws of the form $y=AN^p$, where $p=p(H)$ is the corresponding empirically estimated exponent from \Cref{table-char-pv-case,table-char-ns-case}.
The constant $A$ is defined by 
\[
A\coloneqq \left(\frac{y^{1/p}_r-y^{1/p}_\ell}{N_r-N_\ell}\right)^p,
\quad
\text{if $p\neq 0$,}
\quad
\text{and}
\quad
A\coloneqq y_r,
\quad
\text{if $p=0$}, \tagaligneq \label{defn-A-power-law-fit}
\]
where $(N_\ell,y_\ell)$ and $(N_r,y_r)$ denote the leftmost and rightmost data points of the corresponding solid line, respectively.
Note that since the coefficient $A$ corresponds to the intercept in the log-log plot, its value does not affect the claim $y=\widetilde{\Theta}(N^p)$. That is, provided that the exponent $p$ is correct, $A$ may be chosen conservatively to yield upper and lower bounds for a characteristic $y$, which is exactly what is required for a $\widetilde{\Theta}$ asymptotic relation to hold.
Therefore, the choice of $A$ in \eqref{defn-A-power-law-fit} can be regarded as being made for visual guidance.

\begin{remark}\label{remark-emp-result-support-assump}
The observed power-law scaling and empirical fits indicate that the auto-covariance structures of the RL-fBM, std-fBM, and stationary fOU admit asymptotic bounds of the form
$\lambda_{\max}=\widetilde{\Theta}(N^{p(H)})$
and
$\kappa=\widetilde{\Theta}(N^{q(H)})$,
for suitable exponents $p,q$ depending on $H$ that can be reliably measured from finite (and not too large) values of $N$, due to the stable growth/decay patterns they exhibit.
When we set 
$\widetilde{\lambda}_{\max}$ and $\widetilde{\kappa}$
in \autoref{assmp-estimate-lamb-max-kappa}
to be such power-law fits and choose the corresponding constants $A$ conservatively,
the multiplicative factors $L$ and $K$
in \autoref{assmp-estimate-lamb-max-kappa} can be taken to satisfy $L,K=\widetilde{\Theta}(1)$.
This provides empirical support for \autoref{assump-LK-plus}.
\end{remark}

As shown in \Cref{fig-pv-char-RL-fBM,fig-pv-char-std-fBM,fig-pv-char-fOU,fig-ns-char-RL-fBM,fig-ns-char-std-fBM,fig-ns-char-fOU},
the solid lines converge to the dotted lines as $N$ increases, providing empirical support for the validity of the claimed exponents.
For some values of $H$, a larger $N$ may be required before the convergence becomes visually apparent. 
In particular, values of $H$ near the boundaries and values at which the functional form of $p(H)$ changes may require larger $N$ for the solid lines to converge to the dotted lines.
Although the numerical experiments were conducted for only a limited set of values of $H$, additional out-of-sample tests at independently chosen values of $H$ indicate that the exponents reported in \Cref{table-char-pv-case,table-char-ns-case} are robust across $H\in(0,1)$.

\begin{remark}
For the stationary fOU covariance matrix $\Sigma^{\mathrm{ns}}$, the quantities
$\lambda_{\max}$ and $\lVert \Sigma^{\mathrm{ns}} \rVert_F$ are claimed to be
$\widetilde{\Theta}(1)$ for certain values of $H$
(see \autoref{table-char-ns-case} and \autoref{fig-ns-char-fOU}), although from the plots alone
it may be visually unclear whether they are indeed bounded below by a positive constant,
or could instead decay as $\widetilde{\Theta}(N^p)$ for some negative exponent $p$.
However, since $Y^H_0$ is not a constant almost surely, but follows a Gaussian distribution,
the covariance structure implies that
$\Sigma^{\mathrm{ns}}_{11} = \ee{(Y^H_0)^2} \eqqcolon \sigma_{Y^H}^2 > 0$.
Let $\widehat{e}_1 \coloneqq (1,0,\dots,0)^T$.
Then
$\lambda_{\max} \geq \widehat{e}_1^T \Sigma^{\mathrm{ns}} \widehat{e}_1 = \Sigma^{\mathrm{ns}}_{11} = \sigma_{Y^H}^2 > 0$,
and similarly
$\lVert \Sigma^{\mathrm{ns}} \rVert_F \geq |\Sigma^{\mathrm{ns}}_{11}| = \sigma_{Y^H}^2 > 0$,
confirming the existence of a strictly positive constant lower bound that is independent of $N$.
\end{remark}

\subsubsection{Results for $\Sigma^{\mathrm{pv}}$ cases}
\begin{table}[H]
\centering
\renewcommand{\arraystretch}{1.35}
\begin{tabularx}{0.7\textwidth}{|c|Y|}
\hline
\multirow{2}{*}{$\Sigma^{\mathrm{pv}}$ characteristics}
  & {RL-fBM \& std-fBM \& stationary fOU}
  \\
\cline{2-2}
& {${}^\forall\,H\in(0,1)$} \\ \hline
$\lambda_{\min}$ & 
{$\widetilde{\Theta}\left(N^{-2H}\right)$} \\ \hline
$\lambda_{\max}$ &
{$\widetilde{\Theta}\left(N\right)$} \\ \hline
$\lVert\Sigma^{\mathrm{pv}}\rVert_F$ &
{$\widetilde{\Theta}\left(N\right)$} \\ \hline
\end{tabularx}
\caption{Characteristics of the covariance matrices of path values $\Sigma^{\mathrm{pv}}$}
\label{table-char-pv-case}
\end{table}
\vspace{-20pt}
\begin{figure}[H]
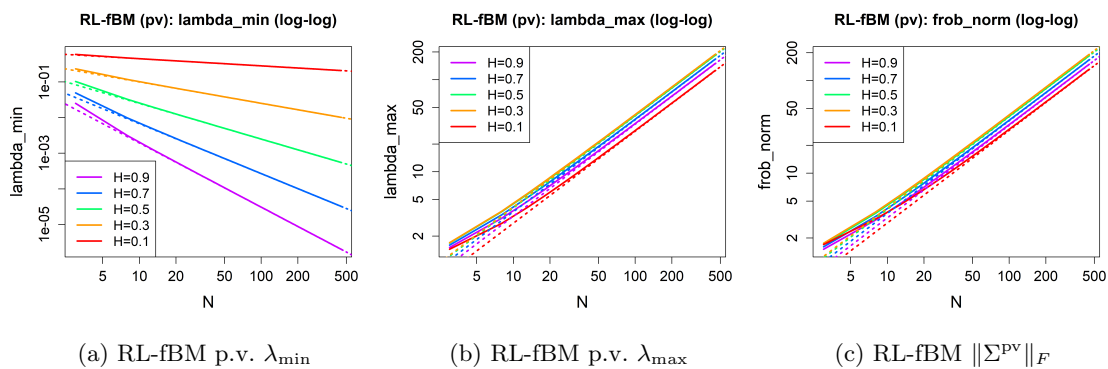

\centering
\begin{subfigure}[b]{0.33\textwidth}
    \centering
    \includegraphics[width=\textwidth]{Rplots/pv\_plot\_1-1.png}
    \caption{RL-fBM p.v.\ $\lambda_{\min}$}
\end{subfigure}\hfill
\begin{subfigure}[b]{0.33\textwidth}
    \centering
    \includegraphics[width=\textwidth]{Rplots/pv\_plot\_1-2.png}
    \caption{RL-fBM p.v.\ $\lambda_{\max}$}
\end{subfigure}\hfill
\begin{subfigure}[b]{0.33\textwidth}
    \centering
    \includegraphics[width=\textwidth]{Rplots/pv\_plot\_1-3.png}
    \caption{RL-fBM $\lVert \Sigma^{\mathrm{pv}} \rVert_F$}
\end{subfigure}
\caption{Characteristics of a Riemann-Liouville fBM covariance matrix of path values}
\label{fig-pv-char-RL-fBM}
\end{figure}
\vspace{-20pt}
\begin{figure}[H]
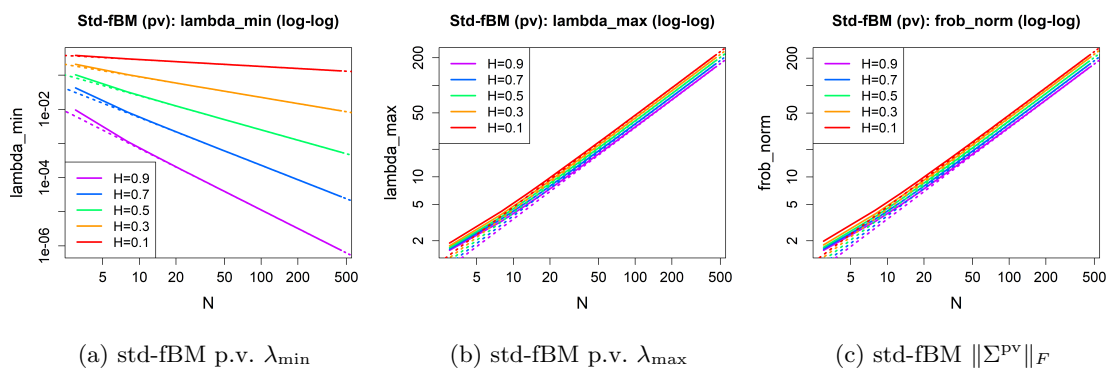

\centering
\begin{subfigure}[b]{0.33\textwidth}
\centering
\includegraphics[width=\textwidth]{Rplots/pv\_plot\_2-1.png}
\caption{std-fBM p.v.\ $\lambda_{\min}$}
\end{subfigure}\hfill
\begin{subfigure}[b]{0.33\textwidth}
\centering
\includegraphics[width=\textwidth]{Rplots/pv\_plot\_2-2.png}
\caption{std-fBM p.v.\ $\lambda_{\max}$}
\end{subfigure}\hfill
\begin{subfigure}[b]{0.33\textwidth}
\centering
\includegraphics[width=\textwidth]{Rplots/pv\_plot\_2-3.png}
\caption{std-fBM $\lVert \Sigma^{\mathrm{pv}} \rVert_F$}
\end{subfigure}
\caption{Characteristics of a standard fBM covariance matrix of path values}
\label{fig-pv-char-std-fBM}
\end{figure}
\vspace{-20pt}
\begin{figure}[H]
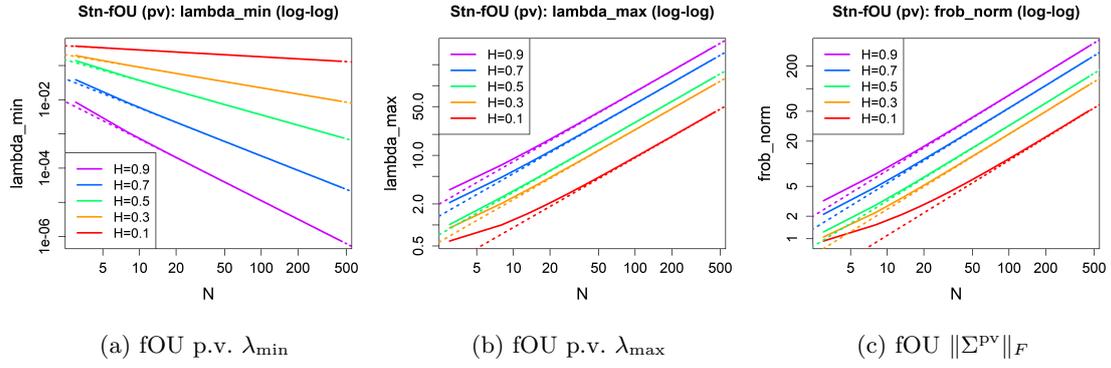

\centering
\begin{subfigure}[b]{0.33\textwidth}
\centering
\includegraphics[width=\textwidth]{Rplots/pv\_plot\_3-1.png}
\caption{fOU p.v.\ $\lambda_{\min}$}
\end{subfigure}\hfill
\begin{subfigure}[b]{0.33\textwidth}
\centering
\includegraphics[width=\textwidth]{Rplots/pv\_plot\_3-2.png}
\caption{fOU p.v.\ $\lambda_{\max}$}
\end{subfigure}\hfill
\begin{subfigure}[b]{0.33\textwidth}
\centering
\includegraphics[width=\textwidth]{Rplots/pv\_plot\_3-3.png}
\caption{fOU $\lVert \Sigma^{\mathrm{pv}} \rVert_F$}
\end{subfigure}
\caption{Characteristics of a stationary fOU covariance matrix of path values}
\label{fig-pv-char-fOU}
\end{figure}
\subsubsection{Results for $\Sigma^{\mathrm{ns}}$ cases}
\label{num-results-subsubsec-ns-cases}
\begin{table}[H]
\centering
\renewcommand{\arraystretch}{1.5}

\begin{tabularx}{0.85\textwidth}{|c|Y|Y|Y|Y|}
\hline
\multirow{2}{*}{$\Sigma^{\mathrm{ns}}$ characteristics}
  & \multicolumn{2}{c|}{RL-fBM \& std-fBM}
  & \multicolumn{2}{c|}{stationary fOU}
  \\
\cline{2-5}
& {$H \in \left(0,\frac12\right)$}
& {$H \in \left[\frac12,1\right)$}
& {$H \in \left(0,\frac12\right)$}
& {$H \in \left[\frac12,1\right)$} \\ \hline

$\lambda_{\min}$ 
& {$\widetilde{\Theta}\left(N^{-1}\right)$}
& {$\widetilde{\Theta}\left(N^{-2H}\right)$}
& {$\widetilde{\Theta}\left(N^{-1}\right)$}
& {$\widetilde{\Theta}\left(N^{-2H}\right)$} \\ \hline

$\lambda_{\max}$ 
& {$\widetilde{\Theta}\left(N^{-2H}\right)$}
& {$\widetilde{\Theta}\left(N^{-1}\right)$}
& \multicolumn{2}{c|}{$\widetilde{\Theta}\left(1\right)$}
 \\ \hline
 
\multirow{2}{*}{$\lVert\Sigma^{\mathrm{ns}}\rVert_F$}
& {$H \in \left(0,\frac34\right)$}
& {$H \in \left[\frac34,1\right)$}
& {$H \in \left(0,\frac14\right)$}
& {$H \in \left[\frac14,1\right)$} \\ \cline{2-5}
& {$\widetilde{\Theta}\left(N^{\frac12-2H}\right)$}
& {$\widetilde{\Theta}\left(N^{-1}\right)$}
& {$\widetilde{\Theta}\left(N^{\frac12-2H}\right)$}
& {$\widetilde{\Theta}\left(1\right)$}  \\ \hline
\end{tabularx}
\caption{Characteristics of the covariance matrices of noises $\Sigma^{\mathrm{ns}}$}
\label{table-char-ns-case}
\end{table}
\vspace{-20pt}
\begin{figure}[H]
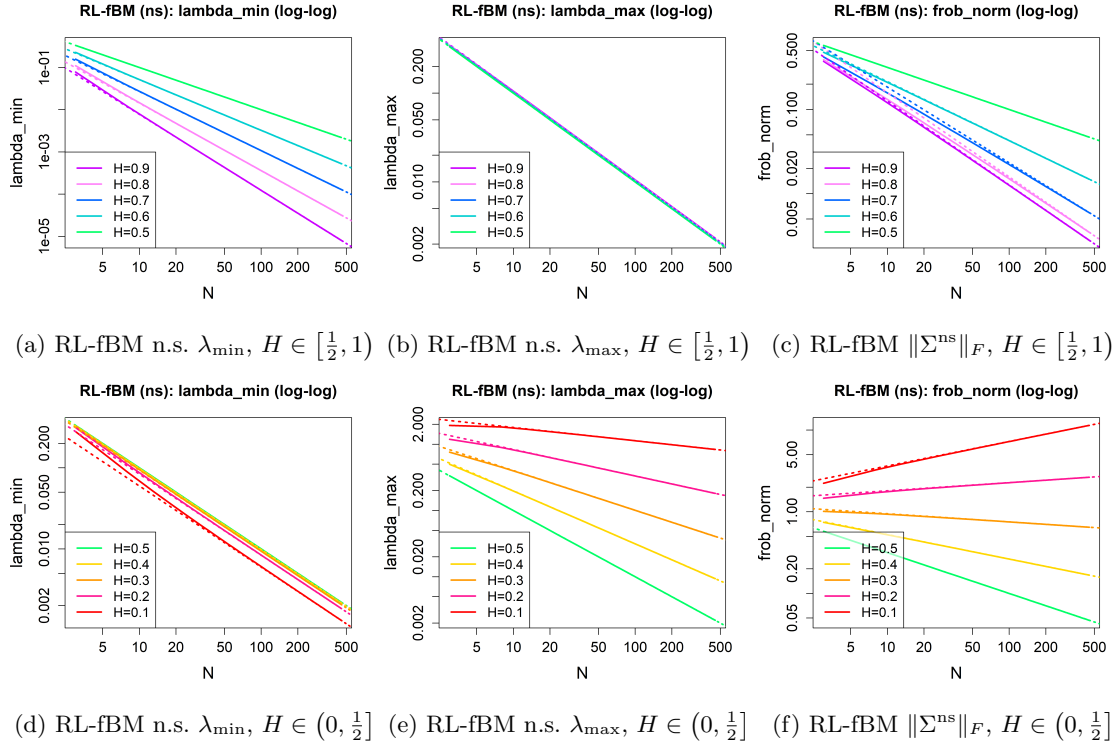

\centering
\begin{subfigure}[b]{0.33\textwidth}
    \centering
    \includegraphics[width=\textwidth]{Rplots/ns\_plot\_1-1-1.png}
    \caption{RL-fBM n.s.\ $\lambda_{\min}$, $H\in\left[\frac12,1\right)$}
\end{subfigure}\hfill
\begin{subfigure}[b]{0.33\textwidth}
    \centering
    \includegraphics[width=\textwidth]{Rplots/ns\_plot\_1-2-1.png}
    \caption{RL-fBM n.s.\ $\lambda_{\max}$, $H\in\left[\frac12,1\right)$}
\end{subfigure}\hfill
\begin{subfigure}[b]{0.33\textwidth}
    \centering
    \includegraphics[width=\textwidth]{Rplots/ns\_plot\_1-3-1.png}
    \caption{RL-fBM $\lVert \Sigma^{\mathrm{ns}} \rVert_F$, $H\in\left[\frac12,1\right)$}
\end{subfigure}\hfill
\begin{subfigure}[b]{0.33\textwidth}
    \centering
    \includegraphics[width=\textwidth]{Rplots/ns\_plot\_1-1-2.png}
    \caption{RL-fBM n.s.\ $\lambda_{\min}$, $H\in\left(0,\frac12\right]$}
\end{subfigure}\hfill
\begin{subfigure}[b]{0.33\textwidth}
    \centering
    \includegraphics[width=\textwidth]{Rplots/ns\_plot\_1-2-2.png}
    \caption{RL-fBM n.s.\ $\lambda_{\max}$, $H\in\left(0,\frac12\right]$}
\end{subfigure}\hfill
\begin{subfigure}[b]{0.33\textwidth}
    \centering
    \includegraphics[width=\textwidth]{Rplots/ns\_plot\_1-3-2.png}
    \caption{RL-fBM $\lVert \Sigma^{\mathrm{ns}} \rVert_F$, $H\in\left(0,\frac12\right]$}
\end{subfigure}
\caption{Characteristics of Riemann-Liouville fBM covariance matrix of increments}
\label{fig-ns-char-RL-fBM}
\end{figure}
\vspace{-15pt}
\begin{figure}[H]
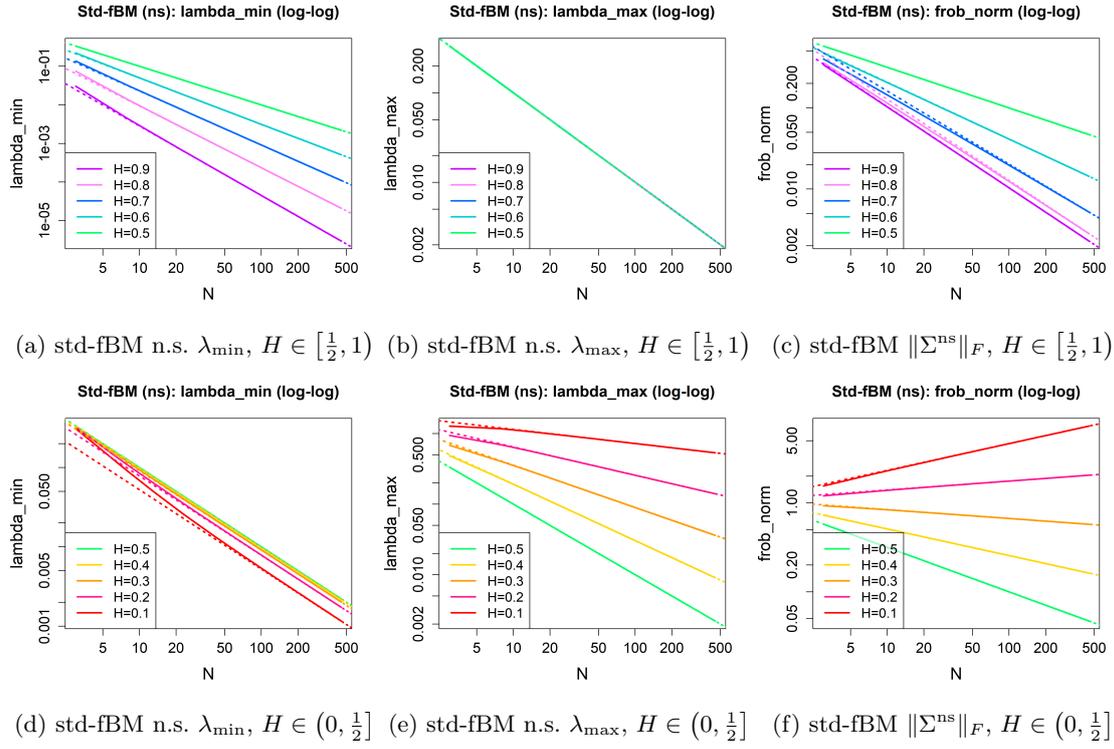

\centering
\begin{subfigure}[b]{0.33\textwidth}
\centering
\includegraphics[width=\textwidth]{Rplots/ns\_plot\_2-1-1.png}
\caption{std-fBM n.s.\ $\lambda_{\min}$, $H\in\left[\frac12,1\right)$}
\end{subfigure}\hfill
\begin{subfigure}[b]{0.33\textwidth}
\centering
\includegraphics[width=\textwidth]{Rplots/ns\_plot\_2-2-1.png}
\caption{std-fBM n.s.\ $\lambda_{\max}$, $H\in\left[\frac12,1\right)$}
\end{subfigure}\hfill
\begin{subfigure}[b]{0.33\textwidth}
\centering
\includegraphics[width=\textwidth]{Rplots/ns\_plot\_2-3-1.png}
\caption{std-fBM $\lVert \Sigma^{\mathrm{ns}} \rVert_F$, $H\in\left[\frac12,1\right)$}
\end{subfigure}\hfill
\begin{subfigure}[b]{0.33\textwidth}
\centering
\includegraphics[width=\textwidth]{Rplots/ns\_plot\_2-1-2.png}
\caption{std-fBM n.s.\ $\lambda_{\min}$, $H\in\left(0,\frac12\right]$}
\end{subfigure}\hfill
\begin{subfigure}[b]{0.33\textwidth}
\centering
\includegraphics[width=\textwidth]{Rplots/ns\_plot\_2-2-2.png}
\caption{std-fBM n.s.\ $\lambda_{\max}$, $H\in\left(0,\frac12\right]$}
\end{subfigure}\hfill
\begin{subfigure}[b]{0.33\textwidth}
\centering
\includegraphics[width=\textwidth]{Rplots/ns\_plot\_2-3-2.png}
\caption{std-fBM $\lVert \Sigma^{\mathrm{ns}} \rVert_F$, $H\in\left(0,\frac12\right]$}
\end{subfigure}
\caption{Characteristics of standard fBM covariance matrix of increments}
\label{fig-ns-char-std-fBM}
\end{figure}
\vspace{-15pt}
\begin{figure}[H]
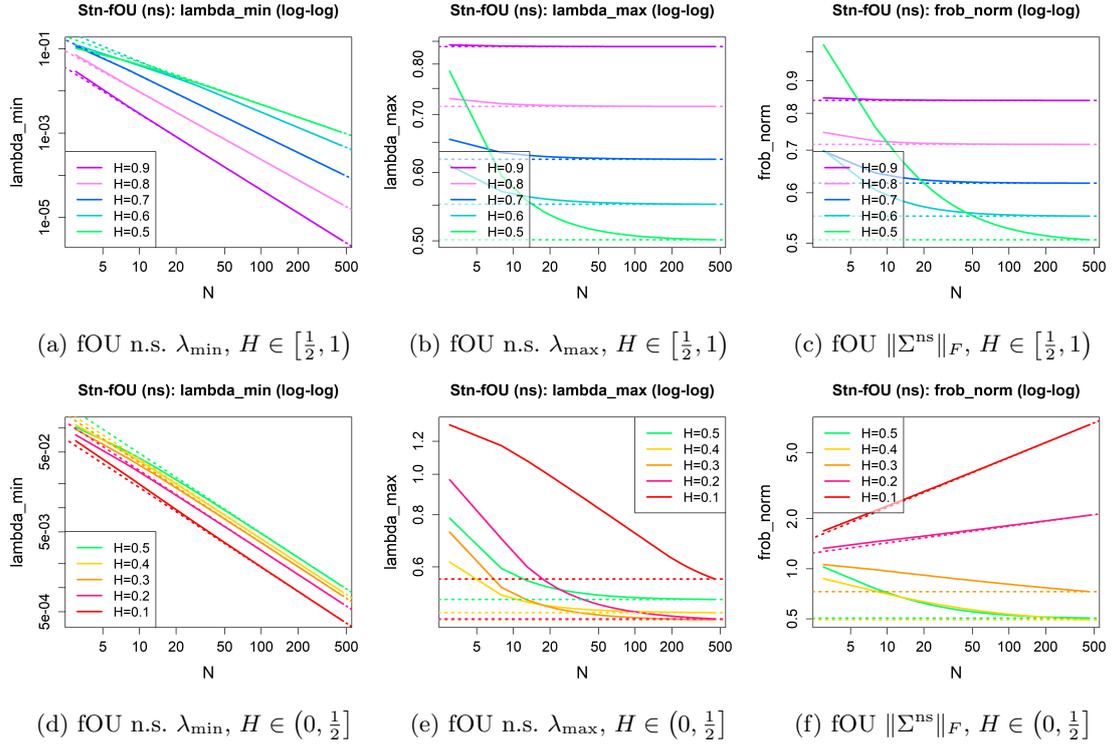

\centering
\begin{subfigure}[b]{0.33\textwidth}
\centering
\includegraphics[width=\textwidth]{Rplots/ns\_plot\_3-1-1.png}
\caption{fOU n.s.\ $\lambda_{\min}$, $H\in\left[\frac12,1\right)$}
\end{subfigure}\hfill
\begin{subfigure}[b]{0.33\textwidth}
\centering
\includegraphics[width=\textwidth]{Rplots/ns\_plot\_3-2-1.png}
\caption{fOU n.s.\ $\lambda_{\max}$, $H\in\left[\frac12,1\right)$}
\label{fig-ns-char-fOU-b}
\end{subfigure}\hfill
\begin{subfigure}[b]{0.33\textwidth}
\centering
\includegraphics[width=\textwidth]{Rplots/ns\_plot\_3-3-1.png}
\caption{fOU $\lVert \Sigma^{\mathrm{ns}} \rVert_F$, $H\in\left[\frac12,1\right)$}
\label{fig-ns-char-fOU-c}
\end{subfigure}\hfill
\begin{subfigure}[b]{0.33\textwidth}
\centering
\includegraphics[width=\textwidth]{Rplots/ns\_plot\_3-1-2.png}
\caption{fOU n.s.\ $\lambda_{\min}$, $H\in\left(0,\frac12\right]$}
\end{subfigure}\hfill
\begin{subfigure}[b]{0.33\textwidth}
\centering
\includegraphics[width=\textwidth]{Rplots/ns\_plot\_3-2-2.png}
\caption{fOU n.s.\ $\lambda_{\max}$, $H\in\left(0,\frac12\right]$}
\label{fig-ns-char-fOU-e}
\end{subfigure}\hfill
\begin{subfigure}[b]{0.33\textwidth}
\centering
\includegraphics[width=\textwidth]{Rplots/ns\_plot\_3-3-2.png}
\caption{fOU $\lVert \Sigma^{\mathrm{ns}} \rVert_F$, $H\in\left(0,\frac12\right]$}
\label{fig-ns-char-fOU-f}
\end{subfigure}
\caption{Characteristics of stationary fOU covariance matrix of increments}
\label{fig-ns-char-fOU}
\end{figure}

\subsubsection{Path simulations via $\Sigma^{\mathrm{pv}}$ and $\Sigma^{\mathrm{ns}}$: complexity comparison}\label{subsubsec-complexity-compare-pv-ns}
Recall from \autoref{main-thm-gen-mat-concrete-cost} and \autoref{main-thm-coro-concrete-cost-sum-of-inc}
that 
the gate complexity requirements 
for preparing $\ket{x}=\vec{x}/\norm{\vec{x}}$ or $\ket{y}=\vec{y}/\norm{\vec{y}}$, where $\vec{x}=(\Sigma^{\mathrm{pv}})^{1/2}\vec{z}$ and $\vec{y}=\mathcal{L}_N(\Sigma^{\mathrm{ns}})^{1/2}\vec{z}$ represent a sample path of the process of interest,
depend on the parameters $\lVert\Sigma \rVert_F/\lambda_{\max}$ and $\kappa\equiv \lambda_{\max}/\lambda_{\min}$.
In particular, the gate complexity is given by
$\widetilde{O}\left(({\lVert\Sigma\rVert_F}/{\lambda_{\max}})\kappa^{1.5}N^c\right)=\widetilde{O}\left(N^{p_1+1.5p_2+c}\right)$
when ${\lVert\Sigma\rVert_F}/{\lambda_{\max}}=\widetilde{\Theta}(N^{p_1})$
and
$\kappa=\widetilde{\Theta}(N^{p_2})$, 
for some appropriate exponents $p_1,p_2$,
with $c=0$ if $\Sigma\equiv \Sigma^{\mathrm{pv}}$ is used via \autoref{main-thm-gen-mat-concrete-cost},
and $c=1$ if $\Sigma\equiv \Sigma^{\mathrm{ns}}$ is used via \autoref{main-thm-coro-concrete-cost-sum-of-inc}.
Using
\Cref{table-char-pv-case,table-char-ns-case},
we deduce the order of $\lVert\Sigma \rVert_F/\lambda_{\max}$, $\kappa$, and the gate complexities in terms of $N$, and summarise these results 
in \Cref{table-prep-cost-pv-case,table-complexity-parameters-ns-case,table-prep-cost-ns-case}.
The gate complexity comparison between state preparation via $\Sigma^{\mathrm{pv}}$ and $\Sigma^{\mathrm{ns}}$ is illustrated in \autoref{cost-comparison-graph-numericals},
where RL-fBM and std-fBM are collectively referred to simply as fBM, since they require the same order of complexity.

As can be seen from \autoref{cost-comparison-graph-numericals},
the minimal gate complexity requirement for preparing a (normalised) sample path of an RL-fBM or a std-fBM
is achieved by using $\Sigma\equiv \Sigma^{\mathrm{pv}}$ via \autoref{main-thm-gen-mat-concrete-cost}
for $H\in(0,0.25]$, and $\Sigma\equiv \Sigma^{\mathrm{ns}}$ via \autoref{main-thm-coro-concrete-cost-sum-of-inc} for $H\in[0.25,1)$.
With this choice, the resulting complexity is at most $\widetilde{O}(N^{2.5})$ for all $H\in(0,1)$, indicating a quantum advantage over the classical $O(N^3)$ complexity of the one-time Cholesky decomposition.
Moreover, a complexity strictly lower than $\widetilde{O}(N^2)$, which is the cost of matrix-vector
multiplication in classical Cholesky-based simulations,
can be achieved for $H\in(0,1/6)$ and $H\in(2/3,5/6)$.

The stationary fOU case follows similarly, but due to the $Y^H_0$ term present in the $\Sigma^{\mathrm{ns}}$ case, 
\autoref{cost-comparison-graph-numericals} shows that the computational advantage of using $\Sigma\equiv \Sigma^{\mathrm{ns}}$ via \autoref{main-thm-coro-concrete-cost-sum-of-inc}
is reduced when compared with the fBM cases.
Sub-quadratic dependence on $N$ is still achievable for $H\in(0,1/6)$ by using
$\Sigma\equiv \Sigma^{\mathrm{pv}}$ via \autoref{main-thm-gen-mat-concrete-cost}, similarly to the fBM cases.
Preparation via 
$\Sigma\equiv \Sigma^{\mathrm{pv}}$ is preferable for $H\in(0,1/3]$,
while
$\Sigma\equiv \Sigma^{\mathrm{ns}}$ is preferable for $H\in[1/3,1)$.
For $H\in(0,2/3)$, the complexity of preparing a stationary fOU sample path is strictly less than
$\widetilde{O}(N^3)$, which is the classical complexity for a one-time Cholesky decomposition.

In order to recover the full information of the sample path, it is also necessary to output the normalising factor $\norm{\vec{x}}$ or $\norm{\vec{y}}$
corresponding to the prepared state $\ket{x}\equiv\vec{x}/\norm{\vec{x}}$ or $\ket{y}\equiv y/\norm{y}$, where $y\equiv \mathcal{L}_N\vec{x}$.
For the purpose of constructing the exponentiation algorithms (to be demonstrated in \autoref{exp-non-lin-section}), a relative-error estimate of the normalising factor obtained via \autoref{qae-w-rel-err-est} or \autoref{qae-for-norm-y}
suffices with a constant $C=\Theta(1)$
(see \autoref{non-lin-transform-concrete-cost-thm} through the use of
\autoref{Hk-using-estimate}, or \autoref{non-lin-thm-sum-of-inc} through the use of
\autoref{Hk-using-estimate-sum-of-inc}),
resulting in the same asymptotic complexity requirement (up to polylogarithmic factor) to the cost for implementing the state-preparation unitary as given in \autoref{main-thm-gen-mat-concrete-cost} or
\autoref{main-thm-coro-concrete-cost-sum-of-inc}.
Since QAE for $\norm{\vec{x}}$ or $\norm{\vec{y}}$ is needed only once and separately from the subsequent circuits, this means that the additional cost arising from QAE is additive to the algorithms that follow. 
If the state-preparation unitary $U_{\ket{x},\varepsilon}$ or $U_{\ket{y},\varepsilon}$ is being used as the building blocks for the subsequent algorithms, this QAE procedure therefore does not affect the overall asymptotic complexity required to implement that algorithm.
On the other hand, if an absolute-error estimate of the normalising factor is required,
the corresponding additive complexity becomes \eqref{qae-cost-abs-err-concrete-cost} from \autoref{qae-w-abs-err-est} or \eqref{qae-cost-sum-of-inc-abs} from \autoref{qae-for-norm-y}.
In this case, the complexity may be reduced due to the change in dependence from
$\kappa^{1.5}$ to $\sqrt{\lambda_{\max}}\,\kappa$, with the improvement arising from
the disappearance of the $\lambda_{\min}^{-1/2}$ factor
(which is always $\widetilde{\Theta}(N^p)$ for a strictly positive $p$ in all cases considered).
However, this potential improvement may be offset by additional multiplicative dependence on
$\norm{\vec{z}}_2$, which is concentrated around $O(\sqrt{N})$, as well as on the
precision parameter $1/\widehat{\varepsilon}$.
Since absolute-error estimates are not used in subsequent sections, we do not pursue a
detailed analysis here and regard this case as application-dependent.

\begin{table}[H]
\centering
\renewcommand{\arraystretch}{1.35}
\begin{tabularx}{0.7\textwidth}{|c|Y|}
\hline
{complexity parameters}
  & {RL-fBM \& std-fBM \& stationary fOU}
  \\
\cline{2-2}
($\Sigma^{\mathrm{pv}}$ cases) & {${}^\forall\,H\in(0,1)$} \\ \hline
$\kappa$ & 
{$\widetilde{\Theta}\left(N^{1+2H}\right)$} \\ \hline
$\lVert\Sigma^{\mathrm{pv}}\rVert_F/\lambda_{\max}$ &
{$\widetilde{\Theta}\left(1\right)$} \\ \hline
{preparing $\ket{x}$ via $\Sigma^{\mathrm{pv}}$} &
{$\widetilde{O}\left(N^{1.5+3H}\right)$} \\ \hline
\end{tabularx}
\caption{Complexity's dependence on $N$ for preparing a sample path via $\Sigma^{\mathrm{pv}}$}
\label{table-prep-cost-pv-case}
\end{table}
\vspace{-15pt}
\begin{table}[H]
\centering
\renewcommand{\arraystretch}{1.5}

\begin{tabularx}{\textwidth}{|c|Y|Y|Y|Y|Y|}
\hline
{complexity parameters}
  & \multicolumn{3}{c|}{RL-fBM \& std-fBM}
  & \multicolumn{2}{c|}{stationary fOU}
  \\
\cline{2-6}
($\Sigma^{\mathrm{ns}}$ cases) 
& \multicolumn{3}{c|}{{${}^\forall H \in (0,1)$}}
& {$H \in \left(0,\frac12\right)$}
& {$H \in \left[\frac12,1\right)$} \\ \hline

$\kappa$ 
& \multicolumn{3}{c|}{$\widetilde{\Theta}\left(N^{\lvert 1-2H\rvert}\right)$}
& {$\widetilde{\Theta}\left(N\right)$} 
& {$\widetilde{\Theta}\left(N^{2H}\right)$}\\ \hline

\multirow{2}{*}{$\displaystyle\frac{\lVert\Sigma^{\mathrm{ns}}\rVert_F}{\lambda_{\max}}$}
& {$H\in\left(0,\frac12\right)$}
& {$H\in\left[\frac12,\frac34\right)$}
& {$H\in\left[\frac34,1\right)$}
& {$H \in \left(0,\frac14\right)$}
& {$H \in \left[\frac14,1\right)$}\\ \cline{2-6}
& {$\widetilde{\Theta}\left(N^{\frac12}\right)$}
& {$\widetilde{\Theta}\left(N^{\frac32-2H}\right)$}
& {$\widetilde{\Theta}\left(1\right)$}
& {$\widetilde{\Theta}\left(N^{\frac12-2H}\right)$}
& {$\widetilde{\Theta}\left(1\right)$}
 \\ \hline
\end{tabularx}
\caption{Complexity parameters' dependence on $N$ ($\Sigma^{\mathrm{ns}}$ cases)}
\label{table-complexity-parameters-ns-case}
\end{table}
\vspace{-15pt}
\begin{table}[H]
\centering
\renewcommand{\arraystretch}{1.5}
\begin{tabularx}{0.75\textwidth}{|c|Y|Y|Y|}
\hline
\multirow{3}{*}{
\shortstack[c]{preparing $\ket{y}$\\[4pt] via $\Sigma^{\mathrm{ns}}$}
}
& \multicolumn{3}{c|}{RL-fBM \& std-fBM}
  \\
\cline{2-4}
& {$H\in\left(0,\frac12\right)$}
& {$H\in\left[\frac12,\frac34\right)$}
& {$H\in\left[\frac34,1\right)$}
\\ \cline{2-4}
& {$\widetilde{O}\left(N^{3-3H}\right)$}
& {$\widetilde{O}\left(N^{1+H}\right)$}
& {$\widetilde{O}\left(N^{-\frac12+3H}\right)$}
\\ \hline
\end{tabularx}\\[3pt]
\begin{tabularx}{0.75\textwidth}{|c|Y|Y|Y|}
\hline
\multirow{3}{*}{
\shortstack[c]{preparing $\ket{y}$\\[4pt] via $\Sigma^{\mathrm{ns}}$}
}
& \multicolumn{3}{c|}{stationary fOU}
  \\
\cline{2-4}
& {$H\in\left(0,\frac14\right)$}
& {$H\in\left[\frac14,\frac12\right)$}
& {$H\in\left[\frac12,1\right)$}
\\ \cline{2-4}
& \multicolumn{1}{Y|}{$\widetilde{O}\left(N^{3-2H}\right)$}
& \multicolumn{1}{Y|}{$\widetilde{O}\left(N^{\frac52}\right)$}
& \multicolumn{1}{Y|}{$\widetilde{O}\left(N^{1+3H}\right)$}
\\ \hline
\end{tabularx}
\caption{Complexity's dependence on $N$ for preparing a sample path via $\Sigma^{\mathrm{ns}}$}
\label{table-prep-cost-ns-case}
\end{table}
\vspace{-20pt}
\begin{figure}[H] 
\centering
\includegraphics[width=0.55\textwidth]{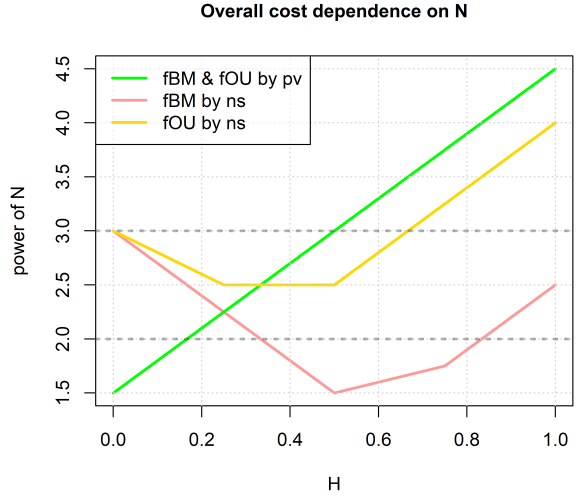}\vspace{-10pt}
\caption{Comparison of overall cost required, when simulate using $\Sigma^{\mathrm{pv}}$ v.s.\ $\Sigma^{\mathrm{ns}}$}
\label{cost-comparison-graph-numericals}
\end{figure}

\section{Exponentiation and simulation of rBergomi variance}\label{exp-non-lin-section}
In this section, we show how to prepare a quantum embedding of an exponentiated Gaussian vector, including, as special cases, the sample path of an exponentiated Gaussian process, with the rBergomi variance/volatility process as a concrete example.
The main result of this section (\autoref{preparing-ket-exp-x-from-ket-x-w-error}, \autoref{non-lin-transform-concrete-cost-thm}, \autoref{non-lin-thm-sum-of-inc}) is the preparation of a state, given in the most abstract form as
\[
\ket{\vec{f}\odot e^{c\vec{x}}}
\equiv
\frac{\vec{f}\odot e^{c\vec{x}}}{\norm{\vec{f}\odot e^{c\vec{x}}}}
\coloneqq
\frac{1}{\sqrt{\sum_{i=1}^N |f_i|^2\cdot e^{2c x_i}}}
\sum_{i=1}^N f_i \cdot e^{c x_i} \ket{i}, \tagaligneq \label{exponentiation-abstract-form}
\]
for $c\in\mathbb{R}\backslash\{0\}$ and $\vec{f}\in\mathbb{R}^N\backslash\{\vec{0}\}$,
where we use the following notation.
\begin{notation}[Hadamard product]
For $\vec{u}=(u_1,\dots,u_N)\in \mathbb{R}^N$ and $\vec{v}=(v_1,\dots,v_N)\in \mathbb{R}^N$, we denote the component-wise product, aka the Hadamard product, by
\[
\vec{u}\odot\vec{v}\coloneqq (u_1v_1,\dots,u_Nv_N)\in \mathbb{R}^N.
\]
\end{notation}
This form includes exponentiated Gaussian state preparations ranging from the most basic cases, such as
$\ket{e^{\vec{x}}}=e^{\vec{x}}/\lVert e^{\vec{x}}\rVert$ (where we let ${}^\forall i$, $f_i\equiv 1$, and $c\equiv 1$ in \eqref{exponentiation-abstract-form}),
to more complex cases in which we require a coefficient in front of the exponential as well as a translation in the exponent (which can also be reduced to another multiplicative coefficient).
The latter is precisely the form taken by the rBergomi variance/volatility process. For more details, see \autoref{rBer-ket-example1}.

Performing QAE on the abstract form \eqref{exponentiation-abstract-form} also allows us to output an estimate of a discrete sum of the form
$\frac1N\sum_{i=1}^N f_i e^{c x_i}$
(see \autoref{QAE-for-time-int-sum-gen-f}, \autoref{QAE-for-time-int-sum-gen-f-sum-of-inc}). 
This includes basic discrete time integrals of exponentiated Gaussian processes, e.g.\ $\frac{1}{N}\sum_{i=1}^N e^{x_i}$, 
as well as the norm $\lVert e^{\vec{x}}\rVert=\sqrt{\sum_{i=1}^N e^{2 x_i}}$, which can be obtained by letting $f_i \equiv N$ for all $i$ (though this can be costly, see \autoref{output-sqrt-of-discrete-sum}).
The discrete sum approximating the integrated rBergomi variance $\int_0^T V_t \diff t$ can also be obtained by the same procedure (see \autoref{time-int-approx-by-QAE}). 
The integrated variance plays a crucial role in the characterisation of the rough Bergomi price process conditioned on a given variance sample path, cf.\ \cite[Equation (2.3)]{MP18}, and also serves as a continuous-time approximation to the realised variance, which is an important measure of price variation and related to volatility forecasting in real market.

\subsection{Implementation of the exponential transformation}
A unitary $U_{\ket{x},\varepsilon}$ as in \autoref{main-thm-for-gen-mat} that prepares a quantum state $\ket{\widetilde{\varphi}}_{\mathrm{Q}\mathrm{I}}=U_{\ket{x},\varepsilon}\ket{0}_{\mathrm{Q}\mathrm{I}}$, where $\norm{\ket{\widetilde{\varphi}}_{\mathrm{Q}\mathrm{I}}-\ket{0}_{\mathrm{Q}}\ket{x}_{\mathrm{I}}}_2\leq\varepsilon$ for some $\varepsilon>0$,
is called a state-preparation block-encoding (SPBE) in the work of \cite{RR23}.
The authors of \cite{RR23} showed that one can create a quantum state embedding a non-linear transformation, including exponentiation, of the amplitudes of $\ket{x}$ by using $U_{\ket{x},\varepsilon}$ as a building-block, see \cite[Appendix B]{RR23}.
Their technique is based on the framework of quantum singular value transformation (QSVT), cf.\ \autoref{lem-qsvt}, which allows for the implementation of a polynomial transformation of the singular values of a block-encoded matrix. 
For the purpose of a non-linear transformation of amplitudes, the target of QSVT is the block-encoding of the diagonal matrix $\mathrm{diag}(\ket{x})$, whose singular values coincide with the amplitudes of $\ket{x}$.
However, we cannot directly apply their results to our situation due to the following two issues. 

The first issue is that the algorithm for the exponentiation (see \cite[Theorem 10, Theorem 5]{RR23}) has complexity linearly depending on the normalising factor $\gamma_f=\sup_{x\in[-1,1]}|a^x|=a$ for any $a>0$, or more precisely on $\gamma_{P_k}=\sup_{x\in[-1,1]}|P_k(x)|$, where $P_k$ is the polynomial that approximates $f(x)=a^x$ (see \cite[Theorem 9]{RR23}).
In our current situation, we want to perform the mapping $\ket{x}\mapsto e^{\norm{\vec{x}}_2 \ket{x}}\equiv e^{\vec{x}}$, 
so, when $\vec{x}$ represents a sample path of a continuous Gaussian process, $\gamma_f=e^{\norm{\vec{x}}_2}\approx O(e^{\sqrt{N}})$ with high probability, cf.\ \autoref{gaussian-sample-path-2-norm-sqrt-N}.
This renders the algorithm too costly to be practical when applied to our case.

The second issue is that \cite[Appendix B]{RR23} prepares the non-linearly transformed quantum state in a way that the block-encoding ancillae are implicitly getting transformed as well, see Lemma 19 and Theorem 11 therein. 
Their method would produce the desired transformation if the non-linear function of transformation $f$ satisfies $f(0)=0$, such that we have $f(\ket{0}_{\mathrm{Q}}\ket{x}_{\mathrm{I}})=\ket{0}_{\mathrm{Q}}f(\ket{x}_{\mathrm{I}})$. 
However, this is not the case for the exponential function $a^x$ as $a^0=1\neq0$ for any $a>0$. 
Consequently, their results cannot be applied as is in our setting.

Our approach to solve the first issue is to use the observation that $\norm{\vec{x}}_2\to\infty$ as $N\to \infty$, but $\norm{\vec{x}}_{\infty}$ stays bounded by a fixed constant $\Xi>0$ with high probability regardless of $N$ (see the discussion prior to \autoref{assump-inf-norm-bound-for-vec-x}). 
This means that the amplitudes of $\ket{x}$ lie in an interval that becomes smaller as $N$ increases. 
Therefore, our function of transformation does not need to perform the exponential map on the whole interval $[-1,1]$, but only over the interval $\left[-\frac{\Xi}{\norm{\vec{x}}_2},\frac{\Xi}{\norm{\vec{x}}_2}\right]\supset\left[-\frac{\norm{\vec{x}}_\infty}{\norm{\vec{x}}_2},\frac{\norm{\vec{x}}_\infty}{\norm{\vec{x}}_2}\right]$. 
The values of the function of transformation outside this interval can be anything bounded.
We show in \autoref{poly-approx-Hk-to-h} how to construct a polynomial that realises such a function, with the property that it is bounded by a constant regardless of $N$ (see \eqref{gamma-H-upp-bnd-claim} and \eqref{gamma-H-upp-bnd-claim-w-eta}).
This allows us to perform the exponentiation procedure with the normalising factor $\gamma_f$ (or $\gamma_{\mathcal{H}^{(\widetilde{\eta})}_k}$ in our case) being $O(1)$, hence removing its asymptotic contribution to the algorithm's overall complexity, similarly to the exponentiation result of \cite[Theorem 5]{RR23}. A detailed treatment for the construction of this polynomial starts from \autoref{subsubsec-polynomial-of-transformation-construction}.

For the second issue, even though \cite{RR23} did not explicitly give a thorough analysis of the case $f(0)\neq 0$ for SPBE, their technique can still be adopted to address this with only a slight modification (see the discussion around \eqref{expand-h-D}).

\subsubsection{Construction of the polynomial of transformation}\label{subsubsec-polynomial-of-transformation-construction}
To construct a mapping $\widetilde{h}$ on $[-1,1]$ that performs $\zeta\mapsto e^{\norm{\vec{x}}_2\zeta}$ on $\left[-\frac{\Xi}{\norm{\vec{x}}_2},\frac{\Xi}{\norm{\vec{x}}_2}\right]$ and is bounded elsewhere, the key idea is to decouple the exponential growth from the factor $\norm{\vec{x}}_2$.
To do this, we decompose $\widetilde{h}$ on the interval $\left[-\frac{\Xi}{\norm{\vec{x}}_2},\frac{\Xi}{\norm{\vec{x}}_2}\right]$ as
$
\widetilde{h}(\zeta)=e^{\norm{\vec{x}}_2\zeta}=e^{\widetilde{C} \widetilde{g}(\zeta)}\equiv (\widetilde{f}\circ\widetilde{g})(\zeta)
$,
where $\widetilde{C}>0$, $\widetilde{f}(\zeta)\coloneqq e^{\widetilde{C}\zeta}$, 
and 
$\widetilde{g}(\zeta)$ is chosen to be approximately linear with slope $\norm{\vec{x}}_2/\widetilde{C}$ on $\left[-\frac{\Xi}{\norm{\vec{x}}_2},\frac{\Xi}{\norm{\vec{x}}_2}\right]$, while remaining bounded on $[-1,1]$ by some constant $M>0$.
Intuitively, 
even though
$\widetilde{h}$ grows increasingly rapidly as $\norm{\vec{x}}_2$ becomes larger, 
this growth is effectively confined to the shrinking interval $\left[-\frac{\Xi}{\norm{\vec{x}}_2},\frac{\Xi}{\norm{\vec{x}}_2}\right]$.
Once outside this interval, the boundedness of $\widetilde{g}$ prevents further increase, so that
$
\gamma_h=\sup_{\zeta\in[-1,1]}|\widetilde{h}(\zeta)|\leq e^{\widetilde{C}M}
$
independently of $\norm{\vec{x}}_2$.

Fortunately, there exists a polynomial approximation to such a function $\widetilde{g}$, see \autoref{poly-lin-lc17-1} provided below. This polynomial is widely used in the so-called linear amplitude amplification algorithm, see e.g.\ \cite[Appendix D(6.)]{MRTC21} and FIG.\ 26 therein for an illustration.

Similarly, for $c\neq 0$, we can define $h(\zeta)=e^{c\widetilde{C}\widetilde{g}(\zeta)}$, by redefining $\widetilde{f}(\zeta)\coloneqq e^{\widetilde{C}\zeta}$ as $\widetilde{f}(\zeta)\coloneqq e^{C\zeta}$ with $C\equiv c\,\widetilde{C}$. This allows arbitrary powers of the original exponentiation\footnote{This modification needs to be made in $\widetilde{f}$, not in $\widetilde{g}$, to satisfy the slope constraint of \autoref{poly-lin-lc17-1}.}, i.e.\ $h= \widetilde{h}^c$.
\begin{lemma}[{\cite[Appendix A, Theorem 10]{LC17}}]\label{poly-lin-lc17-1}
${}^\forall \Gamma\in\left(0,\frac12\right]$, $\varepsilon_{\mathrm{lin}}\in\left(0,O(\Gamma)\right]$, there exists\footnote{\cite{LC17}'s proof is contructive, meaning that one knows how to determine $P_{\mathrm{lin},\Gamma,n}(x)$ for given $\Gamma,\varepsilon_{\mathrm{lin}}$.} an odd polynomial $P_{\mathrm{lin},\Gamma,n}$ of degree $n=O\left(\Gamma^{-1}\log\left(\frac{1}{\varepsilon_{\mathrm{lin}}}\right)\right)$ such that 
\begin{gather}
{}^\forall x\in[-\Gamma,\Gamma],\quad \left|P_{\mathrm{lin},\Gamma,n}(x)-\frac{x}{2\Gamma}\right|
\leq
\frac{\varepsilon_{\mathrm{lin}} |x|}{2\Gamma},  \label{ineq-for-p-lin}
\shortintertext{and}
\max_{x\in[-1,1]}\left|P_{\mathrm{lin},\Gamma,n}(x)\right|
\leq 1.  \label{bound-for-p-lin}
\end{gather}
\end{lemma}
We provide following lemmas for the polynomial approximation of the function $\widetilde{f}(\zeta)\coloneqq e^{C\zeta}$, $C\neq 0$, and about the error bound for the composite function $h=\widetilde{f}\circ \widetilde{g}$.
\begin{lemma}[Stirling's formula bounds]\label{Stirling-bounds2}
For all $n\in\mathbb{N}$,
\[
\sqrt{2\pi n}\left(\frac{n}{e}\right)^n e^{\frac{1}{12n+1}} < n!<\sqrt{2\pi n}\left(\frac{n}{e}\right)^n e^{\frac{1}{12n}}.
\]
\end{lemma}
\begin{proof}
These bounds are due to Robbins \cite[equations (1), (2)]{Rob55}, combined with the fact that an exponential function is increasing in its argument.
\end{proof}

\begin{lemma}\label{poly-degree-for-exp2}
Let $R,\varepsilon>0$ and $k\coloneqq\max\left\{\ceil*{e^2R},\ceil*{\log\left({1}/{\varepsilon}\right)}\right\}$. The exponential function $f(x)=e^x$ on the interval $[-R,R]$ can be approximated uniformly to within $\varepsilon$ by its Taylor polynomial of degree $k$, $P_k(x)=\sum_{m=0}^kx^m/m!$.
\end{lemma}
\begin{proof}
Consider $m\geq k+1$.
Observe that by \autoref{Stirling-bounds2}, we have
\[
\sqrt[m]{m!}
\geq
\left(2\pi m\right)^{\frac{1}{2m}}\cdot \frac{m}{e} \cdot e^{\frac{1}{12m^2+m}}
\geq
\frac{m}{e}
\geq
\frac{k}{e}
\geq
eR. \tagaligneq \label{eq1-exp-taylor-poly-lem}
\]
The approximation error by the degree-$k$ Taylor polynomial 
over $[-R,R]$
is given by
\begin{align*}
\sup_{x\in[-R,R]}
\big| e^x-P_k(x) \big| 
\leq
\sup_{x\in[-R,R]}\sum_{m=k+1}^\infty \frac{|x|^m}{m!}
\leq 
\sum_{m=k+1}^\infty \left(\frac{R}{\sqrt[m]{m!}}\right)^m
\stackrel{\eqref{eq1-exp-taylor-poly-lem}}{\leq}
\sum_{m=k+1}^\infty e^{-m}
=
\frac{e^{-k}}{e-1}.
\end{align*}
As $k
\geq \log\left({1}/{\varepsilon}\right)$, we have
$-k \leq \log(\varepsilon) \leq \log(\varepsilon)+\log(e-1)$, and the claim is proved.
\end{proof}

\begin{lemma}\label{poly-degree-ax-3}
Let $a,\varepsilon>0$ and
$k\coloneqq\max\left\{\ceil*{e^2 \left|\log a\right|},\ceil*{\log\left({1}/{\varepsilon}\right)}\right\}$.
The exponential function $f(x)=a^x$ on the interval $[-1,1]$ can be approximated uniformly to within $\varepsilon$ by its Taylor polynomial of degree $k$, $Q_k(x)=\sum_{m=0}^k (\log a)^mx^m/m!$.
\end{lemma}
\begin{proof}
Letting $R\coloneqq \left|\log a\right|$ and $k\coloneqq\max\left\{\ceil*{e^2 R},\ceil*{\log\left({1}/{\varepsilon}\right)}\right\}=\max\left\{\ceil*{e^2 \left|\log a\right|},\ceil*{\log\left({1}/{\varepsilon}\right)}\right\}$,
we can use \autoref{poly-degree-for-exp2} to construct $P_k$ such that
$
\sup_{x\in [-R,R]}\left|e^x-P_k(x)\right|\leq \varepsilon$ holds.
Define a degree-$k$ polynomial $Q_k(x)\coloneqq P_k((\log a)x)$, and denote $y\coloneqq (\log a)x$. 
Since $x\in[-1,1]$ implies $y \in [-R,R]$, we have
$
\sup_{x\in [-1,1]}\left|a^x-Q_k(x)\right|=
\sup_{y\in [-R,R]}\left|e^{y}-P_k(y)\right|
\leq\varepsilon$.
\end{proof}

\begin{lemma}\label{composite-poly-lem}
Let $\Gamma>0$. 
Suppose a degree-$k_1$ polynomial $P_{k_1}$ can approximate a function $f$ uniformly on $[-1,1]$ to within $\varepsilon_f>0$, and a degree-$k_2$ polynomial $Q_{k_2}$ can approximate a function $g$ uniformly on $[-\Gamma,\Gamma]$ to within $\varepsilon_g>0$. Suppose further that for all $x\in[-\Gamma,\Gamma]$, $g(x)\in[-1,1]$ and $Q_{k_2}(x)\in[-1,1]$. Denote by $L_P=\sup_{x\in[-1,1]}|P_{k_1}'(x)|$ a Lipschitz constant for $P_{k_1}$ on $[-1,1]$. 
Then, the composite polynomial $P_{k_1}\circ Q_{k_2}$ of degree $k_1k_2$ satisfies
\[
\sup_{x\in[-\Gamma,\Gamma]}\left|\big(f\circ g\big)(x)-\big(P_{k_1}\circ Q_{k_2}\big)(x)\right|
\leq
\varepsilon_f+L_P\, \varepsilon_g.
\]
\end{lemma}
\begin{proof}
The fact that $P_{k_1}\circ Q_{k_2}$ has degree $k_1k_2$ is trivial.
To prove the error bound, observe that
\begin{align*}
&\sup_{x\in[-\Gamma,\Gamma]}\left|\big(f\circ g\big)(x)-\big(P_{k_1}\circ Q_{k_2}\big)(x)\right|\\
&\leq \sup_{x\in[-\Gamma,\Gamma]}\left|f\left(g(x)\right)-P_{k_1}\left(g(x)\right)\right|
+ \sup_{x\in[-\Gamma,\Gamma]}\left|P_{k_1}\left(g(x)\right)-P_{k_1}\left(Q_{k_2}(x)\right)\right|\\
&\leq \sup_{g(x)\in[-1,1]}\left|f\left(g(x)\right)-P_{k_1}\left(g(x)\right)\right|
+ L_P\sup_{x\in[-\Gamma,\Gamma]}\left|g(x)-Q_{k_2}(x)\right|\\
&\leq \varepsilon_f
+ L_P\, \varepsilon_g. \qedhere
\end{align*}
\end{proof}
Though it is not assumed in the following proposition, we can think of $\varkappa$ as the variable into which we want to plug the value $\norm{\vec{x}}_2$. 
\begin{proposition}\label{poly-approx-Hk-to-h}
Let $c\in\mathbb{R}\backslash\{0\}$ and $\Xi>0$ be fixed constants.
Suppose we are given $\varkappa$ such that $0<2\Xi\leq\varkappa$.
Then, for any $\varepsilon\in\left(0,2e^{|c|(2\Xi)}\right]$, we can construct a degree-$k$ polynomial $\mathcal{H}_k^{(\varkappa)}$, where $k=O(\varkappa\log^2\left(\frac{\varkappa}{\varepsilon}\right))$, satisfying
\begin{gather}
\sup_{\zeta\in \left[-\frac{\Xi}{\varkappa},\frac{\Xi}{\varkappa}\right]}\left|
\mathcal{H}_k^{(\varkappa)}(\zeta)-e^{c\varkappa\cdot\zeta}
\right|
\leq \varepsilon, \label{approx-precision-claim}
\shortintertext{and}
\gamma_{\mathcal{H}_k^{(\varkappa)}}\coloneqq\sup_{\zeta\in[-1,1]}\left|
\mathcal{H}_k^{(\varkappa)}(\zeta)
\right|
\leq 2e^{|c|(2\Xi)} =O(1). \label{gamma-H-upp-bnd-claim}
\end{gather}
\end{proposition}
\begin{proof}
Denote 
$f(x)\coloneqq e^{c(2\Xi) x}$ 
and 
$g(x)\coloneqq \frac{\varkappa}{2\Xi} x$.
We have that $(f \circ g)(x) = e^{c\varkappa\cdot x}$ is the target function we want to approximate.
Let
$\varepsilon_f\coloneqq \varepsilon/2$
and
$\varepsilon_g\coloneqq \varepsilon/(|c|(2\Xi) e^{|c|(2\Xi)})$.
By \autoref{poly-degree-ax-3},
$P_{k_1}(x)\equiv\sum_{n=0}^{k_1} {\left(c(2\Xi)\right)^n}x^n/{n!}$, where $k_1=\max\left\{\ceil*{e^2\cdot |c|(2\Xi)},\ceil*{\log\left({1}/{\varepsilon_f}\right)}\right\}$, can approximate $f$ uniformly on $[-1,1]$ to within $\varepsilon_f$.
Besides, according to \autoref{poly-lin-lc17-1}, with $\Gamma\coloneqq {\Xi}/{\varkappa}\in\left(0,1/2\right]$, $\varepsilon_{\mathrm{lin}}\coloneqq \varepsilon_g \cdot(2\Gamma)=\varepsilon_g \cdot(2\Xi/\varkappa)\in\left(0,O(\Gamma)\right]$, and $k_2=O\left(\Gamma^{-1}\log\left(\frac{1}{\varepsilon_{\mathrm{lin}}}\right)\right)=O\left(\varkappa\log\left(\frac{\varkappa}{\varepsilon_g}\right)\right)$,
$Q_{k_2}\equiv P_{\mathrm{lin},\frac{\Xi}{\varkappa},k_2}$ approximates $g$ uniformly on $\left[-\frac{\Xi}{\varkappa},\frac{\Xi}{\varkappa}\right]$ to within\footnote{Note that since \eqref{ineq-for-p-lin} holds for $|x|\leq \Gamma\in\left(0,\frac12\right]$, it is true that $\left|P_{\mathrm{lin},\Gamma,k_2}(x)-\frac{x}{2\Gamma}\right|
\leq
\frac{\varepsilon_{\mathrm{lin}} |x|}{2\Gamma}
=
\varepsilon_g |x|
\leq
\frac{\varepsilon_g}{2}$.} 
$\varepsilon_g/2$.
Note that we have
\[
L_P\coloneqq \sup_{x\in[-1,1]}\left|P_{k_1}'(x)\right|
\leq 
|c|(2\Xi) \sum_{n=0}^{k_1-1} \frac{\left(|c|(2\Xi)\right)^n}{n!}\leq |c|(2\Xi) e^{|c|(2\Xi)},
\]
and
${}^\forall x\in \left[-\frac{\Xi}{\varkappa},\frac{\Xi}{\varkappa}\right]$, $|g(x)|\leq 1$ by definition and $|Q_{k_2}(x)|\leq 1$ by \eqref{bound-for-p-lin}
as
$\left[-\frac{\Xi}{\varkappa},\frac{\Xi}{\varkappa}\right]
\subset
[-1,1]$.
Therefore,
by defining $\mathcal{H}_k^{(\varkappa)}\coloneqq P_{k_1}\circ Q_{k_2}$,
the fact that $\mathcal{H}_k^{(\varkappa)}$ has maximal degree $k=k_1k_2=O(\varkappa\log^2\left(\frac{\varkappa}{\varepsilon}\right))$ follows straightforwardly,
and we can use \autoref{composite-poly-lem} to show that
\begin{align*}
\sup_{x\in\left[-\frac{\Xi}{\varkappa},\frac{\Xi}{\varkappa}\right]}\left|\big(f\circ g\big)(x)-\big(P_{k_1}\circ Q_{k_2}\big)(x)\right|
\leq
\varepsilon_f+L_P\, \frac{\varepsilon_g}{2}
\leq 
\frac{\varepsilon}{2}+|c|(2\Xi) e^{|c|(2\Xi)}
\left(
\frac{(\varepsilon/2)}{|c|(2\Xi) e^{|c|(2\Xi)}}
\right)
=\varepsilon.
\end{align*}
This proves \eqref{approx-precision-claim}.
To prove \eqref{gamma-H-upp-bnd-claim}, denote $\gamma_f\coloneqq \sup_{x\in[-1,1]}|f(x)|=e^{|c|(2\Xi)}$.
Recall that $\varepsilon_f\equiv\varepsilon/2\leq e^{|c|(2\Xi)}$ holds by definition of $\varepsilon$, and
from \eqref{bound-for-p-lin},
$\sup_{x\in[-1,1]}|Q_{k_2}(x)|\leq 1$,
so
\begin{align*}
\sup_{x\in[-1,1]}\left|
\mathcal{H}_k^{(\varkappa)}(x)
\right|
&=
\sup_{x\in[-1,1]}\Big|
\left(P_{k_1}\circ Q_{k_2}\right) (x)
-\left(f\circ Q_{k_2}\right) (x)
+\left(f\circ Q_{k_2}\right) (x)
\Big|\\
&\leq 
\sup_{Q_{k_2}(x)\in[-1,1]}\left|
f\left(Q_{k_2}(x)\right)
-P_{k_1}\left(Q_{k_2}(x)\right)
\right|
+
\sup_{Q_{k_2}(x)\in[-1,1]}\left|
f\left(Q_{k_2}(x)\right)
\right|\\
&\leq
\varepsilon_f+\gamma_f\\
&\leq
2e^{|c|(2\Xi)}. \qedhere
\end{align*}
\end{proof}
\begin{proposition}\label{Hk-using-estimate}
Let $c\in\mathbb{R}\backslash\{0\}$ be a fixed constant.
Let $\Sigma$ be a matrix satisfying the same assumption as \autoref{main-thm-for-gen-mat}, and suppose that we have access to the unitaries $U_{\Sigma}$ and $U_{\ket{z}}$ in \autoref{main-thm-for-gen-mat}.
Suppose further\footnote{See \autoref{rem-z-chi-dist}.} that 
we know $\norm{\vec{z}}_2$.
Denote $\ket{x}=\widetilde{\Sigma}^{1/2}\ket{z}/\lVert \widetilde{\Sigma}^{1/2}\ket{z}\rVert_2=\vec{x}/\norm{\vec{x}}_2$ for $\widetilde{\Sigma}\coloneqq \Sigma/\widetilde{\lambda}_{\max}$ and $\vec{x}\equiv\norm{\vec{z}}_2\Sigma^{1/2}\ket{z}=\Sigma^{1/2}\vec{z}$ whose components are unknown.
Suppose there exists a known constant $\Xi=\Theta(1)$ from \autoref{assump-inf-norm-bound-for-vec-x} such that $0<\norm{\vec{x}}_{\infty}\leq \Xi$ and $4\Xi\leq\norm{\vec{x}}_2$ hold. 
For any $\varepsilon_{\mathcal H_k}\in\left(0,e^{|c|(2\Xi)}\cdot\min\left\{4,|c|(2\Xi)\right\}\right]$,
if the block-encoding error $\varepsilon_{U_\Sigma}$ of $U_{\Sigma}$ satisfies
\[
\frac{\varepsilon_{U_\Sigma}}{\widetilde{\lambda}_{\max}}=o\left(\varepsilon_{\mathcal H_k}\cdot (L\widetilde{\kappa})^{-1.5} \cdot \log^{-3}\left(\frac{L\widetilde{\kappa}}{\varepsilon_{\mathcal H_k}}\right)\right),
\tagaligneq \label{cond-for-using-chakraborty-et-al-poly-construct1}
\]
then we can construct a degree-$k$ polynomial $\mathcal{H}^{(\widetilde{\eta})}_k$, where $k=O\left(\norm{\vec{x}}_2\log^2\left(\frac{\norm{\vec{x}}_2}{\varepsilon_{\mathcal H_k}}\right)\right)$, satisfying
\begin{gather}
\sup_{\zeta\in \left[-\frac{\Xi}{\norm{\vec{x}}_2},\frac{\Xi}{\norm{\vec{x}}_2}\right]}\left|
\mathcal{H}^{(\widetilde{\eta})}_k(\zeta)-e^{c\norm{\vec{x}}_2\cdot\zeta}
\right|
\leq \varepsilon_{\mathcal H_k}, \label{approx-precision-claim-w-eta}
\shortintertext{and}
\gamma_{\mathcal{H}^{(\widetilde{\eta})}_k}\coloneqq\sup_{\zeta\in[-1,1]}\left|
\mathcal{H}^{(\widetilde{\eta})}_k(\zeta)
\right|
\leq 2e^{|c|(2\Xi)} =O(1), \label{gamma-H-upp-bnd-claim-w-eta}
\end{gather}
where $\widetilde{\eta}$ is obtainable via QAE and satisfies $\widetilde{\eta}/\norm{\vec{x}}_2\in\left[\frac12,1\right]$. 
The QAE procedure that outputs such $\widetilde{\eta}$ with probability at least $0.99$ uses
\[
O
\left(
\frac{\sqrt{L\widetilde{\kappa}}}{\varepsilon_{\mathcal H_k}}
\left(
\widetilde{q}
+
n
+
\frac{\alpha_{U_\Sigma}}{\widetilde{\lambda}_{\max}}
\,L\widetilde{\kappa}\,\left(a_{U_\Sigma}+T_{U_\Sigma}\right)\log^2\left(\frac{L\widetilde{\kappa}}{\varepsilon_{\mathcal H_k}}\right)
+T_{U_{\ket{z}}}
\right)
\right) \tagaligneq \label{qae-cost-for-Hk-eta-tilde-abstract}
\]
elementary gates, where $\mathbb{Z}^{+}\ni\widetilde{q}=a_{U_\Sigma}+O\left(\log\log(L\widetilde{\kappa}/{\varepsilon_{\mathcal H_k}})\right)$, and an ancillary register of $\widetilde{q}+2$ qubits.
In particular, if \autoref{assump-plus-TU-ket-y-log2-N}, \autoref{assump-QROM-like-unitary}, and \autoref{assump-LK-plus} hold, we have $\varepsilon_{U_\Sigma}=0$ and condition \eqref{cond-for-using-chakraborty-et-al-poly-construct1} is satisfied automatically. Under such a circumstance, the QAE circuit requires
\[
O\left(\frac{1}{\varepsilon_{\mathcal H_k}}
\frac{\lVert \Sigma \rVert_F}{\lambda_{\max}}\,\kappa^{1.5}\,\polylog(N)
\log^2\left(\frac{\kappa}{\varepsilon_{\mathcal H_k}}\right)
\right) \tagaligneq \label{qae-cost-for-Hk-eta-tilde-concrete}
\]
elementary gate depth
and an ancillary register of $\widetilde{q}+2=O(\log(N)+\log\log(\kappa/\varepsilon_{\mathcal H_k}))$ qubits.
\end{proposition}
\begin{proof}
Set $\varepsilon_{\mathrm{QAE}}\equiv \frac{\varepsilon_{\mathcal H_k}}{4}\frac{1}{\sqrt{L\widetilde{\kappa}}}\frac{e^{-|c|(2\Xi)}}{|c|(2\Xi)}=\Theta\left({\varepsilon_{\mathcal H_k}}/{\sqrt{L\widetilde{\kappa}}}\right)$.
To use \autoref{qae-for-main-thm}, we need that \eqref{cond-for-using-chakraborty-et-al-qae-prop1} holds, which is already assumed in \eqref{cond-for-using-chakraborty-et-al-poly-construct1}. Therefore, by using \autoref{qae-for-main-thm}, we can output $\eta$ that approximates $\norm{\widetilde{\Sigma}^{1/2}\ket{z}}_2=\norm{\frac{\Sigma^{1/2}}{\sqrt{\widetilde{\lambda}_{\max}}}\frac{\vec{z}}{\norm{\vec{z}}_2}}_2=\norm{\vec{x}}_2/\big(\norm{\vec{z}}_2\sqrt{\widetilde{\lambda}_{\max}}\big)$ to within $\varepsilon_{\mathrm{QAE}}$, with resource requirements \eqref{qae-cost-for-Hk-eta-tilde-abstract} or \eqref{qae-cost-for-Hk-eta-tilde-concrete} as claimed.
Denote
$\widehat{\eta}\coloneqq \norm{\vec{z}}_2\sqrt{\widetilde{\lambda}_{\max}}\cdot\eta$
and
$\widehat{\varepsilon}\equiv \norm{\vec{z}}_2\sqrt{\widetilde{\lambda}_{\max}}\ \varepsilon_{\mathrm{QAE}}$.
By the assumption $\varepsilon_{\mathcal H_k}\leq |c|(2\Xi)e^{|c|(2\Xi)}$, we have $\varepsilon_{\mathrm{QAE}}\leq 1/(4\sqrt{L\widetilde{\kappa}})$.
Hence,
\autoref{epsilon-QAE-implies-epsilon-hat} applies (with $C= 4$), and we have $\left|\widehat{\eta}-\norm{\vec{x}}_2\right|\leq\widehat{\varepsilon}\leq\norm{\vec{x}}_2/4$.
Define $\widetilde{\eta}\coloneqq \widehat{\eta}-\widehat{\varepsilon}$. 
Then, $\norm{\vec{x}}_2-\widehat{\varepsilon}
\leq
\widehat{\eta}
\leq
\norm{\vec{x}}_2+\widehat{\varepsilon}$ 
implies 
$\norm{\vec{x}}_2-2\widehat{\varepsilon}
\leq
\widetilde{\eta}
\leq
\norm{\vec{x}}_2$, which further implies
$
\frac{\norm{\vec{x}}_2}{2}
\leq
\widetilde{\eta}
\leq
\norm{\vec{x}}_2
$ since $\widehat{\varepsilon}\leq\norm{\vec{x}}_2/4$.
This shows that 
$\widetilde{\eta}$ is a lower bound for $\norm{\vec{x}}_2$
and proves the claim
$\widetilde{\eta}/\norm{\vec{x}}_2\in\left[\frac12,1\right]$.
Now, as $2\Xi \leq \frac{\norm{\vec{x}}_2}{2}\leq \widetilde{\eta}$ and $\varepsilon_{\mathcal H_k}/2\leq 2e^{|c|(2\Xi)}$, we can apply \autoref{poly-approx-Hk-to-h} to construct a degree-$k$ polynomial $\mathcal{H}^{(\widetilde{\eta})}_k$ such that
\begin{gather}
\sup_{\zeta\in \left[-\frac{\Xi}{\widetilde{\eta}},\frac{\Xi}{\widetilde{\eta}}\right]}\left|
\mathcal{H}^{(\widetilde{\eta})}_k(\zeta)-e^{c\cdot\widetilde{\eta}\cdot\zeta}
\right|
\leq \frac{\varepsilon_{\mathcal H_k}}{2}, \label{approx-precision-claim-widetilde-eta}
\shortintertext{and}
\gamma_{\mathcal{H}^{(\widetilde{\eta})}_k}\coloneqq\sup_{\zeta\in[-1,1]}\left|
\mathcal{H}^{(\widetilde{\eta})}_k(\zeta)
\right|
\leq 2e^{|c|(2\Xi)} =O(1), \label{gamma-H-upp-bnd-claim-widetilde-eta}
\end{gather}
hold with degree $k=O\left(\widetilde{\eta}\log^2\left(\frac{\widetilde{\eta}}{\varepsilon_{\mathcal H_k}}\right)\right)=O\left(\norm{\vec{x}}_2\log^2\left(\frac{\norm{\vec{x}}_2}{\varepsilon_{\mathcal H_k}}\right)\right)$.
This proves \eqref{gamma-H-upp-bnd-claim-w-eta}.
To prove \eqref{approx-precision-claim-w-eta}, observe that
\begin{align*}
&\sup_{\zeta\in \left[-\frac{\Xi}{\widetilde{\eta}},\frac{\Xi}{\widetilde{\eta}}\right]}
\left|
\mathcal{H}^{(\widetilde{\eta})}_k(\zeta)-e^{c\norm{\vec{x}}_2\cdot\zeta}
\right|
\leq 
\sup_{\zeta\in \left[-\frac{\Xi}{\widetilde{\eta}},\frac{\Xi}{\widetilde{\eta}}\right]}
\left|
\mathcal{H}^{(\widetilde{\eta})}_k(\zeta)-e^{c \cdot\widetilde{\eta}\cdot\zeta}
\right|
+
\sup_{\zeta\in \left[-\frac{\Xi}{\widetilde{\eta}},\frac{\Xi}{\widetilde{\eta}}\right]}
\left|
e^{c \cdot\widetilde{\eta}\cdot\zeta}-e^{c\norm{\vec{x}}_2\cdot\zeta}
\right|, \tagaligneq \label{eval-bound-approx-err-Hk-eta-tilde}
\end{align*}
where the first term on the RHS is upper-bounded by $\varepsilon_{\mathcal H_k}/2$ due to \eqref{approx-precision-claim-widetilde-eta}. The rest is to show that the second term on the RHS is also bounded by $\varepsilon_{\mathcal H_k}/2$. 
To see this, observe that by the mean value theorem\footnote{See e.g.\ \cite[5.10 Theorem]{Rud76}.}, there exists some $\xi_{c,\zeta}\in
\Big(
\min\left\{c\widetilde{\eta}\zeta,c\norm{\vec{x}}_2\zeta\right\},
\max\left\{c\widetilde{\eta}\zeta,c\norm{\vec{x}}_2\zeta\right\}
\Big)\eqqcolon\mathcal{R}_{c,\zeta}$, 
such that the difference $\left|
e^{c \cdot\widetilde{\eta}\cdot\zeta}-e^{c\norm{\vec{x}}_2\cdot\zeta}
\right|
$
equals
$
\Big|
c
(\widetilde{\eta}-\norm{\vec{x}}_2)\cdot \zeta
\Big|
\cdot
\left|
e^{\xi_{c,\zeta}}
\right|
$. That is,
\begin{align*}
\sup_{\zeta\in \left[-\frac{\Xi}{\widetilde{\eta}},\frac{\Xi}{\widetilde{\eta}}\right]}
\left|
e^{c \cdot\widetilde{\eta}\cdot\zeta}-e^{c\norm{\vec{x}}_2\cdot\zeta}
\right|
&\leq 
\sup_{\zeta\in \left[-\frac{\Xi}{\widetilde{\eta}},\frac{\Xi}{\widetilde{\eta}}\right]}
\Big|
c
(\widetilde{\eta}-\norm{\vec{x}}_2)\cdot \zeta
\Big|
\cdot
\sup_{\zeta\in \left[-\frac{\Xi}{\widetilde{\eta}},\frac{\Xi}{\widetilde{\eta}}\right]}
\sup_{\xi_{c,\zeta}\in \mathcal{R}_{c,\zeta}}
\left| 
e^{\xi_{c,\zeta}}
\right|. \tagaligneq \label{eval-exp-widetilde-eta-vs-exp-norm-x}
\end{align*}
By monotonicity, ${}^{\forall}\xi_{c,\zeta}\in\mathcal{R}_{c,\zeta}$, $\left|e^{\xi_{c,\zeta}}\right| \leq \max\left\{e^{c\norm{\vec{x}}_2\cdot \zeta},e^{c \cdot \widetilde{\eta} \cdot \zeta}\right\}
\leq e^{|c|\max\left\{\norm{\vec{x}}_2,\widetilde{\eta}\right\} \cdot |\zeta|}
=e^{|c|\norm{\vec{x}}_2\cdot |\zeta|}$.
Since ${}^\forall \zeta\in \left[-\frac{\Xi}{\widetilde{\eta}},\frac{\Xi}{\widetilde{\eta}}\right]$, $|\zeta|\leq \frac{\Xi}{\widetilde{\eta}} \leq \frac{2\Xi}{\norm{\vec{x}}_2}$, we have ${\displaystyle \sup_{\zeta\in \left[-\frac{\Xi}{\widetilde{\eta}},\frac{\Xi}{\widetilde{\eta}}\right]}\sup_{\xi_{c,\zeta}\in \mathcal{R}_{c,\zeta}}}
\left| 
e^{\xi_{c,\zeta}}
\right| 
\leq
{\displaystyle
\sup_{\zeta\in \left[-\frac{\Xi}{\widetilde{\eta}},\frac{\Xi}{\widetilde{\eta}}\right]}
}
e^{|c|\norm{\vec{x}}_2\cdot |\zeta|}
\leq 
e^{|c|(2\Xi)}$ holds for all $c\in\mathbb{R}$.
Therefore,
\begin{align*}
\eqref{eval-exp-widetilde-eta-vs-exp-norm-x} 
&\leq
|c|\cdot2\widehat{\varepsilon} \cdot \frac{2\Xi}{\norm{\vec{x}}_2} \cdot e^{|c|(2\Xi)}
= 2 \varepsilon_{\mathrm{QAE}} \cdot \frac{\norm{\vec{z}}_2}{\norm{\vec{x}}_2}\sqrt{\widetilde{\lambda}_{\max}}
\cdot |c|(2\Xi)
\cdot e^{|c|(2\Xi)}
\equiv \frac{\varepsilon_{\mathcal H_k}}{2}\frac{\sqrt{\widetilde{\lambda}_{\max}}}{\sqrt{L\widetilde{\kappa}}} \frac{\norm{\vec{z}}_2}{\norm{\vec{x}}_2}
\leq \frac{\varepsilon_{\mathcal H_k}}{2},
\end{align*}
where the last inequality follows because $\frac{\sqrt{\widetilde{\lambda}_{\max}}}{\sqrt{L\widetilde{\kappa}}} \frac{\norm{\vec{z}}_2}{\norm{\vec{x}}_2}\leq1$ by the same caluculation as in the proof of \autoref{epsilon-QAE-implies-epsilon-hat}. Combining this with \eqref{eval-bound-approx-err-Hk-eta-tilde} and \eqref{approx-precision-claim-widetilde-eta}, we have shown that
\begin{align*}
&\sup_{\zeta\in \left[-\frac{\Xi}{\widetilde{\eta}},\frac{\Xi}{\widetilde{\eta}}\right]}
\left|
\mathcal{H}^{(\widetilde{\eta})}_k(\zeta)-e^{c\norm{\vec{x}}_2\cdot\zeta}
\right|
\leq 
\frac{\varepsilon_{\mathcal H_k}}2+\frac{\varepsilon_{\mathcal H_k}}2
=
\varepsilon_{\mathcal H_k}.
\end{align*}
Since $\widetilde{\eta}\leq \norm{\vec{x}}_2$, $\left[-\frac{\Xi}{\norm{\vec{x}}_2},\frac{\Xi}{\norm{\vec{x}}_2}\right]\subset \left[-\frac{\Xi}{\widetilde{\eta}},\frac{\Xi}{\widetilde{\eta}}\right]$, and hence the claim \eqref{approx-precision-claim-w-eta} is proved.
\end{proof}
The following corollary is provided to handle the case where $\vec{y}$ represents the cumulative sum of a correlated Gaussian vector.
Note that when $\vec{y}$ is the cumulative sum of increments of a Gaussian process, its law is identical to the law of the path values on the same grid, so we also have $\norm{\vec{y}}_{\infty} < \Xi$ with high probability.
\begin{corollary}\label{Hk-using-estimate-sum-of-inc}
Let $c\in\mathbb{R}\backslash\{0\}$ be a fixed constant.
Let $\Sigma$ be a matrix satisfying the same assumption as \autoref{main-thm-for-gen-mat} and $\ket{z}$ be as in \autoref{main-thm-for-gen-mat}.
Suppose that the assumptions about data loaders in polylogarithmic time (\autoref{assump-plus-TU-ket-y-log2-N} and \autoref{assump-QROM-like-unitary}), and the assumption about parameter estimation (\autoref{assump-LK-plus}) hold in the same way as
\autoref{main-thm-coro-concrete-cost-sum-of-inc}.
Suppose further\footnote{See \autoref{rem-z-chi-dist}.} that we know $\norm{\vec{z}}_2$.
Denote $\ket{y}=\mathcal{L}_N\widetilde{\Sigma}^{1/2}\ket{z}/\lVert \mathcal{L}_N\widetilde{\Sigma}^{1/2}\ket{z}\rVert_2=\vec{y}/\norm{\vec{y}}_2$, for $\widetilde{\Sigma}\coloneqq \Sigma/\widetilde{\lambda}_{\max}$ and  $\vec{y}\equiv\norm{\vec{z}}_2\mathcal{L}_N\Sigma^{1/2}\ket{z}=\mathcal{L}_N\Sigma^{1/2}\vec{z}$ whose components are unknown.
Suppose there exists a known constant $\Xi=\Theta(1)$ from \autoref{assump-inf-norm-bound-for-vec-x} such that $0<\norm{\vec{y}}_{\infty}\leq \Xi$ and $4\Xi\leq\norm{\vec{y}}_2$ hold.
Then, for any $\varepsilon_{\mathcal H_k}\in\left(0,e^{|c|(2\Xi)}\cdot\min\left\{4,|c|(2\Xi)\right\}\right]$, 
we can construct a degree-$k$ polynomial $\mathcal{H}^{(\widetilde{\eta})}_k$ that satisfies
\begin{gather*}
\sup_{\zeta\in \left[-\frac{\Xi}{\norm{\vec{y}}_2},\frac{\Xi}{\norm{\vec{y}}_2}\right]}\left|
\mathcal{H}^{(\widetilde{\eta})}_k(\zeta)-e^{c\norm{\vec{y}}_2\cdot\zeta}
\right|
\leq \varepsilon_{\mathcal H_k}, 
\shortintertext{and}
\gamma_{\mathcal{H}^{(\widetilde{\eta})}_k}\coloneqq\sup_{\zeta\in[-1,1]}\left|
\mathcal{H}^{(\widetilde{\eta})}_k(\zeta)
\right|
\leq 2e^{|c|(2\Xi)} =O(1),
\end{gather*}
where $\widetilde{\eta}$ is obtainable via QAE and satisfies $\widetilde{\eta}/\norm{\vec{y}}_2\in\left[\frac12,1\right]$, and the polynomial degree is $k=O\left(\norm{\vec{y}}_2\log^2\left(\frac{\norm{\vec{y}}_2}{\varepsilon_{\mathcal H_k}}\right)\right)$. 
The QAE procedure that outputs such $\widetilde{\eta}$ with probability at least $0.99$ has
\[
O\left(\frac{1}{\varepsilon_{\mathcal H_k}}
\frac{\lVert \Sigma \rVert_F}{\lambda_{\max}}\,\kappa^{1.5}\,N\polylog(N)
\log^2\left(\frac{\kappa}{\varepsilon_{\mathcal H_k}}\right)
\right) \tagaligneq \label{qae-cost-for-Hk-eta-tilde-concrete-sum-of-inc-case}
\]
elementary gate depth, and an ancillary register of $\widetilde{q}+2=O(\log(N)+\log\log(\kappa/\varepsilon_{\mathcal H_k}))$ qubits is used.
\end{corollary}
\begin{proof}
Set
$\varepsilon_{\mathrm{QAE}}\equiv \frac{\varepsilon_{\mathcal H_k}}{8}\frac{1}{\sqrt{L\widetilde{\kappa}}}\frac{e^{-|c|(2\Xi)}}{|c|(2\Xi)}=\Theta\left({\varepsilon_{\mathcal H_k}}/{\sqrt{L\widetilde{\kappa}}}\right)$.
Using \autoref{qae-for-main-thm-sum-of-inc-case}, we can output $\eta$ that approximates $\norm{\mathcal{L}_N\widetilde{\Sigma}^{1/2}\ket{z}}_2=\norm{\mathcal{L}_N\frac{\Sigma^{1/2}}{\sqrt{\widetilde{\lambda}_{\max}}}\frac{\vec{z}}{\norm{\vec{z}}_2}}_2=\norm{\vec{y}}_2/\big(\norm{\vec{z}}_2\sqrt{\widetilde{\lambda}_{\max}}\big)$ to within $\varepsilon_{\mathrm{QAE}}$, with resource requirements \eqref{qae-cost-for-Hk-eta-tilde-concrete-sum-of-inc-case} as claimed, since under \autoref{assump-LK-plus}, $\sqrt{L\widetilde{\kappa}}=\widetilde{\Theta}(\sqrt{\kappa})$.
By the assumption $\varepsilon_{\mathcal H_k}\leq |c|(2\Xi)e^{|c|(2\Xi)}$, we have $\varepsilon_{\mathrm{QAE}}\leq 1/(8\sqrt{L\widetilde{\kappa}})$.
Similarly to the discussion under \eqref{norm-y-norm-z-ratio-ineq} in \autoref{qae-for-norm-y},
this implies (with $C= 4$)
$\left|\widehat{\eta}-\norm{\vec{y}}_2\right|\leq\widehat{\varepsilon}\leq\norm{\vec{y}}_2/4$.
The rest follows directly in the same manner as the proof of \autoref{Hk-using-estimate}, 
but with $\norm{\vec{x}}_2$ replaced by $\norm{\vec{y}}_2$. 
The upper bound evaluation for \eqref{eval-exp-widetilde-eta-vs-exp-norm-x}
also needs us to replace $\norm{\vec{z}}_2/\norm{\vec{x}}_2\leq 1/\sqrt{\lambda_{\min}}$ with $\norm{\vec{z}}_2/\norm{\vec{y}}_2\leq 2/\sqrt{\lambda_{\min}}$, which also holds from \eqref{norm-y-norm-z-ratio-ineq} of \autoref{qae-for-norm-y}.
The current definition of $\varepsilon_{\mathrm{QAE}}$ already makes up for the factor $2$ in the numerator.
\end{proof}
\subsubsection{Exponentiation algorithms}
The key subroutine used to realise the non-linear transformation algorithm of \cite{RR23} is called the quantum singular value transformation (QSVT), a technique developed from quantum signal processing (QSP) and reformulated by \cite{GSLW19}. Given a degree-$k$ polynomial $P_k$, QSVT can map $A\equiv U\Lambda V^{\dagger}$ to $P_k(A)\equiv UP_k(\Lambda) V^{\dagger}$. If $P_k$ approximates some non-linear function $f$ in the uniform norm, $P_k(A)$ approximates $f(A)\equiv Uf(\Lambda) V^{\dagger}$ in the induced 2-norm. 
We formulate the main result of QSVT in \autoref{lem-qsvt} given below.
\begin{lemma}[{\cite[Theorem 56 from full paper: arXiv:1806.01838]{GSLW19}}]\label{lem-qsvt}
Given a degree-$k$ real polynomial $P_k$ such that $\sup_{x\in[-1,1]}|P_k(x)|\leq 1/2$, and a $(\alpha,a,\varepsilon)$-block-encoding $U_A$ of a Hermitian matrix $A$, we can construct $U_{P_k(A/\alpha)}$, a $(1,a+2,4k\sqrt{\varepsilon/\alpha}+\delta)$-block-encoding of $P_k(A/\alpha)$ for $\delta\geq0$, by using $k$ calls to $U_A$ and $U_A^{\dagger}$, a single application of controlled-$U_A$, and $O(ka)$ calls to single and two-qubit elementary gates. The description for the QSVT circuit to construct $U_{P_k(A/\alpha)}$ can be computed in $O(\mathrm{poly}\{k,\log(1/\delta)\})$ classical computation time.
\end{lemma}
Simply put, the non-linear transformation algorithm, which is initially developed in \cite{GMF24} and then refined by \cite{RR23}, goes as follows.
It makes use of the state-preparation unitary $U_{\ket{\psi}}$ of a state $\ket{\psi}=\sum_{i=1}^{N}\psi_i\ket{i}$ to create a block-encoding of the diagonal matrix $\mathrm{diag}(\ket{\psi})$ whose entries are amplitudes of $\ket{\psi}$.
Then, QSVT is used to transform these diagonal entries non-linearly, and finally multiplying this onto a quantum state such as $\ket{+_N}=\sum_{i=1}^N\ket{i}/\sqrt{N}$ or any $\ket{f}=\vec{f}/\lVert \vec{f}\rVert$ of interest for weighted sum yields the quantum state with amplitudes transformed non-linearly.
The procedure to create a block-encoding of $\mathrm{diag}(\ket{\psi})$ is discussed in \cite{RR23}, whose result is formulated as follows.
\begin{lemma}[Diagonal block-encoding of amplitudes {\cite[Lemma 6]{RR23}}]\label{rr23-thm2-2}
Given an $n$-qubit state preparation unitary $U_{\ket{\psi}}$ that prepares $U_{\ket{\psi}}\ket{0}^{\otimes n}=\ket{\psi}=\sum_{i=0}^{2^n-1}\psi_i\ket{i}$, where ${}^\forall i$, $\psi_i \in\mathbb{R}$, we can construct a $(1,n+2,0)$-block-encoding $U_{\mathrm{diag}(\ket{\psi})}$ of the diagonal matrix $\mathrm{diag}(\psi_0,\dots,\psi_{N-1})$ with $O(n)$ elementary gate depth\footnote{Although stated originally as `gate depth', the proofs suggest $O(n)$ elementary gate `count' requirement, see Definition 9., Definition 10., Lemma 3, and Lemma 6 in \cite{RR23}.} and a total of $6$ queries to a controlled-$U_{\ket{\psi}}$ circuit\footnote{Each controlled-$U_{\ket{\psi}}$ costs $1$ call to $U_{\ket{\psi}}$ and $O(1)$ elementary gates, cf.\ \cite[Definition 9]{RR23}.\label{footnote-controlled-U}}.
\end{lemma}
\begin{lemma}\label{bl-en-of-tilde-D-is-bl-en-of-D}
Suppose $\ket{\psi}$ and $\ket{\varphi}$ are normalised quantum states of the same number of qubits, and suppose $U$ is a $(1,a,0)$-block-encoding of $\mathrm{diag}(\ket{\psi})$. Then, for $\varepsilon>0$, $\norm{\ket{\psi}-\ket{\varphi}}_\infty\leq \varepsilon$ implies $U$ is a $(1,a,\varepsilon)$-block-encoding of $\mathrm{diag}(\ket{\varphi})$.
\end{lemma}
\begin{proof}
Observe that $\left\lVert\mathrm{diag}(\ket{\varphi})-\mathrm{diag}(\ket{\psi})\right\rVert=\norm{\ket{\psi}-\ket{\varphi}}_\infty$ because the induced 2-norm of a matrix $A$ is just the largest singular value; when $A$ is diagonal, this equals the absolute value of the entry with largest magnitude. Then, recall the definition of a block-encoding:
\begin{align*}
\left\lVert
\mathrm{diag}(\ket{\varphi})
-1\cdot \left(\bra{0}^{\otimes a}\otimes I\right)
U
\left(\ket{0}^{\otimes a}\otimes I\right)
\right\rVert
&=
\left\lVert
\mathrm{diag}(\ket{\varphi})
-
\mathrm{diag}(\ket{\psi})
\right\rVert
\leq
\varepsilon. \qedhere
\end{align*}
\end{proof}
\autoref{alg-exponentiation} illustrates the steps of how our proposed algorithm (the following \autoref{preparing-ket-exp-x-from-ket-x-w-error}) realises the preparation of an exponentiated correlated Gaussian vector.
\begin{algorithm}[H]
\caption{Constructing an SPBE for an exponentiated correlated Gaussian vector}\label{alg-exponentiation}
\begin{algorithmic}[1]
\Require $\widetilde{\lambda}_{\max},\widetilde{\kappa},L,K,U_{\Sigma},U_{\ket{z}}$ as in \autoref{main-thm-for-gen-mat}, $U_{\ket{f_N}}$ as in \autoref{preparing-ket-exp-x-from-ket-x-w-error}, and $\norm{\vec{z}}\in(0,\infty)$ for some $\vec{z}$ being a realistion of $Z\sim\mathcal{N}(0,I_{N})$.
\Ensure a unitary $U_{\ket{\vec{f}\odot e^{c\vec{x}}},\varepsilon}:\ket{0}_{\mathrm{Q}\mathrm{A}\mathrm{I}}\mapsto\ket{\widetilde{\varphi}}_{\mathrm{Q}\mathrm{A}\mathrm{I}}$ such that $\bigg\lVert\ket{\widetilde{\varphi}}_{\mathrm{Q}\mathrm{A}\mathrm{I}}-\ket{0}_{\mathrm{Q}\mathrm{A}}\ket{\vec{f}\odot e^{c\vec{x}}}_{\mathrm{I}}\bigg\rVert \leq\varepsilon$.
\Statex \hrulefill
\Statex (One-time QAE for defining QSVT polynomial of transformation)
\State Perform a one-time QAE to obtain an estimate $\eta$ of $\norm{\widetilde{\Sigma}^{1/2}\ket{z}}$, where $\widetilde{\Sigma}=\Sigma/\widetilde{\lambda}_{\max}$, via \autoref{qae-for-main-thm}.
\State Use $\eta$ to define $\widetilde{\eta}$ satisfying $\widetilde{\eta}/\norm{\vec{x}}\in\left[\frac12,1\right]$ as suggested in \autoref{Hk-using-estimate}, where $\norm{\vec{x}}\equiv \norm{\Sigma^{1/2}\ket{z}}\norm{\vec{z}}$.
\State Use $\widetilde{\eta}$ to define a polynomial $\mathcal{H}^{(\widetilde{\eta})}_k(\zeta)$ that approximates $\zeta \mapsto e^{c\norm{\vec{x}} \cdot\zeta}$ on the interval $\left[-\frac{\Xi}{\norm{\vec{x}}_2},\frac{\Xi}{\norm{\vec{x}}_2}\right]$ via \autoref{Hk-using-estimate}.
\State Define $\widehat{H}_k$, a scaled version of $\mathcal{H}^{(\widetilde{\eta})}_k$, to satisfy the condition required for using QSVT.
\Statex \hrulefill
\Statex (Preparing the QSVT building-block unitary)
\State Construct a unitary $U_{\ket{\varrho},\varepsilon}^{\widetilde{q},\norm{\vec{x}}_2}\equiv U_{\ket{\varrho},\varepsilon_{\varrho}}:\ket{0}_{\mathrm{A}\mathrm{I}}\mapsto\ket{\varrho}_{\mathrm{A}\mathrm{I}}$ 
such that 
$\lVert \ket{\varrho}_{\mathrm{A}\mathrm{I}}-\ket{0}_{\mathrm{A}}\ket{x}_{\mathrm{I}} \rVert_2 \leq \varepsilon_{\varrho}$ using \autoref{main-thm-for-gen-mat}, with requirements as in \autoref{preparing-ket-exp-x-from-ket-x-w-error}.
\State Construct a unitary $U_{\widetilde{\mathcal{D}}}$, a block-encoding of $\mathcal{D}\equiv \mathrm{diag}(\ket{0}^{\otimes \widetilde{q}}\ket{x})$, using $U_{\ket{\varrho},\varepsilon}^{\widetilde{q},\norm{\vec{x}}_2}$ via \autoref{rr23-thm2-2}.
\Statex \hrulefill
\Statex (Performing QSVT, projecting back transformed amplitudes, and normalisation via QAA)
\State Use QSVT to transform $\mathcal{D}$ according to the polynomial $\widehat{H}_k$ (that is, to
create $U_{\widehat{H}_k(\mathcal{D})}$, a block-encoding of $\widehat{H}_k(\mathcal{D})$, using $U_{\widetilde{\mathcal{D}}}$ as building blocks).
\State Applying $U_{\widehat{H}_k(\mathcal{D})}$ to the state $\ket{0}^{\otimes \widetilde{q}}\ket{f_N}$ produces a state that is close to the direction of $\ket{0}^{\otimes \widetilde{q}}(e^{c \vec{x}}\odot \ket{f_N})$ (in practice, we consider the composition of $U_{\widehat{H}_k(\mathcal{D})}$ and $U_{\ket{f_N}}$ than really creating a state).
\State Perform QAA to obtain a unitary $U_{\ket{\vec{f}\odot e^{c\vec{x}}},\varepsilon}:\ket{0}_{\mathrm{Q}\mathrm{A}\mathrm{I}}\mapsto\ket{\widetilde{\varphi}}_{\mathrm{Q}\mathrm{A}\mathrm{I}}$, where $\ket{\widetilde{\varphi}}_{\mathrm{Q}\mathrm{A}\mathrm{I}}$ approximates the normalised state $\ket{0}_{\mathrm{Q}\mathrm{A}}\ket{e^{c \vec{x}} \odot \vec{f}}_{\mathrm{I}}$, using $U_{\widehat{H}_k(\mathcal{D})}$ and $U_{\ket{f_N}}$ as building blocks.
\end{algorithmic}
\end{algorithm}
\begin{remark}\label{rem-z-chi-dist}
$\norm{\vec{z}}$ can be sampled from chi-distribution independently of $\ket{z}$, cf.\ the discussion following \autoref{rem-abt-LK-assmp-plus}.
\end{remark}
\begin{theorem}\label{preparing-ket-exp-x-from-ket-x-w-error}
Let $\Sigma$ be a matrix satisfying the same assumption as \autoref{main-thm-for-gen-mat},
and suppose that we have access to the unitaries $U_{\Sigma}$ and $U_{\ket{z}}$ in \autoref{main-thm-for-gen-mat}.
Let $\Xi$ and $\vec{x}$ be as in \autoref{Hk-using-estimate}.
Suppose that $\|\vec{z}\|_2$ is known.
Let $c \in \mathbb{R}\backslash\{0\}$ be a fixed constant, 
and
suppose we have access to a unitary $U_{\ket{f_N}}:\ket{0}^{\otimes n}\mapsto \ket{f_N}=\vec{f}/\lVert \vec{f} \rVert$,
for $\vec{f}\in\mathbb{R}^N\backslash \{\vec{0}\}$ whose components are known.
Suppose implementing $U_{\ket{f_N}}$ requires $T_{U_{\ket{f_N}}}$ elementary gates.
For any
$\varepsilon\in(0,2]$,
if the block-encoding error $\varepsilon_{U_\Sigma}$ of $U_{\Sigma}$ satisfies
\[
\frac{\varepsilon_{U_\Sigma}}{\widetilde{\lambda}_{\max}}=o\left(\frac{\varepsilon^2}{\norm{\vec{x}}_2^2}\log^{-4}\left(\frac{\norm{\vec{x}}_2}{\varepsilon}\right)\cdot (L\widetilde{\kappa})^{-1.5} \cdot \log^{-3}\left(\frac{L\widetilde{\kappa}\norm{\vec{x}}_2}{\varepsilon}\right)\right),
\tagaligneq \label{cond-for-using-chakraborty-et-al-non-lin-transform-1}
\]
then
we can construct a unitary $U_{\ket{\vec{f}\odot e^{c\vec{x}}},\varepsilon}$ that maps $\ket{0}_{\mathrm{Q}\mathrm{A}\mathrm{I}}\mapsto\ket{\widetilde{\varphi}}_{\mathrm{Q}\mathrm{A}\mathrm{I}}$ satisfying $\bigg\lVert\ket{\widetilde{\varphi}}_{\mathrm{Q}\mathrm{A}\mathrm{I}}-\ket{0}_{\mathrm{Q}\mathrm{A}}\ket{\vec{f}\odot e^{c\vec{x}}}_{\mathrm{I}}\bigg\rVert \leq\varepsilon$, where 
\[
\ket{\vec{f}\odot e^{c\vec{x}}}
=
\frac{\vec{f}\odot e^{c\vec{x}}}{\norm{\vec{f}\odot e^{c\vec{x}}}}
=
\frac{1}{\sqrt{\sum_{i=1}^N |f_i|^2\cdot e^{2c x_i}}}
\sum_{i=1}^N f_i \cdot e^{c x_i} \ket{i},
\]
for
$\vec{x}$ being a realisation of the random vector $X\sim\mathcal{N}(0,\Sigma)$,
$\mathrm{I}$ is a system register of $n=\ceil{\log_2(N+1)}$ qubits, $\mathrm{Q}$ is an ancillary register of $\widetilde{q}+n+4$ qubits and $\mathrm{A}$ is an ancillary register of $\widetilde{q}$ qubits, for $\mathbb{Z}^{+}\ni\widetilde{q}=a_{U_\Sigma}+O\left(\log\log(\norm{\vec{x}}_2L\widetilde{\kappa}/{\varepsilon})\right)$. 
The implementation of $U_{\ket{\vec{f}\odot e^{c\vec{x}}},\varepsilon}$ requires
\[
T_{U_{\ket{\vec{f}\odot e^{c\vec{x}}},\varepsilon}}
=
O\left(
\norm{\vec{x}}_2
\cdot
\log^3\left(\frac{1}{\varepsilon}\right)
\cdot
\left(\widetilde{q}+n
+
T_{U_{\ket{\varrho},\varepsilon}^{\widetilde{q},\norm{\vec{x}}_2}}
\right)
+
\log\left(\frac{1}{{\varepsilon}}\right)
\cdot
T_{U_{\ket{f_N}}}
\right) \tagaligneq \label{total-abstract-cost-non-linear-transform}
\]
elementary gates, where
\begin{gather*}
T_{U_{\ket{\varrho},\varepsilon}^{\widetilde{q},\norm{\vec{x}}_2}}=O\left(\sqrt{L\widetilde{\kappa}}
\cdot
\log\left(\frac{\norm{\vec{x}}_2}{{\varepsilon}}\right)
\cdot
\Biggl(
\widetilde{q}
+
n+T_{U_{\ket{z}}}+\frac{\alpha_{U_\Sigma}}{\widetilde{\lambda}_{\max}}\,L\widetilde{\kappa}\,\Big(a_{U_\Sigma}+T_{U_\Sigma}\Big)\log^2\left(\frac{\norm{\vec{x}}_2 L\widetilde{\kappa}}{\varepsilon}\right)\Biggr)
\right). \tagaligneq \label{cost-to-prep-ket-x-within-epx-via-main-thm}
\end{gather*}
For creating the description of the QSVT circuit used in $U_{\ket{\vec{f}\odot e^{c\vec{x}}},\varepsilon}$, a one-time
classical computation in $O\left(\mathrm{poly}\left\{\norm{\vec{x}}_2\log^2\left({\norm{\vec{x}}_2}/{\varepsilon}\right),\log\left({1}/{\varepsilon}\right)\right\}\right)$ time is run. For determining the QSVT polynomial of transformation, a one-time QAE procedure is run with success probability at least $0.99$, using
\[
O
\left(
\frac{\sqrt{L\widetilde{\kappa}}}{\varepsilon}
\left(
\widehat{q}
+
n
+
\frac{\alpha_{U_\Sigma}}{\widetilde{\lambda}_{\max}}
\,L\widetilde{\kappa}\,\left(a_{U_\Sigma}+T_{U_\Sigma}\right)\log^2\left(\frac{L\widetilde{\kappa}}{\varepsilon}\right)
+T_{U_{\ket{z}}}
\right)
\right) \tagaligneq \label{qae-cost-for-Hk-eta-tilde-abstract-in-non-lin-thm}
\]
elementary gates and an ancillary register of $\widehat{q}+2$ qubits, where $\mathbb{Z}^{+}\ni\widehat{q}=a_{U_\Sigma}+O\left(\log\log(L\widetilde{\kappa}/{\varepsilon})\right)$.
\end{theorem}
\begin{proof}
Denote $h(\zeta)\coloneqq e^{c\norm{\vec{x}}_2\cdot\zeta}$.
Let $\widetilde{\varepsilon}\coloneqq \frac{\varepsilon}{16}e^{-3|c|\Xi}=\Theta(\varepsilon)$ and let $\varepsilon_{\mathcal H_k}\coloneqq 2e^{2|c|\Xi}\widetilde{\varepsilon}$.
Note that \eqref{cond-for-using-chakraborty-et-al-non-lin-transform-1} implies ${\varepsilon_{U_\Sigma}}/{\widetilde{\lambda}_{\max}}=o\left(\varepsilon\cdot (L\widetilde{\kappa})^{-1.5} \cdot \log^{-3}\left(\frac{L\widetilde{\kappa}}{\varepsilon}\right)\right)$,
which further implies the requirement \eqref{cond-for-using-chakraborty-et-al-poly-construct1} of \autoref{Hk-using-estimate}, since $\varepsilon_{\mathcal H_k}=\Theta(\widetilde{\varepsilon})=\Theta(\varepsilon)$.
Also, as $\varepsilon\leq 2$, we have $\widetilde{\varepsilon}\leq 1$ and hence $\varepsilon_{\mathcal H_k}\leq 2e^{|c|(2\Xi)}$, so 
we can use \autoref{Hk-using-estimate} to construct a degree-$k$ polynomial $H_k\coloneqq \mathcal{H}^{(\widetilde{\eta})}_k$, where $k=O\left(\norm{\vec{x}}_2\log^2\left(\frac{\norm{\vec{x}}_2}{\varepsilon_{H_k}}\right)\right)=O\left(\norm{\vec{x}}_2\log^2\left(\frac{\norm{\vec{x}}_2}{\varepsilon}\right)\right)$, satisfying
\[
\sup_{\zeta\in \left[-\frac{\Xi}{\norm{\vec{x}}_2},\frac{\Xi}{\norm{\vec{x}}_2}\right]}\left|
H_k(\zeta)-h(\zeta)
\right|
\leq \varepsilon_{\mathcal H_k}, \tagaligneq \label{non-lin-proof-Hk-approx-err1}
\]
and $\sup_{\zeta\in[-1,1]}\left|H_k(\zeta)\right|\leq 2e^{2|c|\Xi}$. The QAE circuit in \autoref{Hk-using-estimate} requires gate and qubit resources as claimed in \eqref{qae-cost-for-Hk-eta-tilde-abstract-in-non-lin-thm}, following from \eqref{qae-cost-for-Hk-eta-tilde-abstract} and $\varepsilon_{\mathcal H_k}=\Theta(\varepsilon)$.
Let $\widehat{H}_k\coloneqq H_k/(4e^{2|c|\Xi})$, so that $\sup_{\zeta\in[-1,1]}\left|\widehat{H}_k(\zeta)\right|\leq \frac{1}{2}$, and hence we can use this polynomial $\widehat{H}_k$ to perform QSVT (\autoref{lem-qsvt}). 
Let $\varepsilon_{\varrho}
\coloneqq \frac{\widetilde{\varepsilon}^2}{16k^2}$.
Recall that
$k=O\left(\norm{\vec{x}}_2\log^2\left(\frac{\norm{\vec{x}}_2}{\varepsilon}\right)\right)$ implies $\frac{\varepsilon^2}{\norm{\vec{x}}_2^2}\log^{-4}\left(\frac{\norm{\vec{x}}_2}{\varepsilon}\right) = O\left(\frac{\varepsilon^2}{k^2}\right)=O\left(\varepsilon_{\varrho}\right)$.
Also,
$\log\left(\frac{L\widetilde{\kappa}}{\varepsilon_{\varrho}}\right)=O\left(\log\left(\frac{L\widetilde{\kappa}}{\varepsilon^2}\norm{\vec{x}}_2^2\log^4\left(\frac{\norm{\vec{x}}}{\varepsilon}\right)\right)\right)=O\left(\log\left(\frac{L\widetilde{\kappa}}{\varepsilon}\norm{\vec{x}}_2\right)\right)$ 
implies
$\log^{-3}\left(\frac{L\widetilde{\kappa}}{\varepsilon}\norm{\vec{x}}_2\right)
=
O\left(\log^{-3}\left(\frac{L\widetilde{\kappa}}{\varepsilon_{\varrho}}\right)\right)$.
Taken together, we have that
\eqref{cond-for-using-chakraborty-et-al-non-lin-transform-1} implies ${\varepsilon_{U_\Sigma}}/{\widetilde{\lambda}_{\max}}=o\left(\varepsilon_{\varrho}\cdot (L\widetilde{\kappa})^{-1.5} \cdot \log^{-3}\left(\frac{L\widetilde{\kappa}}{\varepsilon_{\varrho}}\right)\right)$, which is exactly the requirement \eqref{eps-U-sig-cond-main-thm-for-gen-mat} to use \autoref{main-thm-for-gen-mat} with accuracy $\varepsilon_{\varrho}$. 
The fact that $\varepsilon_{\varrho}\in(0,2]$ holds trivially from $\widetilde{\varepsilon}\leq 1$ and $k\in\mathbb{Z}^+$. So, we can use \autoref{main-thm-for-gen-mat}
to construct $U_{\ket{\varrho},\varepsilon}^{\widetilde{q},\norm{\vec{x}}_2}\equiv U_{\ket{\varrho},\varepsilon_{\varrho}}:\ket{0}_{\mathrm{A}\mathrm{I}}\mapsto\ket{\varrho}_{\mathrm{A}\mathrm{I}}$ 
such that 
$\lVert \ket{\varrho}_{\mathrm{A}\mathrm{I}}-\ket{0}_{\mathrm{A}}\ket{x}_{\mathrm{I}} \rVert_2 \leq \varepsilon_{\varrho}$, 
where
$\mathrm{A}$ is an ancillary register of $\widetilde{q}=a_{U_\Sigma}+O\left(\log\log(L\widetilde{\kappa}/{\varepsilon_{\varrho}})\right)$ qubits.
The implementation of $U_{\ket{\varrho},\varepsilon}^{\widetilde{q},\norm{\vec{x}}_2}$ requires $T_{U_{\ket{\varrho},\varepsilon}^{\widetilde{q},\norm{\vec{x}}_2}}$ elementary gates from \eqref{cost-to-prep-ket-x-within-epx-via-main-thm} as claimed,
since $1/\varepsilon_{\varrho}=O\left(\frac{\norm{\vec{x}}_2^2}{\varepsilon^2}\log^4\left(\frac{\norm{\vec{x}}_2}{\varepsilon}\right)\right)$ implies $\log\left({1}/{\varepsilon_{\varrho}}\right)=
O\left(\log\left(\frac{\norm{\vec{x}}_2}{\varepsilon}\right)\right)$.
Denote $\mathcal{D}\coloneqq \mathrm{diag}(\ket{0}^{\otimes \widetilde{q}}\ket{x})$ and $\widetilde{\mathcal{D}}\coloneqq \mathrm{diag}(\ket{\varrho})$. 
Then, we can use \autoref{rr23-thm2-2} to prepare a $(1,\widetilde{q}+n+2,0)$-block-encoding $U_{\widetilde{\mathcal{D}}}$ of $\widetilde{\mathcal{D}}$ via $O(\widetilde{q}+n)$ elementary gates and $O(1)$ queries to the controlled-$U_{\ket{\varrho},\varepsilon}^{\widetilde{q},\norm{\vec{x}}_2}$ gate\footnote{See \autoref{footnote-controlled-U}.}, requiring a total of 
\[
T_{U_{\widetilde{\mathcal{D}}}}
=
O\left(
\widetilde{q}+n+T_{U_{\ket{\varrho}}^{\widetilde{q},\varepsilon_{\varrho}}}
\right)
\]
elementary gates. 
By \autoref{bl-en-of-tilde-D-is-bl-en-of-D},
we know that $U_{\widetilde{\mathcal{D}}}$ is also a $(1,\widetilde{q}+n+2,\varepsilon_{\varrho})$-block-encoding of $\mathcal{D}$, because $\lVert \ket{\varrho}-\ket{0}^{\otimes \widetilde{q}}\ket{x} \rVert_\infty \leq \lVert \ket{\varrho}-\ket{0}^{\otimes \widetilde{q}}\ket{x} \rVert_2 \leq \varepsilon_{\varrho}$. 
Applying QSVT (\autoref{lem-qsvt}) with $\delta_0\coloneqq {\widetilde{\varepsilon}}/{4}$ to $\mathcal{D}$ (via block-encoding $U_{\widetilde{\mathcal{D}}}$) according to the polynomial $\widehat{H}_k$, we can construct a $(1,\widetilde{q}+n+4,k\sqrt{\varepsilon_{\varrho}}+\delta_0)$-block-encoding $U_{\widehat{H}_k(\mathcal{D})}$ of $\widehat{H}_k(\mathcal{D})$, using $k$ calls to $U_{\widetilde{\mathcal{D}}}$, 
a single application of controlled-$U_{\widetilde{\mathcal{D}}}$\footnote{See \autoref{footnote-controlled-U}.}, and $O\left(k(\widetilde{q}+n)\right)$ calls to single and two-qubit elementary gates. 
In total, implementing $U_{\widehat{H}_k(\mathcal{D})}$ takes
\[
T_{U_{\widehat{H}_k(\mathcal{D})}}
=
O\left(k\left(\widetilde{q}+n+T_{U_{\widetilde{\mathcal{D}}}}\right)\right)
=
O\left(
\norm{\vec{x}}_2
\cdot
\log^2\left(\frac{1}{\varepsilon}\right)
\cdot
\left(\widetilde{q}+n
+
T_{U_{\ket{\varrho}}^{\widetilde{q},\varepsilon_{\varrho}}}
\right)\right) \tagaligneq \label{cost-for-U-hat-Hk-D}
\]
elementary gates, along with $O(\mathrm{poly}\{k,\log(1/\delta_0)\})=O\left(\mathrm{poly}\left\{\norm{\vec{x}}_2\log^2\left({\norm{\vec{x}}_2}/{\varepsilon}\right),\log\left({1}/{\varepsilon}\right)\right\}\right)$ classical computation time as preprocessing.
Denote by $\widetilde{H}_k(\mathcal{D})$ the matrix that $U_{\widehat{H}_k(\mathcal{D})}$ exactly encodes (i.e., $U_{\widehat{H}_k(\mathcal{D})}$ is a $(1,\widetilde{q}+n+4,0)$-block-encoding of $\widetilde{H}_k(\mathcal{D})$). 
Since
$\norm{H_k(\mathcal{D})-h(\mathcal{D})}=\norm{H_k(\ket{0}^{\otimes \widetilde{q}}\ket{x})-h(\ket{0}^{\otimes \widetilde{q}}\ket{x})}_\infty$ (similarly to the proof of \autoref{bl-en-of-tilde-D-is-bl-en-of-D})
and
all entries of $\ket{0}^{\otimes \widetilde{q}}\ket{x}$ lie in $\left[-{\Xi}/{\norm{\vec{x}}},{\Xi}/{\norm{\vec{x}}}\right]$, 
the QSVT block-encoding error bound for $U_{\widehat{H}_k(\mathcal{D})}$ and \eqref{non-lin-proof-Hk-approx-err1}
imply
\begin{align*}
\norm{\widetilde{H}_k(\mathcal{D})-\frac{h(\mathcal{D})}{4e^{2|c|\Xi}}}
\leq
\norm{\widetilde{H}_k(\mathcal{D})-\widehat{H}_k(\mathcal{D})}
+
\norm{\frac{H_k(\mathcal{D})-h(\mathcal{D})}{4e^{2|c|\Xi}}}
\leq
\left(k\sqrt{\varepsilon_{\varrho}}+\delta_0\right)
+
\frac{\varepsilon_{\mathcal H_k}}{4e^{2|c|\Xi}}
\leq \widetilde{\varepsilon}. \tagaligneq \label{bound-tranformed-diag-qsvt}
\end{align*}
However,
since the ancilla-based embedding introduces additional zero entries in $\ket{0}^{\otimes \widetilde{q}}\ket{x}$,
$h(\mathcal{D})=\mathrm{diag}\left(h\left(\ket{0}_{\mathrm{A}}\ket{x}_{\mathrm{I}}\right)\right)$ 
contains diagonal entries equal to $h(0)\neq 0$
at indices outside $\{1,\dots,N\}$.
If we do not address these entries properly, they can lead to an incorrectly large normalisation factor and, more importantly, to a state different from the intended target,
as can be seen by
\[
h(\ket{0}_{\mathrm{A}}\ket{x}_{\mathrm{I}})
=
h(0)\ket{0}_{\mathrm{A}}\ket{0}_{\mathrm{I}}
+
\sum_{i=1}^{N}h(\zeta_i)\ket{0}_{\mathrm{A}}\ket{i}_{\mathrm{I}}
+
\sum_{j=N+1}^{2^{n}-1}h(0)\ket{0}_{\mathrm{A}}\ket{j}_{\mathrm{I}}
+
\ket{h_x^{\perp}}_{\mathrm{A}\mathrm{I}}, \tagaligneq\label{expand-h-D}
\]
where $\ket{h_x^{\perp}}$ is an unnormalised state such that $\left(\kb{0}_{\mathrm{A}} \otimes I^{\otimes n}\right)\ket{h_x^{\perp}}=0$. Here, only the second term on the RHS of \eqref{expand-h-D} corresponds to the desired component, as for $i\in\{1,\dots,N\}$, $h(\zeta_i)=e^{c \norm{\vec{x}}\zeta_i}=e^{c x_i}$, while other terms are extraneous. 
These extraneous entries are inherited by $\widetilde{H}_k(\mathcal{D})$ as well.
To extract only the desired component, we multiply the matrix $\widetilde{H}_k(\mathcal{D})$ by $\ket{0}^{\otimes \widetilde{q}}\ket{f_N}=\ket{0}^{\otimes \widetilde{q}}\sum_{i=1}^N f_i\ket{i}/\lVert\vec{f}\rVert$, 
whose amplitudes corresponding to the first, third, and fourth terms in \eqref{expand-h-D} are zero.
The bound \eqref{bound-tranformed-diag-qsvt} guarantees that the resulting state is $\widetilde{\varepsilon}$-close to the one produced by the exact operator ${h(\mathcal{D})}/\left(4e^{2|c|\Xi}\right)$, which is
\[
\frac{h(\mathcal{D})}{4e^{2|c|\Xi}}\ket{0}^{\otimes \widetilde{q}}\ket{f_N}
=
\ket{0}^{\otimes \widetilde{q}}
\otimes
\frac{1}{\lVert\vec{f}\rVert}
\sum_{i=1}^N
f_i 
\left(
\frac{e^{c\norm{\vec{x}}_2\cdot \zeta_i}}{4e^{2|c|\Xi}}
\right)
\ket{i}
=\ket{0}^{\otimes \widetilde{q}}
\otimes
\frac{1}{\lVert\vec{f}\rVert \cdot {4e^{2|c|\Xi}}}
\left(
\vec{f}
\odot
e^{c\vec{x}}
\right). \tagaligneq \label{expand-h-D-times-plus-ket}
\]
Since $\norm{\vec{f}\odot e^{c\vec{x}}}=\sqrt{\sum_i |f_i|^2\cdot e^{2c x_i}}\in\left[e^{-|c|\Xi}\lVert \vec{f} \rVert , e^{|c|\Xi}\lVert \vec{f} \rVert \right]$, we have
\[
\norm{\frac{h(\mathcal{D})}{4e^{2|c|\Xi}}\ket{0}^{\otimes \widetilde{q}}\ket{f_N}}
=
\frac{\norm{\vec{f}\odot e^{c\vec{x}}}}{4e^{2|c|\Xi}\lVert\vec{f}\rVert}
\in
\left[
\frac{1}{4e^{3|c|\Xi}}
,
\frac{1}{4e^{|c|\Xi}}
\right]. \tagaligneq \label{bounds-for-amp-of-prepared-exp-x}
\]
From \eqref{bound-tranformed-diag-qsvt}, \eqref{bounds-for-amp-of-prepared-exp-x}, and the assumption $\varepsilon\in(0,2]$, we also have\footnote{Here, we use the notation $A\pm\delta$, for $A\in\mathbb{R}$ and $\delta>0$, to mean the interval $[A-\delta,A+\delta]$.}
\begin{align*}
\norm{\widetilde{H}_k(\mathcal{D})\ket{0}^{\otimes \widetilde{q}}\ket{f_N}}
\in
\norm{\frac{h(\mathcal{D})}{4e^{2|c|\Xi}}\ket{0}^{\otimes \widetilde{q}}\ket{f_N}} \pm \widetilde{\varepsilon}
=
\frac{\norm{\vec{f}\odot e^{c\vec{x}}}}{4e^{2|c|\Xi}\lVert\vec{f}\rVert}
\pm
\frac{\varepsilon}{16e^{3|c|\Xi}}
=
\left[
\frac{1}{8e^{3|c|\Xi}},
\frac{1+2e^{2|c|\Xi}}{8e^{3|c|\Xi}}
\right]. \tagaligneq \label{bounds-for-amp-of-prepared-widehat-Hk-x}
\end{align*}
Therefore, 
we can prepare the state
\begin{align*}
U_{\widehat{H}_k(\mathcal{D})} \left(I^{\otimes (2\widetilde{q}+n+4)} \otimes U_{\ket{f_N}}\right)
\ket{0}_{\mathrm{Q}\mathrm{A}\mathrm{I}}
=\ket{0}_{\mathrm{Q}}
\widetilde{H}_k(\mathcal{D})
\left(\ket{0}_{\mathrm{A}}\ket{f_N}_{\mathrm{I}}\right)+\ket{\Psi^{\perp}}_{\mathrm{Q}\mathrm{A}\mathrm{I}}, \tagaligneq \label{state-after-apply-U-hat-Hk-D}
\end{align*}
where $\ket{\Psi^{\perp}}_{\mathrm{Q}\mathrm{A}\mathrm{I}}$ is an unnormalised state such that $\left(\kb{0}_{\mathrm{Q}}\otimes I^{\otimes(\widetilde{q}+n)}\right)\ket{\Psi^{\perp}}_{\mathrm{Q}\mathrm{A}\mathrm{I}}=0$.
We want to use QAA (\autoref{normalising-block-encoding-prep-state-via-qaa}) to prepare the normalised version of the desired part.
In order to do so, we need to know a lower bound to the target's amplitude, but this is already obtained in \eqref{bounds-for-amp-of-prepared-widehat-Hk-x}.
Accordingly,
by
denoting $\ket{\chi}\coloneqq 
\frac{\widetilde{H}_k(\mathcal{D})\ket{0}^{\otimes \widetilde{q}}\ket{f_N}}{\norm{\widetilde{H}_k(\mathcal{D})\ket{0}^{\otimes \widetilde{q}}\ket{f_N}}}$,
letting $\varepsilon_\chi\coloneqq\frac{\varepsilon}2$,
and 
letting $a_\ell\coloneqq 1/(8e^{3|c|\Xi})=\Theta(1)$ in
\autoref{normalising-block-encoding-prep-state-via-qaa}, we can construct a unitary 
$U_{\ket{\vec{f}\odot e^{c\vec{x}}},\varepsilon}
\equiv U_{\ket{\widetilde{\varphi}},\varepsilon_{\chi}}:\ket{0}_{\mathrm{Q}\mathrm{A}\mathrm{I}}\mapsto \ket{\widetilde{\varphi}}_{\mathrm{Q}\mathrm{A}\mathrm{I}}$,
such that $\norm{\ket{\widetilde{\varphi}}_{\mathrm{Q}\mathrm{A}\mathrm{I}}-\ket{0}_{\mathrm{Q}}\ket{\chi}_{\mathrm{A}\mathrm{I}}}\leq \varepsilon_\chi$.
The implementation of $U_{\ket{\vec{f}\odot e^{c\vec{x}}},\varepsilon}$ requires 
\begin{align*}
&O\left(\frac{1}{{a_\ell}}\log\left(\frac{1}{{\varepsilon_{\chi}}}\right)\left(2\widetilde{q}+2n+T_{U_{\widehat{H}_k(\mathcal{D})}}+T_{U_{\ket{f_N}}}\right)\right)\\
&=
O\left(
\norm{\vec{x}}_2
\cdot
\log^3\left(\frac{1}{\varepsilon}\right)
\cdot
\left(\widetilde{q}+n
+
T_{U_{\ket{\varrho}}^{\widetilde{q},\varepsilon_{\varrho}}}
\right)
+
\log\left(\frac{1}{{\varepsilon}}\right)
\cdot
T_{U_{\ket{f_N}}}
\right)
\end{align*}
elementary gates, where we use \eqref{cost-for-U-hat-Hk-D}.
Lastly, to prove the approximation error, \autoref{lem-for-unit-vector-approx2}
and \eqref{bound-tranformed-diag-qsvt} imply
\begin{align*}
\Big\lVert
\ket{\chi}-
\ket{0}^{\otimes \widetilde{q}}
\ket{\vec{f}\odot e^{c\vec{x}}}
\Big\rVert
=
\norm{
\frac{\widetilde{H}_k(\mathcal{D})\ket{0}^{\otimes \widetilde{q}}\ket{f_N}}{\norm{\widetilde{H}_k(\mathcal{D})\ket{0}^{\otimes \widetilde{q}}\ket{f_N}}}
-
\frac{\frac{h(\mathcal{D})}{4e^{2|c|\Xi}}\ket{0}^{\otimes \widetilde{q}}\ket{f_N}}{\norm{\frac{h(\mathcal{D})}{4e^{2|c|\Xi}}\ket{0}^{\otimes \widetilde{q}}\ket{f_N}}}
}
\leq
\frac{2\widetilde{\varepsilon}}{\norm{\frac{h(\mathcal{D})}{4e^{2|c|\Xi}}\ket{0}^{\otimes \widetilde{q}}\ket{f_N}}}.
\end{align*}
Therefore, the prepared state $\ket{\widetilde{\varphi}}$ satisfies
\begin{align*}
\Big\lVert
\ket{\widetilde{\varphi}}_{\mathrm{Q}\mathrm{A}\mathrm{I}}-\ket{0}_{\mathrm{Q}}
\ket{0}_{\mathrm{A}}
\ket{\vec{f}\odot e^{c\vec{x}}}_{\mathrm{I}}
\Big\rVert
&\leq
\Big\lVert
\ket{\widetilde{\varphi}}_{\mathrm{Q}\mathrm{A}\mathrm{I}}
-
\ket{0}_{\mathrm{Q}}\ket{\chi}_{\mathrm{A}\mathrm{I}}
\Big\rVert
+
\Big\lVert
\ket{0}_{\mathrm{Q}}\ket{\chi}_{\mathrm{A}\mathrm{I}}
-
\ket{0}_{\mathrm{Q}}
\ket{0}_{\mathrm{A}}
\ket{\vec{f}\odot e^{c\vec{x}}}_{\mathrm{I}}
\Big\rVert \\
&\leq
\varepsilon_{\chi} + \frac{2\widetilde{\varepsilon}}{\norm{\frac{h(\mathcal{D})}{4e^{2|c|\Xi}}\ket{0}^{\otimes \widetilde{q}}\ket{f_N}}} \\
&\stackrel{\eqref{bounds-for-amp-of-prepared-exp-x}}{\leq}
\frac{\varepsilon}{2}+{8\,\widetilde{\varepsilon}\, e^{3|c|\Xi}} \\
&=\varepsilon. \qedhere
\end{align*}
\end{proof}

\begin{corollary}\label{non-lin-transform-concrete-cost-thm}
Let $\varepsilon\in(0,2]$.
If \autoref{assump-plus-TU-ket-y-log2-N}, \autoref{assump-QROM-like-unitary}, and \autoref{assump-LK-plus} hold,
implementing 
$U_{\ket{\vec{f}\odot e^{c\vec{x}}},\varepsilon}:\ket{0}_{\mathrm{Q}\mathrm{A}\mathrm{I}}\mapsto\ket{\widetilde{\varphi}}_{\mathrm{Q}\mathrm{A}\mathrm{I}}$ 
such that $\bigg\lVert\ket{\widetilde{\varphi}}_{\mathrm{Q}\mathrm{A}\mathrm{I}}-\ket{0}_{\mathrm{Q}\mathrm{A}}\ket{\vec{f}\odot e^{c\vec{x}}}_{\mathrm{I}}\bigg\rVert_2 \leq\varepsilon$
by \autoref{preparing-ket-exp-x-from-ket-x-w-error} 
requires
\[
T_{U_{\ket{\vec{f}\odot e^{c\vec{x}}},\varepsilon}}
=
{O}\left(
\norm{\vec{x}}_2
\frac{\lVert \Sigma \rVert_F}{\lambda_{\max}}\,\kappa^{1.5}\,\polylog(N)
\log^6\left(\frac{\norm{\vec{x}}_2\kappa}{\varepsilon}\right)
\right) \tagaligneq \label{cost-to-prep-ket-x-within-epx-via-main-thm-concrete}
\]
elementary gate depth and $\mathbb{Z}^{+}\ni 2\widetilde{q}+n+4=O\big(\log(N)+\log\log({\norm{\vec{x}}_2\kappa}/{\varepsilon})\big)$ ancillary qubits.
For creating the description of the QSVT circuit used in $U_{\ket{\vec{f}\odot e^{c\vec{x}}},\varepsilon}:\ket{0}_{\mathrm{Q}\mathrm{A}\mathrm{I}}\mapsto\ket{\widetilde{\varphi}}_{\mathrm{Q}\mathrm{A}\mathrm{I}}$, a one-time classical computation in
$O\left(\mathrm{poly}\left\{\norm{\vec{x}}_2\log^2\left({\norm{\vec{x}}_2}/{\varepsilon}\right),\log\left({1}/{\varepsilon}\right)\right\}\right)$ time is run, 
and for determining the QSVT polynomial of transformation, a one-time QAE procedure with
\[
O\left(\frac{1}{\varepsilon}
\frac{\lVert \Sigma \rVert_F}{\lambda_{\max}}\,\kappa^{1.5}\,\polylog(N)
\log^2\left(\frac{\kappa}{\varepsilon}\right)
\right) \tagaligneq \label{qae-cost-for-Hk-eta-tilde-concrete-in-non-lin-thm-coro}
\]
elementary gate depth and an ancillary register of $\widehat{q}+2=O(\log(N)+\log\log(\kappa/\varepsilon))$ qubits is run with success probability at least $0.99$.
\end{corollary}
\begin{proof}
The proof proceeds the same way as in \autoref{preparing-ket-exp-x-from-ket-x-w-error}.
In accordance with \autoref{main-thm-gen-mat-concrete-cost},
\autoref{assump-plus-TU-ket-y-log2-N}, \autoref{assump-QROM-like-unitary}, and \autoref{assump-LK-plus} imply
$T_{U_{\ket{f_N}}}=O(\log^2 N)$ and
that the implementation of
$U_{\ket{\varrho},\varepsilon}^{\widetilde{q},\norm{\vec{x}}_2}\equiv U_{\ket{\varrho},\varepsilon_{\varrho}}:\ket{0}_{\mathrm{A}\mathrm{I}}\mapsto\ket{\varrho}_{\mathrm{A}\mathrm{I}}$ 
for 
$1/\varepsilon_{\varrho}=O\left(\frac{\norm{\vec{x}}_2^2}{\varepsilon^2}\log^4\left(\frac{1}{\varepsilon}\right)\right)$
requires
\[
T_{U_{\ket{\varrho},\varepsilon}^{\widetilde{q},\norm{\vec{x}}_2}}
=
{O}\left(
\frac{\lVert \Sigma \rVert_F}{\lambda_{\max}}\,\kappa^{1.5}\,\polylog(N)
\log^3\left(\frac{\norm{\vec{x}}_2\kappa}{\varepsilon}\right)
\right)
\]
elementary gate depth, and an ancillary register of $\widetilde{q}=O\big(\log(N)+\log\log({\norm{\vec{x}}_2\kappa}/{\varepsilon})\big)$ qubits.
Plugging these into \eqref{total-abstract-cost-non-linear-transform} proves \eqref{cost-to-prep-ket-x-within-epx-via-main-thm-concrete}.
The classical computation time required is unaffected by the addtional assumptions.
The QAE circuit needed for constructing the polynomial $H_k$ reduces to \eqref{qae-cost-for-Hk-eta-tilde-concrete-in-non-lin-thm-coro}, according to \eqref{qae-cost-for-Hk-eta-tilde-concrete} in \autoref{Hk-using-estimate}.
\end{proof}
\begin{corollary}\label{non-lin-thm-sum-of-inc}
In the setting of \autoref{Hk-using-estimate-sum-of-inc}, let $\vec{f}\in\mathbb{R}^N\backslash \{\vec{0}\}$ be a vector whose components are known.
For any $\varepsilon\in(0,2]$, we can construct a unitary $U_{\ket{\vec{f}\odot e^{c\vec{y}}},\varepsilon}:\ket{0}_{\mathrm{Q}\mathrm{A}\mathrm{I}}\mapsto\ket{\widetilde{\varphi}}_{\mathrm{Q}\mathrm{A}\mathrm{I}}$ such that $\bigg\lVert\ket{\widetilde{\varphi}}_{\mathrm{Q}\mathrm{A}\mathrm{I}}-\ket{0}_{\mathrm{Q}\mathrm{A}}\ket{\vec{f}\odot e^{c\vec{y}}}_{\mathrm{I}}\bigg\rVert \leq\varepsilon$, where 
\[
\ket{\vec{f}\odot e^{c\vec{y}}}
=
\frac{\vec{f}\odot e^{c\vec{y}}}{\norm{\vec{f}\odot e^{c\vec{y}}}}
=
\frac{1}{\sqrt{\sum_{i=1}^N |f_i|^2\cdot e^{2c y_i}}}
\sum_{i=1}^N f_i \cdot e^{c y_i} \ket{i},
\]
$\vec{y}=\mathcal{L}_N\vec{x}$ for a realisation $\vec{x}=\Sigma^{1/2}\vec{z}$ of $X\sim\mathcal{N}(0,\Sigma)$,
$\mathrm{I}$ is a system register of $n=\ceil{\log_2(N+1)}$ qubits, $\mathrm{Q}$ is an ancillary register of $\widetilde{q}+n+4$ qubits, $\mathrm{A}$ is an ancillary register of $\widetilde{q}$ qubits, and $\mathbb{Z}^{+}\ni\widetilde{q}=O\big(\log(N)+\log\log({\norm{\vec{y}}_2\kappa}/{\varepsilon})\big)$. 
The implementation of $U_{\ket{\vec{f}\odot e^{c\vec{y}}},\varepsilon}$ requires
\[
T_{U_{\ket{\vec{f}\odot e^{c\vec{y}}},\varepsilon}}
=
{O}\left(
\norm{\vec{y}}_2
\frac{\lVert \Sigma \rVert_F}{\lambda_{\max}}\,\kappa^{1.5}\,N\polylog(N)
\log^6\left(\frac{\norm{\vec{y}}_2\kappa}{\varepsilon}\right)
\right) \tagaligneq \label{cost-to-prep-ket-x-within-epx-via-main-thm-concrete-sum-of-inc}
\]
elementary gate depth.
For creating the description of the QSVT circuit used in $U_{\ket{\vec{f}\odot e^{c\vec{y}}},\varepsilon}$, 
a one-time classical computation in $O\left(\mathrm{poly}\left\{\norm{\vec{y}}_2\log^2\left({\norm{\vec{y}}_2}/{\varepsilon}\right),\log\left({1}/{\varepsilon}\right)\right\}\right)$ time is run, and for determining the QSVT polynomial of transformation, a one-time QAE procedure with
\[
O\left(\frac{1}{\varepsilon}
\frac{\lVert \Sigma \rVert_F}{\lambda_{\max}}\,\kappa^{1.5}\,N\polylog(N)
\log^2\left(\frac{\kappa}{\varepsilon}\right)
\right) \tagaligneq \label{qae-cost-for-Hk-eta-tilde-concrete-in-non-lin-thm-coro-sum-of-inc}
\]
elementary gate depth and an ancillary register of $\widehat{q}+2=O(\log(N)+\log\log(\kappa/\varepsilon))$ qubits is run with success probability at least $0.99$.
\end{corollary}
\begin{proof}
The proof proceeds in the same way as that of \autoref{preparing-ket-exp-x-from-ket-x-w-error}, but
by replacing \autoref{Hk-using-estimate}
with \autoref{Hk-using-estimate-sum-of-inc} for the construction of $H_k$,
and replacing \autoref{main-thm-for-gen-mat}
with \autoref{main-thm-coro-concrete-cost-sum-of-inc}
for the implementation of
$U_{\ket{\varrho},\varepsilon}^{\widetilde{q},\norm{\vec{y}}_2}\equiv U_{\ket{\varrho},\varepsilon_{\varrho}}:\ket{0}_{\mathrm{A}\mathrm{I}}\mapsto\ket{\varrho}_{\mathrm{A}\mathrm{I}}$ 
such that 
$\lVert \ket{\varrho}_{\mathrm{A}\mathrm{I}}-\ket{0}_{\mathrm{A}}\ket{y}_{\mathrm{I}} \rVert_2 \leq \varepsilon_{\varrho}$. The rest of the proof follows directly just by replacing $\vec{x}$ in the proof of \autoref{preparing-ket-exp-x-from-ket-x-w-error} with $\vec{y}$, and the costs \eqref{cost-to-prep-ket-x-within-epx-via-main-thm-concrete-sum-of-inc} and \eqref{qae-cost-for-Hk-eta-tilde-concrete-in-non-lin-thm-coro-sum-of-inc} follow similarly to \autoref{non-lin-transform-concrete-cost-thm}.
\end{proof}
\begin{remark}
The classical computation in $O(\mathrm{poly}\{k,\log(1/\delta)\})$ time (due to \autoref{lem-qsvt})
that arises in \autoref{preparing-ket-exp-x-from-ket-x-w-error}, 
\autoref{non-lin-transform-concrete-cost-thm},
and \autoref{non-lin-thm-sum-of-inc},
where $k\equiv \widetilde{O}\left(\norm{\vec{x}}\right)$ or $\widetilde{O}\left(\norm{\vec{y}}\right)$
and $\delta \equiv \Theta(\varepsilon)$,
is for the rotation-angle-finding procedure inherent in QSVT and all other variants of QSP.
This classical pre-processing is not a trivial task and has expanded into its own sub-field of study.
The method proposed in \cite{NY24}, when applied to our case, has complexity $\widetilde{O}(k^2+k\log(1/\delta))$.
By \autoref{gaussian-sample-path-2-norm-sqrt-N}, if ${\vec{x}}$ or ${\vec{y}}$ represents a sample path of an fBM or fOU process, its norm is concentrated around $O(\sqrt{N})$, and hence this implies $\widetilde{O}(N\log(1/\varepsilon))$ classical computation time, which has a strictly better dependence on $N$ than that of the classical counterpart to our algorithms, which is $O(N^2)$ for the matrix-vector multiplication or $O(N^3)$ for the Cholesky decomposition. For a more detailed review on this topic, see e.g.\ \cite[Section 1.3]{NY24} or the survey \cite{Ske25}.
\end{remark}
\subsubsection{A rough Bergomi variance or volatility sample path}
\begin{example}\label{rBer-ket-example1}
Let $T>0$, 
and $\mathcal{P}^N_{[0,T]}=(0=t_0<t_1<\dots<t_N=T)$ be a time discretisation of $[0,T]$.
If \autoref{assump-plus-TU-ket-y-log2-N}, \autoref{assump-QROM-like-unitary}, and \autoref{assump-LK-plus} hold,
by using \autoref{non-lin-transform-concrete-cost-thm} or \autoref{non-lin-thm-sum-of-inc},
we can construct a unitary $U_{\ket{\upsilon}}$ that prepares a quantum state
\[
\ket{\upsilon}=\frac{\vec{\upsilon}}{\norm{\vec{\upsilon}}}=\frac{1}{\sqrt{\sum_{i=1}^N \upsilon_i^2}}\sum_{i=1}^N \upsilon_i \ket{i}, 
\]
to within $\varepsilon>0$ in $\ell_2$-norm,
where $\vec{\upsilon}$ represents
a sample path of the rough Bergomi variance process $V$, or alternatively the rough Bergomi volatility process $\sqrt{V}$, with Hurst index $H\in\left(0,1\right)$ according to \autoref{rBer-S-V-defn},
on $\mathcal{P}^N_{[0,T]}$. 
The implementation of $U_{\ket{\upsilon}}$ has resource requirements as given in \autoref{non-lin-transform-concrete-cost-thm} or \autoref{non-lin-thm-sum-of-inc}.
In particular,
if we want a sample path $\vec{\upsilon}$ on the uniform grid\footnote{The proposed algorithms do not require a uniform grid to work, but we have only the simulation results for $\norm{\Sigma}_F/\lambda_{\max}$ and $\kappa$ over a uniform grid presented in this work.} $\widehat{\mathcal{P}}^N_{[0,T]}=(iT/N)_{i=0}^N$ over $[0,T]$, 
based on the empirical results of \autoref{subsubsec-complexity-compare-pv-ns},
define 
\[
\widetilde{p}(H)\coloneqq
\begin{cases}
\min\left\{\frac32+3H\,,\,3-3H\right\}, &\text{if $H\in\left(0,\frac12\right]$}, \\
\max\left\{1+H\,,\,-\frac12+3H\right\}, &\text{if $H\in\left[\frac12,1\right)$}.
\end{cases} \tagaligneq \label{defn-pH-rBer-ket-emp-cost}
\]
Note that $\widetilde{p}(H)\in\left[1.5,2.5\right]$ for $H\in\left(0,1\right)$.
Then, 
for $\vec{w}$ being a sample path of a Riemann-Liouville fractional Brownian motion $W^H$ on $\widehat{\mathcal{P}}^N_{[0,T]}$ that constitutes $\vec{\upsilon}$,
the implementation of $U_{\ket{\upsilon}}$ requires\footnote{See \autoref{discussion-example-complexities} for a discussion regarding the parameter $\norm{\vec{w}}_2$ and the total complexity of \eqref{cost-to-prep-ket-rBer-v}.}
\[
T_{U_{\ket{\upsilon}}}
=
{O}\left(
\norm{\vec{w}}_2
N^{\widetilde{p}(H)}\polylog(N)
\log^6\left(\frac{\norm{\vec{w}}_2}{\varepsilon}\right)
\right) \tagaligneq \label{cost-to-prep-ket-rBer-v}
\]
elementary gate depth and a total of $O\big(\log(N)+\log\log({\norm{\vec{w}}_2}/{\varepsilon})\big)$ qubits.
For creating the description of the QSVT circuit used in $U_{\ket{\upsilon}}$, a one-time classical computation in $O\left(\mathrm{poly}\left\{\norm{\vec{w}}_2\log^2\left({\norm{\vec{w}}_2}/{\varepsilon}\right),\log\left({1}/{\varepsilon}\right)\right\}\right)$ time is run, and for determining the QSVT polynomial of transformation, a QAE procedure with
\[
O\left(\frac{1}{\varepsilon}\cdot
N^{\widetilde{p}(H)}\polylog(N)
\log^2\left(\frac{1}{\varepsilon}\right)
\right) 
\]
elementary gate depth and an ancillary register of $O(\log(N)+\log\log(1/\varepsilon))$ qubits is run with success probability at least $0.99$.
\end{example}
\begin{proof}
From \autoref{defn-of-rl-fbm-bfg16-ver} and \autoref{rBer-S-V-defn}, it holds that
\[
V_t
=
\xi_0(t) e^{\eta {\widetilde{W}}^H_{t}-\frac{\eta^2}{2}t^{2H}},\quad\text{for $t\geq 0$,}
\]
where, in usual settings, parameters $\eta$, $H$, and the forward variance curve $\xi_0(t)$ are assumed known from market-data calibration, and hence can be used as input parameters.
Therefore,
we can set the vector $\vec{f}$ as
\[
\forall i\in\{1,\dots,N\},
\quad
f_i
\equiv
\left\{\xi_0(t_i)\right\}^{\widetilde{c}} e^{-\widetilde{c}\frac{\eta^2}{2}t_i^{2H}},
\]
where $\widetilde{c}= 1$ if we want to prepare a sample path of $V$, or $\widetilde{c}= 1/2$ if we want to prepare a sample path of $\sqrt{V}$,
and let $c \equiv \widetilde{c} \eta$ 
in \autoref{non-lin-transform-concrete-cost-thm} or \autoref{non-lin-thm-sum-of-inc}.
The covariance matrix $\Sigma$ on the time discretisation $\mathcal{P}_{[0,T]}^N$ can be calculated via \autoref{cov-of-RL-fbm}, and we have that $\vec{w}\equiv\vec{x}$ (or $\vec{w}\equiv\vec{y}$) will encode an exact sample path of $\widetilde{W}^{H}$ on $\mathcal{P}_{[0,T]}^N$.
Whether we want to use 
$\Sigma^{\mathrm{pv}}$ (via \autoref{non-lin-transform-concrete-cost-thm}) 
or $\Sigma^{\mathrm{ns}}$ (via \autoref{non-lin-thm-sum-of-inc})
depends on which results in fewer computational resource requirements. 
If the time discretisation is given by the uniform grid $\widehat{\mathcal{P}}^N_{[0,T]}=(iT/N)_{i=0}^N$ on $[0,T]$, 
the empirical results from \autoref{subsubsec-complexity-compare-pv-ns} can be used\footnote{See \autoref{remark-self-sim-cost-regardless-of-T} for how the results on $[0,1]$ apply to $[0,T]$ for self-similar processes.}.
This allows us to model the parameters $\norm{\Sigma}_F/\lambda_{\max}$ and $\kappa$ as power laws in $N$.
As can be seen from \autoref{table-prep-cost-pv-case}, \autoref{table-prep-cost-ns-case}, and \autoref{cost-comparison-graph-numericals}, using $\Sigma^{\mathrm{pv}}$ (via \autoref{non-lin-transform-concrete-cost-thm})
for $H\in(0,0.25]$, and using $\Sigma^{\mathrm{ns}}$
(via \autoref{non-lin-thm-sum-of-inc}) otherwise, gives the optimal preparation cost for $\widetilde{W}^H$, hence the definition of $\widetilde{p}(H)$ in \eqref{defn-pH-rBer-ket-emp-cost}.
\end{proof}
\begin{remark}\label{remark-self-sim-cost-regardless-of-T}
The choice of maturity $T>0$ has no effect on the parameters $\lVert \Sigma \rVert_F/\lambda_{\max}$ and $\kappa$ that determine the complexity required to perform the algorithm \autoref{non-lin-transform-concrete-cost-thm} or \autoref{non-lin-thm-sum-of-inc}, 
because,
by self-similarity of RL-fBM (see e.g.\ \cite{Lim01}), we have $\ee{\widetilde{W}^H_{Ts}\widetilde{W}^H_{Tt}}=T^{2H}\ee{\widetilde{W}^H_{s}\widetilde{W}^H_{t}}$. 
This means that a simulation on any grid $\widetilde{\mathcal{P}}_{[0,T]}^N$ is equivalent to a simulation (of a scaled process) on some grid $\widetilde{\mathcal{P}}_{[0,1]}^N$ on $[0,1]$ with covariance matrix $\Sigma^{[0,T]}\equiv T^{2H}\Sigma^{[0,1]}$, provided that $\widetilde{\mathcal{P}}_{[0,T]}^N\equiv T\widetilde{\mathcal{P}}_{[0,1]}^N$. Note that, even though $\lambda_{\max}$, $\lambda_{\min}$, and $\lVert \Sigma \rVert_F$ will be scaled by the same factor $T^{2H}$, 
the ratios $\lVert \Sigma \rVert_F/\lambda_{\max}$ and $\kappa \equiv \lambda_{\max}/\lambda_{\min}$ 
will have the factor $T^{2H}$ cancelled out, and hence their contributions to the implementation cost are unaffected, regardless of the value of $T$.
However, a probabilistic contribution of $T$ to the required complexity may arise through the (random) factor $\norm{\vec{x}}_2$ or $\norm{\vec{y}}_2$ since $\widetilde{W}^{H}_{Tt}\stackrel{d}{=}T^H\widetilde{W}^{H}_{t}$, which equates to a sublinear dependence on $T$. Nevertheless, in practical simulations, the maturity $T$ is usually treated as a fixed constant.
The same principle applies to all other self-similar Gaussian processes.
\end{remark}
\subsection{Using QAE to extract discrete sums and estimate norm}\label{qae-to-estimate-riemann-sum-time-int-subsec}
We first present how to use the QAE procedure to estimate the discrete sum of the form $\frac{1}{N}\sum_{i=1}^N |f_i| e^{cx_i}$, i.e.\ when the coefficients are restricted to non-negative real numbers, in \autoref{QAE-for-time-int-sum-abs-f}. We then use this result to develop the general case $\frac{1}{N}\sum_{i=1}^N f_i e^{cx_i}$ for $\vec{f}\in\mathbb{R}^N\backslash\{\vec{0}\}$ in \autoref{QAE-for-time-int-sum-gen-f}. 

\begin{proposition}
\label{QAE-for-time-int-sum-abs-f}
Under the setting of \autoref{preparing-ket-exp-x-from-ket-x-w-error},
we can output an estimate to the quantity $\frac1N\sum_{i=1}^N |f_i| e^{c x_i}$ to within $\widehat{\varepsilon}>0$ with probability at least $0.99$
by using a QAE procedure with
\[
O\left(
\frac{\lVert\vec{f}\rVert_\infty}{\widehat{\varepsilon}}
\left(
\norm{\vec{x}}_2
\cdot
\log^2\left(\frac{\lVert\vec{f}\rVert_\infty}{\widehat{\varepsilon}}\right)
\cdot
\left(\widetilde{q}+n
+
T_{U_{\ket{\varrho},\varepsilon}^{\widetilde{q},\norm{\vec{x}}_2}}
\right)
+
T_{U_{\ket{g_N}}}
\right\}
\right) \tagaligneq \label{qae-cost-time-int-approx-1-abs-ver}
\]
elementary gates and
$2\widetilde{q}+n+6$ ancillary qubits,
where 
$\mathbb{Z}^{+}\ni\widetilde{q}=a_{U_\Sigma}+O\left(\log\log(\lVert\vec{f}\rVert_\infty\norm{\vec{x}}_2L\widetilde{\kappa}/{\widehat{\varepsilon}})\right)$, $T_{U_{\ket{g_N}}}$ denotes gate complexity for implementing $\ket{g_N}=\vec{g}/\norm{\vec{g}}$, $\vec{g}= (\sqrt{|f_1|},\dots,\sqrt{|f_N|})$,
and $T_{U_{\ket{\varrho},\varepsilon}^{\widetilde{q},\norm{\vec{x}}_2}}$ is 
\begin{gather*}
O\left(\sqrt{L\widetilde{\kappa}}
\cdot
\log\left(\frac{\norm{\vec{x}}_2\lVert\vec{f}\rVert_{\infty}}{\widehat{\varepsilon}}\right)
\cdot
\Biggl(
\widetilde{q}
+
n+T_{U_{\ket{z}}}+\frac{\alpha_{U_\Sigma}}{\widetilde{\lambda}_{\max}}\,L\widetilde{\kappa}\,\Big(a_{U_\Sigma}+T_{U_\Sigma}\Big)\log^2\left(\frac{\norm{\vec{x}}_2 \lVert\vec{f}\rVert_{\infty} L\widetilde{\kappa}}{\widehat{\varepsilon}}\right)\Biggr)
\right),
\end{gather*}
which results from \eqref{cost-to-prep-ket-x-within-epx-via-main-thm} with choice $\varepsilon\equiv\Theta\left(\widehat{\varepsilon}/\lVert\vec{f}\rVert_{\infty}\right)$.
For creating the description of the QSVT circuit used in the oracle for the QAE procedure, a one-time classical computation in $O\left(\mathrm{poly}\left\{\norm{\vec{x}}_2\log^2\left({\norm{\vec{x}}_2\lVert\vec{f}\rVert_\infty}/{\widehat{\varepsilon}}\right),\log\left({\lVert\vec{f}\rVert_\infty}/{\widehat{\varepsilon}}\right)\right\}\right)$ time is run, and for determining the QSVT polynomial of transformation, a one-time prerequisite QAE procedure, whose complexity is additive and upper-bounded by \eqref{qae-cost-time-int-approx-1-abs-ver}, is run with success probability at least $0.99$.
In particular, if \autoref{assump-plus-TU-ket-y-log2-N}, \autoref{assump-QROM-like-unitary}, and \autoref{assump-LK-plus} hold, the required elementary gate depth and ancillary qubit number become
\[
{O}\left(
\frac{\lVert\vec{f}\rVert_\infty\norm{\vec{x}}_2}{\widehat{\varepsilon}}
\frac{\lVert \Sigma \rVert_F}{\lambda_{\max}}\,\kappa^{1.5}\polylog(N)
\log^5\left(\frac{\lVert\vec{f}\rVert_\infty\norm{\vec{x}}_2\,\kappa}{\widehat{\varepsilon}}\right)
\right) \tagaligneq \label{qae-cost-time-int-approx-1-concrete-ver}
\]
and $2\widetilde{q}+n+6=O\big(\log(N)+\log\log({\lVert\vec{f}\rVert_\infty\norm{\vec{x}}_2\kappa}/\widehat{\varepsilon})\big)$, respectively.
\end{proposition}
\begin{proof}
We use \autoref{preparing-ket-exp-x-from-ket-x-w-error}, but with parameters $\varepsilon\coloneqq 
\min\left\{2,\frac{1}{2e^{0.5|c|\Xi}}
\frac{1}{\lVert\vec{f}\rVert_\infty}
\frac{\widehat{\varepsilon}}{2}\right\}$, $c$ being $c/2$, and $\vec{f}$ being $\vec{g}= (\sqrt{|f_1|},\dots,\sqrt{|f_N|})$ instead.
That is, we let $h(\zeta)\coloneqq e^{\frac{c}{2}\norm{\vec{x}}_2\cdot\zeta}$, and construct $U_{\widehat{H}_k(\mathcal{D})}$ according to this $h$. 
The one-time prerequisite QAE procedure needed for the construction of $H_k$ has complexity as given by \eqref{qae-cost-for-Hk-eta-tilde-abstract-in-non-lin-thm}, but when compared to \eqref{cost-to-prep-ket-x-within-epx-via-main-thm}, it is obvious that this is upper-bounded by \eqref{qae-cost-time-int-approx-1-abs-ver}.
Considering the expression in \eqref{expand-h-D-times-plus-ket},
let us denote 
\[
\Upsilon
\coloneqq
\norm{\frac{h(\mathcal{D})}{4e^{2\left\lvert\frac{c}{2}\right\rvert\Xi}}\ket{0}^{\otimes \widetilde{q}}\ket{g_N}}^2
=\frac{1}{16e^{2|c|\Xi}}\frac{1}{\lVert \vec{g} \rVert^2}\sum_{i=1}^N g_i^2 \cdot e^{2\left(\frac{c}{2}\right)x_i}
=\frac{1}{16e^{2|c|\Xi}}\frac{1}{\sum_{i=1}^N |f_i|}\sum_{i=1}^N |f_i| \cdot e^{c x_i}
\]
to be the target squared amplitude we want to estimate.
Let
$\varepsilon_{\mathrm{QAE}}\coloneqq \frac{1}{16e^{2|c|\Xi}}
\frac{1}{\lVert\vec{f}\rVert_\infty}
\frac{\widehat{\varepsilon}}{2}$.
Similarly to the proof of \autoref{qae-for-main-thm},
by adding a single-qubit flag register to the state \eqref{state-after-apply-U-hat-Hk-D} and applying an $\mathrm{MC}X$ gate on the flag register conditioned on the register $\mathrm{Q}$ being in $\ket{0}_{\mathrm{Q}}$, we can apply QAE (\autoref{lem-QAE}) to output an estimate $\widetilde{p}$ for the quantity $p\coloneqq \norm{\widetilde{H}_k(\mathcal{D})\ket{0}^{\otimes \widetilde{q}}\ket{g_N}}^2$, satisfying $|\widetilde{p}-p|\leq \varepsilon_{\mathrm{QAE}}$.
The unitary operator that QAE needs to query is 
$U\coloneqq (\mathrm{MC}_{\mathrm{Q}}X\otimes I^{\otimes (\widetilde{q}+ n)}) (I\otimes U_{\widehat{H}_k(\mathcal{D})}) \left(I^{\otimes (2\widetilde{q}+n+5)} \otimes U_{\ket{g_N}}\right)$. 
Again,
by \autoref{implement-q-qubit-toffoli}, $\mathrm{MC}X$ controlled on $\widetilde{q}+n+4$ qubits costs $O(\widetilde{q}+n)$ two-qubit gates and a single-qubit ancillary register. 
Therefore, the system including the ancillas has a total of $2\widetilde{q}+2n+6$ qubits, and the QAE procedure requires
\[
O\left(\frac{1}{\varepsilon_{\mathrm{QAE}}}\left(\widetilde{q}+n+T_{U_{\widehat{H}_k(\mathcal{D})}}+T_{U_{\ket{g_N}}}\right)\right) 
\]
elementary gates in total. Substituting \eqref{cost-for-U-hat-Hk-D} for $T_{U_{\widehat{H}_k(\mathcal{D})}}$ proves \eqref{qae-cost-time-int-approx-1-abs-ver}, and the classical computation time requirement follows similarly to the discussion under \eqref{cost-for-U-hat-Hk-D}.
To prove the error bound,
note that if $\vec{a},\vec{b}$ satisfy $\lVert\vec{a}\rVert,\lVert\vec{b}\rVert\leq 1$ and $\lVert\vec{a}-\vec{b}\rVert\leq \delta$, 
then $\left|\lVert\vec{a}\rVert^2-\lVert\vec{b}\rVert^2\right|\leq 2\delta$. 
Since \eqref{bounds-for-amp-of-prepared-exp-x} and \eqref{bounds-for-amp-of-prepared-widehat-Hk-x} guarantee that $p,\Upsilon\leq 1$, 
\eqref{bound-tranformed-diag-qsvt}\footnote{Recall that $\widetilde{H}_k$ is now created with parameter $c$ being $c/2$.} 
implies
$\left|p-\Upsilon \right|\leq 2\widetilde{\varepsilon}\equiv{\varepsilon}/(8e^{3\left\lvert\frac{c}{2}\right\rvert\Xi})$, where $\widetilde{\varepsilon}$ is defined the same way as in \autoref{preparing-ket-exp-x-from-ket-x-w-error}.
The estimate $\widetilde{p}$ then satisfies
$
\big|
\widetilde{p}-\Upsilon
\big|
\leq 
\big|
\widetilde{p}-p
\big|
+
\big|
p-\Upsilon
\big|
\leq \varepsilon_{\mathrm{QAE}}+{\varepsilon}/(8e^{\frac32|c|\Xi})
\leq
\widehat{\varepsilon}
/(16e^{2|c|\Xi}\lVert\vec{f}\rVert_\infty)$ by definitions of 
$\varepsilon_{\mathrm{QAE}}$ and ${\varepsilon}$.
That is,
\begin{align*}
\left|
\mathcal{C}
\widetilde{p}
-
\left(
\frac{1}{N}\sum_{i=1}^N |f_i| e^{c x_i}
\right)
\right|
=\mathcal{C}\big|\widetilde{p}-\Upsilon\big|
\stackrel{(\dagger)}{\leq} \widehat{\varepsilon}, \tagaligneq \label{discrete-sum-qae-err-bound-ineq}
\end{align*}
where $\mathcal{C} \coloneqq 16e^{2|c|\Xi}
\frac{\norm{\vec{g}}^2}{N}$, and $(\dagger)$ follows from $\norm{\vec{g}}^2=\sum_{i=1}^N|f_i|\leq N\lVert\vec{f}\rVert_\infty$.
To prove
\eqref{qae-cost-time-int-approx-1-concrete-ver}, 
we proceed analogously to the proof of
\autoref{non-lin-transform-concrete-cost-thm}.
\end{proof}
\begin{corollary}\label{QAE-for-time-int-sum-gen-f}
Under the setting of \autoref{preparing-ket-exp-x-from-ket-x-w-error}, we can estimate $\frac1N\sum_{i=1}^N f_i e^{c x_i}$ to within $\widehat{\varepsilon}>0$ using the same resource requirements as in \autoref{QAE-for-time-int-sum-abs-f}.
\end{corollary}
\begin{proof}
Denote $u^+\coloneqq \max\{u,0\}$ and $u^{-}\coloneqq -\min\{u,0\}$, so that $u^{+},u^{-}\geq 0$ and $u=u^{+}-u^{-}$ hold for any $u\in\mathbb{R}$. For any given $\mathbb{R}^N\ni\vec{f}=(f_1,\dots,f_N)$, define $\vec{f}_{+}\coloneqq (f_i^{+})_{i=1}^N$ and $\vec{f}_{-}\coloneqq (f_i^{-})_{i=1}^N$, so that we have $\vec{f}=\vec{f}_{+}-\vec{f}_{-}$ and that all components of $\vec{f}_{+}$ and $\vec{f}_{-}$ are non-negative.
Denote $F^{+}\coloneqq \frac{1}{N}\sum_{i=1}^N f_i^{+} e^{cx_i}$, $F^{-}\coloneqq \frac{1}{N}\sum_{i=1}^N f_i^{-} e^{cx_i}$, and $F\coloneqq \frac{1}{N}\sum_{i=1}^N f_i e^{cx_i}$.
Note that $F$ is the target quantity we want to estimate, and $F=F^{+}-F^{-}$ holds.
Then, we can output estimates $G^{+}$ and $G^{-}$ such that $\left|G^{+}-F^{+}\right|\leq \widehat{\varepsilon}/2$ and $\left|G^{-}-F^{-}\right|\leq \widehat{\varepsilon}/2$ by using \autoref{QAE-for-time-int-sum-abs-f} twice. 
Note that this does not change the overall asymptotic cost.
Finally, the estimates $G^{+}$ and $G^{-}$ satisfy
$
\left|
\left(
G^{+}-G^{-}
\right)
-F
\right|
\leq
\left|G^{+}-F^{+}\right|+\left|G^{-}-F^{-}\right|
\leq
\widehat{\varepsilon}$.
\end{proof}

For clarity, we summarise as \autoref{alg-qae-for-discrete-sum} the steps for estimating $\frac1N\sum_{i=1}^N f_i e^{c x_i}$ described in the proofs of \autoref{QAE-for-time-int-sum-abs-f} and \autoref{QAE-for-time-int-sum-gen-f}.

\begin{algorithm}[H]
\caption{Outputting the discrete sum $\frac1N\sum_{i=1}^N f_i e^{cx_i}$ for a sample $\vec{x}$ via QAE}\label{alg-qae-for-discrete-sum}
\begin{algorithmic}[1]
\Require $\widetilde{\lambda}_{\max},\widetilde{\kappa},L,K,U_{\Sigma},U_{\ket{z}}$ as in \autoref{main-thm-for-gen-mat}, $U_{\ket{g_N^+}},U_{\ket{g_N^-}}$, where $\ket{g_N^{\pm}}\coloneqq \vec{g}_{\pm}/\norm{\vec{g}_{\pm}}$, $\vec{g}_{\pm}=(g_i^{\pm})_{i=1}^N$, and $g_i^{\pm}\coloneqq \sqrt{f_i^{\pm}}$, as in the proof of \autoref{QAE-for-time-int-sum-gen-f} (through the setting of \autoref{QAE-for-time-int-sum-abs-f}), and $\norm{\vec{z}}\in(0,\infty)$ for some $\vec{z}$ being a realistion of $Z\sim\mathcal{N}(0,I_{N})$.
\Ensure a classical output that is $\widehat{\varepsilon}$-close to $\frac{1}{N}\sum_{i=1}^N f_i e^{c x_i}$.
\Statex \hrulefill
\Statex (One-time QAE for defining QSVT polynomial of transformation)
\State Do steps 1--4 of \autoref{alg-exponentiation} (a one-time QAE whose complexity is as in \autoref{Hk-using-estimate}), but with parameter $c$ being $c/2$, to construct the QSVT polynomial of transformation $\widehat{H}_k$, which is a scaled version of the polynomial $\mathcal{H}^{(\widetilde{\eta})}_k(\zeta)$ that approximates $\zeta\mapsto e^{\frac{c}{2}\norm{\vec{x}} \cdot\zeta}$ on the interval $\left[-\frac{\Xi}{\norm{\vec{x}}_2},\frac{\Xi}{\norm{\vec{x}}_2}\right]$.
\Statex \hrulefill
\Statex (Preparing the QSVT building-block unitary)
\State Do steps 5--6 of \autoref{alg-exponentiation}
to construct a unitary $U_{\widetilde{\mathcal{D}}}$, which is a block-encoding of $\mathcal{D}\equiv \mathrm{diag}(\ket{0}^{\otimes \widetilde{q}}\ket{x})$.
\Statex \hrulefill
\Statex (Performing QSVT to construct the block-encoding $U_{\widehat{H}_k(\mathcal{D})}$)
\State Do steps 7--8 of \autoref{alg-exponentiation} to create $U_{\widehat{H}_k(\mathcal{D})}$, which is a block-encoding of $\widehat{H}_k(\mathcal{D})$. This is done by performing QSVT according to the polynomial $\widehat{H}_k$ and by using $U_{\widetilde{\mathcal{D}}}$ as the building block for the QSVT circuit. Then, the composite unitary $\widetilde{U}\coloneqq U_{\widehat{H}_k(\mathcal{D})} \left(I^{\otimes (2\widetilde{q}+n+4)} \otimes U_{\ket{g_N^{\pm}}}\right)$ prepares a state close to the direction of $\ket{0}^{\otimes \widetilde{q}}(e^{\frac{c}2 \vec{x}}\odot \ket{g_N^{\pm}})$.
\Statex \hrulefill
\Statex (Main QAE procedure to output an estimate to the desired discrete sum)
\State Use QAE (\autoref{lem-QAE}) to output an estimate $\widetilde{p}_\pm$ to the quantity $\frac{\norm{\vec{g}_{\pm} \odot e^{\frac{c}2 \vec{x}}}^2}{16e^{2|c|\Xi}\norm{\vec{g}_\pm}^2}\propto\frac1N\sum_{i=1}^N f_i^{\pm} e^{cx_i}$ by using $\widetilde{U}$ to construct the oracle for the QAE (the reflection operator) according to \autoref{lem-QAE}.
\Statex \hrulefill
\Statex (Classical post-processing of the QAE estimate(s))
\State Multiply $\widetilde{p}_\pm$ by $\mathcal{C}_\pm \coloneqq 16e^{2|c|\Xi}
\frac{\norm{\vec{g}_\pm}^2}{N}$ to get estimates $G^{\pm}$ of $\frac1N\sum_{i=1}^N f_i^{\pm} e^{cx_i}$, respectively.
\State Output $G^{+} - G^{-}$ as an estimate of $\frac1N\sum_{i=1}^N f_i e^{cx_i}$.
\end{algorithmic}
\end{algorithm}

\begin{corollary}\label{output-sqrt-of-discrete-sum}
Under the setting of \autoref{preparing-ket-exp-x-from-ket-x-w-error}, we can estimate 
$\frac1{\sqrt{N}}\sqrt{\sum_{i=1}^N |f_i| e^{c x_i}}$
to within $\widehat{\varepsilon}>0$ with probability $0.99$
by using a QAE procedure with
\[
O\left(
\frac{\lVert\vec{f}\rVert_\infty^{1/2}}{\widehat{\varepsilon}}
\left(
\norm{\vec{x}}_2
\cdot
\log^2\left(\frac{\lVert\vec{f}\rVert_\infty}{\widehat{\varepsilon}}\right)
\cdot
\left(\widetilde{q}+n
+
T_{U_{\ket{\varrho},\varepsilon}^{\widetilde{q},\norm{\vec{x}}_2}}
\right)
+
T_{U_{\ket{g_N}}}
\right)
\right) \tagaligneq \label{qae-cost-time-int-approx-1-abs-ver-sqrt}
\]
elementary gates and
$2\widetilde{q}+n+6$ ancillary qubits,
where 
$\mathbb{Z}^{+}\ni\widetilde{q}=a_{U_\Sigma}+O\left(\log\log(\lVert\vec{f}\rVert_\infty\norm{\vec{x}}_2L\widetilde{\kappa}/{\widehat{\varepsilon}})\right)$, $T_{U_{\ket{g_N}}}$ denotes gate complexity for implementing $\ket{g_N}=\vec{g}/\norm{\vec{g}}$, $\vec{g}= (\sqrt{|f_1|},\dots,\sqrt{|f_N|})$,
and $T_{U_{\ket{\varrho},\varepsilon}^{\widetilde{q},\norm{\vec{x}}_2}}$ is 
\begin{gather*}
O\left(\sqrt{L\widetilde{\kappa}}
\cdot
\log\left(\frac{\norm{\vec{x}}_2\lVert\vec{f}\rVert_{\infty}}{\widehat{\varepsilon}}\right)
\cdot
\Biggl(
\widetilde{q}
+
n+T_{U_{\ket{z}}}+\frac{\alpha_{U_\Sigma}}{\widetilde{\lambda}_{\max}}\,L\widetilde{\kappa}\,\Big(a_{U_\Sigma}+T_{U_\Sigma}\Big)\log^2\left(\frac{\norm{\vec{x}}_2 \lVert\vec{f}\rVert_{\infty} L\widetilde{\kappa}}{\widehat{\varepsilon}}\right)\Biggr)
\right).
\end{gather*}
For creating the description of the QSVT circuit used in the oracle for the QAE procedure, a one-time classical computation in $O\left(\mathrm{poly}\left\{\norm{\vec{x}}_2\log^2\left({\norm{\vec{x}}_2\lVert\vec{f}\rVert_\infty}/{\widehat{\varepsilon}}\right),\log\left({\lVert\vec{f}\rVert_\infty}/{\widehat{\varepsilon}}\right)\right\}\right)$ time is run, and for determining the QSVT polynomial of transformation, a one-time prerequisite QAE procedure, whose complexity is additive and upper-bounded by \eqref{qae-cost-time-int-approx-1-abs-ver-sqrt}, is run with success probability at least $0.99$.
In particular, if \autoref{assump-plus-TU-ket-y-log2-N}, \autoref{assump-QROM-like-unitary}, and \autoref{assump-LK-plus} hold, the required elementary gate depth and ancillary qubit number become
\[
{O}\left(
\frac{\lVert\vec{f}\rVert_\infty^{1/2}\norm{\vec{x}}_2}{\widehat{\varepsilon}}
\frac{\lVert \Sigma \rVert_F}{\lambda_{\max}}\,\kappa^{1.5}\polylog(N)
\log^5\left(\frac{\lVert\vec{f}\rVert_\infty\norm{\vec{x}}_2\,\kappa}{\widehat{\varepsilon}}\right)
\right) \tagaligneq \label{qae-cost-time-int-approx-1-concrete-ver-sqrt}
\]
and $2\widetilde{q}+n+6=O\big(\log(N)+\log\log({\lVert\vec{f}\rVert_\infty\norm{\vec{x}}_2\kappa}/\widehat{\varepsilon})\big)$, respectively.
\end{corollary}
\begin{proof}
We can equally estimate the target amplitude $\sqrt{\Upsilon}$ instead of the squared amplitude $\Upsilon$ in the proof of \autoref{QAE-for-time-int-sum-abs-f}, since QAE can output either quantity using the same cost, cf.\ \autoref{lem-QAE}.
The rest of the proof follows in the same manner to that of \autoref{QAE-for-time-int-sum-abs-f}, but with the definitions of $\varepsilon_{\mathrm{QAE}}$ and $\varepsilon$ being $\Theta\left(\widehat{\varepsilon}/\lVert \vec{f} \rVert_{\infty}^{1/2}\right)$ instead of the original $\Theta\left(\widehat{\varepsilon}/\lVert \vec{f} \rVert_{\infty}\right)$, since we now only need to make up for the term $\sqrt{\mathcal{C}}=\Theta\left(\norm{\vec{g}}/\sqrt{N}\right)=O\left(\lVert \vec{f} \rVert_{\infty}^{1/2}\right)$ to derive the error bound as in \eqref{discrete-sum-qae-err-bound-ineq}.
This reduces the dependence on $\lVert \vec{f} \rVert_{\infty}$ in \eqref{qae-cost-time-int-approx-1-abs-ver} and \eqref{qae-cost-time-int-approx-1-concrete-ver} from $O(\lVert \vec{f} \rVert_{\infty})$ to $O(\lVert \vec{f} \rVert_{\infty}^{1/2})$, giving \eqref{qae-cost-time-int-approx-1-abs-ver-sqrt} and \eqref{qae-cost-time-int-approx-1-concrete-ver-sqrt} respectively.
Other complexities whose dependence on $\lVert \vec{f} \rVert_{\infty}$ is logarithmic are asymptotically unchanged.
\end{proof}
\begin{corollary}\label{QAE-for-time-int-sum-gen-f-sum-of-inc}
Under the setting of \autoref{non-lin-thm-sum-of-inc},
we can output an estimate to the quantity $\frac1N\sum_{i=1}^N f_i e^{c y_i}$ to within $\widehat{\varepsilon}>0$ with probability at least $0.99$
by using a QAE procedure with
\[
{O}\left(
\frac{\lVert\vec{f}\rVert_\infty\norm{\vec{y}}_2}{\widehat{\varepsilon}}
\frac{\lVert \Sigma \rVert_F}{\lambda_{\max}}\,\kappa^{1.5}\, N\polylog(N)
\log^5\left(\frac{\lVert\vec{f}\rVert_\infty\norm{\vec{y}}_2\,\kappa}{\widehat{\varepsilon}}\right)
\right) \tagaligneq \label{qae-cost-time-int-approx-sum-of-inc}
\]
elementary gate depth and $2\widetilde{q}+n+6=O\big(\log(N)+\log\log({\lVert\vec{f}\rVert_\infty\norm{\vec{y}}_2\kappa}/\widehat{\varepsilon})\big)$ ancillary qubits.
For creating the description of the QSVT circuit used in the oracle for the QAE procedure,
a one-time classical computation in $O\left(\mathrm{poly}\left\{\norm{\vec{y}}_2\log^2\left({\norm{\vec{y}}_2\lVert\vec{f}\rVert_\infty}/{\widehat{\varepsilon}}\right),\log\left({\lVert\vec{f}\rVert_\infty}/{\widehat{\varepsilon}}\right)\right\}\right)$ time is run, and for determining the QSVT polynomial of transformation, a one-time prerequisite QAE procedure, whose complexity is additive and upper-bounded by \eqref{qae-cost-time-int-approx-sum-of-inc}, is run with success probability at least $0.99$.
\end{corollary}
\begin{proof}
The proof is conducted the same way as that of \autoref{QAE-for-time-int-sum-abs-f}, but via using \autoref{non-lin-thm-sum-of-inc} instead of \autoref{preparing-ket-exp-x-from-ket-x-w-error}.
\end{proof}
\subsubsection{Calculating a time integral of a rough Bergomi variance or volatility process}
In the following example, we show how to output the discrete sum 
$\frac{T}{N} \sum_{i=1}^N \widetilde{f}_{t_i} \upsilon_{t_i}$,
which
approximates the integral $\int_0^T \widetilde{f}(t) V_t \diff t$, 
or the integral $\int_0^T \widetilde{f}(t) \sqrt{V_t} \diff t$, 
for $V$ being the rBergomi variance process, and $\widetilde{f}$ being any bounded function, e.g.\ a discount factor $\widetilde{f}(t)=e^{-rt}$ for some interest rate $r\in\mathbb{R}$. This includes the case $\widetilde{f}(t)\equiv 1$, where the integral in question becomes the time integral $\int_0^T V_t \diff t$, aka the integrated variance.
\begin{example}\label{time-int-approx-by-QAE}
Let $T>0$, $\widetilde{f}:[0,T]\to\mathbb{R}$ be a bounded function over $[0,T]$,
and $\mathcal{P}^N_{[0,T]}=(0=t_0<t_1<\dots<t_N=T)$ be a time discretisation of $[0,T]$.
If \autoref{assump-plus-TU-ket-y-log2-N}, \autoref{assump-QROM-like-unitary}, and \autoref{assump-LK-plus} hold,
by using \autoref{QAE-for-time-int-sum-gen-f} or \autoref{QAE-for-time-int-sum-gen-f-sum-of-inc},
we can output to within $\widehat{\varepsilon}>0$ the discrete sum
\[
\frac{T}{N} \sum_{i=1}^N \widetilde{f}(t_i) V_{t_i}^{\widetilde{c}}
\approx
\int_0^T \widetilde{f}(t) V_t^{\widetilde{c}} \diff t,
\]
where
$\widetilde{c}\in\{1,1/2\}$ so that $V^{\widetilde{c}}$ is either
the rough Bergomi variance process $V$, or alternatively the rough Bergomi volatility process $\sqrt{V}$, with Hurst index $H\in\left(0,1\right)$ according to \autoref{rBer-S-V-defn}.
The QAE procedure to output this classical information has resource requirements as given in \autoref{QAE-for-time-int-sum-gen-f} or \autoref{QAE-for-time-int-sum-gen-f-sum-of-inc}.
In particular,
if we want the discrete sum to be computed over the uniform grid $\widehat{\mathcal{P}}^N_{[0,T]}=(iT/N)_{i=0}^N$ on $[0,T]$, define $\widetilde{p}(H)$ as in \eqref{defn-pH-rBer-ket-emp-cost}, where we have $\widetilde{p}(H)\in\left[1.5,2.5\right]$ for all $H\in\left(0,1\right)$.
Then, for $\vec{w}$ being a sample path of a Riemann-Liouville fractional Brownian motion $W^H$ on $\widehat{\mathcal{P}}^N_{[0,T]}$ that drives $V$, the QAE procedure requires
\[
{O}\left(
\frac{
\norm{\vec{w}}_2}
{\widehat{\varepsilon}}
N^{\widetilde{p}(H)}\polylog(N)
\log^5\left(\frac{\norm{\vec{\upsilon}}_2}{\widehat{\varepsilon}}\right)
\right) \tagaligneq \label{cost-to-discrete-sum-time-int-concrete}
\]
elementary gate depth and a total of $O\big(\log(N)+\log\log({\norm{\vec{w}}_2}/{\widehat{\varepsilon}})\big)$ qubits. 
For creating the description of the QSVT circuit used in the oracle for the QAE procedure,
a one-time classical computation in
$O\left(\mathrm{poly}\left(\norm{\vec{w}}_2\log^2\left({\norm{\vec{w}}_2}/{\varepsilon}\right),\log\left({1}/{\varepsilon}\right)\right)\right)$ time is run,
and for determining the QSVT polynomial of transformation, a one-time prerequisite QAE procedure, whose complexity is additive and upper-bounded by \eqref{cost-to-discrete-sum-time-int-concrete}, is run with success probability at least $0.99$.
\end{example}
\begin{proof}
Similarly to the proof of \autoref{rBer-ket-example1},
for $H\in(0,1)$, $t\mapsto\xi_0(t)$, and a fixed $\eta>0$, all obtained from market calibration,
we set $\vec{f}=(f_i)_{i=1}^N$ as
\[
\forall i\in\{1,\dots,N\},
\quad
f_i 
\equiv
T 
\cdot
\widetilde{f}(t_i)
\cdot
\left(
\xi_0(t_i) \, e^{-\frac{\eta^2}{2}t_i^{2H}}
\right)^{\widetilde{c}},
\]
where $\widetilde{c}= 1$ if we want the integrated sample path of $V$, or $\widetilde{c}= 1/2$ if we want the integrated sample path of $\sqrt{V}$,
and let $c \equiv \widetilde{c} \eta$ 
in \autoref{QAE-for-time-int-sum-gen-f} or \autoref{QAE-for-time-int-sum-gen-f-sum-of-inc}.
Since the function $\xi_0(t)$ is continuous by \autoref{rBer-S-V-defn}, the product function $t\mapsto \xi_0(t) \exp(-{\eta^2}t^{2H}/{2})$ is bounded on $[0,T]$ (a compact interval) by the extreme value theorem\footnote{See e.g.\ \cite[4.16 Theorem]{Rud76}.}. Since the function $\widetilde{f}$ is also assumed bounded\footnote{A common choice such as $\widetilde{f}(t)=e^{-rt}$ (the discount factor) for any $r\in\mathbb{R}$ is continuous and hence bounded on a compact interval $[0,T]$. So, the assumption that $\widetilde{f}$ is bounded is not too restricting in practical settings.}, the factor $\lVert \vec{f}\rVert_{\infty}=\Theta(1)$ by construction.
The discussion about the choice between preparing the path value vector $\vec{w}\equiv\vec{x}$ using $\Sigma^{\mathrm{pv}}$ via \autoref{QAE-for-time-int-sum-gen-f} and preparing the sum of increments vector $\vec{w}\equiv\vec{y}$ using $\Sigma^{\mathrm{ns}}$ via \autoref{QAE-for-time-int-sum-gen-f-sum-of-inc} follows the same way as that of \autoref{rBer-ket-example1}.
If the desired time discretisation is the uniform grid $\widehat{\mathcal{P}}^N_{[0,T]}=(iT/N)_{i=0}^N$ on $[0,T]$, 
the empirical results from \autoref{subsubsec-complexity-compare-pv-ns} can be used, with $\widetilde{p}(H)$ as defined in \eqref{defn-pH-rBer-ket-emp-cost}, following the same argument as in \autoref{rBer-ket-example1}.
\end{proof}
\begin{remark}\label{discussion-example-complexities}
The parameter $\norm{\vec{w}}_2$ appearing in the complexities \eqref{cost-to-prep-ket-rBer-v} of \autoref{rBer-ket-example1} and \eqref{cost-to-discrete-sum-time-int-concrete} of \autoref{time-int-approx-by-QAE} is $O(\sqrt{N})$ for a fixed $T>0$ by \autoref{gaussian-sample-path-2-norm-sqrt-N} and self-similarity.
These complexities can therefore be expressed as
$\widetilde{O}(N^{\widetilde{p}(H)+1/2}\widehat{\varepsilon}^{-\widetilde{c}})$ for
$\widetilde{c}\in\{0,1\}$, 
where $\widetilde{c}=0$ corresponds to rBergomi variance/volatility-sample-path state preparation and $\widetilde{c}=1$ to QAE-based time integration of a sample path.
Since $\widetilde{p}(H)\in[1.5,2.5]$ for all $H\in(0,1)$,
these complexities range between $\widetilde{O}(N^{2}\widehat{\varepsilon}^{-\widetilde{c}})$ and $\widetilde{O}(N^{3}\widehat{\varepsilon}^{-\widetilde{c}})$, indicating a potential quantum advantage for all $H\in(0,1)$ over the classical one-time Cholesky decomposition $O(N^3)$ complexity, but not over the classical matrix-vector multiplication $O(N^2)$ complexity, in both the sample path simulation (state preparation) and time-integration tasks.
\end{remark}
This suggests a mild quantum advantage in the situation where $1/\varepsilon=O(N^{1-\delta})$ for some $\delta\in(0,1)$ and when the bottleneck of classical computation lies in the cubic complxity arising from the Cholesky decomposition.
It is worth noting that the rough Bergomi model serves here as a financially motivated and representative example within the broader class of log-normal fractional stochastic volatility models.
In particular, fractional Brownian motions have nearly degenerate covariance matrices that allow for low-rank approximation methods.
In this sense, they may not represent the strongest regime in which our proposed algorithms could demonstrate quantum advantage.
When considering the broader log-normal fractional stochastic volatility family, lower complexity may be achievable if the driving Gaussian process exhibits a better-conditioned auto-covariance structure, leading to better-conditioned covariance matrices, similarly to the discussion in \autoref{application-ex-and-post-process-subsec}.

\section{Conclusion and discussion}
\subsection{Contributions}
In this work, we derive quantum algorithms to prepare the exact simulation of a normalised correlated Gaussian vector 
$\ket{x}=\vec{x}/\norm{\vec{x}}$ and its $\norm{\vec{x}}$-compensated, normalised exponentiation in the abstract form 
$\ket{\vec{f} \odot e^{c\vec{x}}}=(\vec{f} \odot e^{c\vec{x}})/\lVert \vec{f} \odot e^{c\vec{x}} \rVert$, 
which allows for a multiplicative prefactor as well as a translation in the exponent, the form required for constructing a rough Bergomi variance.
We also analyse the setting of cumulative sums of correlated Gaussian vectors, including as a special case the path simulation via the sum-of-increments method, which may reduce the overall complexity due to the potentially better-conditioned covariance matrix of Gaussian noises.
We provide the QAE procedures required to estimate the norm $\norm{\vec{x}}$ associated with the prepared sample, and to output classical quantities such as the Riemann sum of the exponentiated correlated Gaussian. 
We present complexities both in an abstract cost model, allowing different data-loading assumptions, and under a concrete polylogarithmic data-loader assumption to allow for the comparison with classical methods whose complexities are given in terms of $N$.

Under the polylogarithmic data-loader assumption, the state-preparation procedure before exponentiation has gate-depth complexity 
$\widetilde{O}\left(\frac{\norm{\Sigma}_F}{\lambda_{\max}}\kappa^{1.5}\right)$, 
which reduces to the fastest $\widetilde{O}(\sqrt{N})$ in the well-conditioned case with $\kappa=\widetilde{O}(1)$. 
Since well-conditioned cases include dense covariance matrices, for which the classical cubic complexity of Cholesky decomposition is inevitable, this constitutes the strongest use case of our proposed algorithm.
Mildly ill-conditioned covariance matrices provide the second most favourable regime. 
This includes fractional Gaussian noises (fGNs), whose covariance matrices satisfy 
$\frac{\norm{\Sigma}_F}{\lambda_{\max}}\kappa^{1.5}=\widetilde{\Theta}(N^{1-\varepsilon})$, 
$\varepsilon\in(0,1/2]$, for $H\in(1/3,5/6)$, providing a quantum advantage over classical FFT-based methods in this parameter range. 
The exponentiation stage introduces an additional multiplicative factor of $\norm{\vec{x}}$ (up to polylogarithmic terms), resulting in gate-depth complexity 
$\widetilde{O}\left(\norm{\vec{x}}\frac{\norm{\Sigma}_F}{\lambda_{\max}}\kappa^{1.5}\right)$.

Motivated by financial applications, we provide numerical analysis for fBMs and stationary fOU processes to determine the overall dependence on $N$, including the implicit dependence through ${\norm{\Sigma}_F}/{\lambda_{\max}}$ and $\kappa$, 
and analyse the overall costs for rough Bergomi variance state preparation and the extraction of the discrete sum approximating the rough Bergomi integrated variance. 
The resulting complexity for state preparation ranges from $\widetilde{O}(N^2)$ to $\widetilde{O}(N^3)$ depending on the Hurst parameter $H\in(0,1)$, and the QAE-based classical output introduces an additional multiplicative $O(1/\varepsilon)$ factor to this cost. 
This yields a mild quantum advantage in scenarios where $1/\varepsilon=O(N^{1-\delta})$ for some $\delta\in(0,1)$ and where the classical bottleneck arises from the cubic complexity of Cholesky decomposition.

Although the fractional Brownian motion does not constitute the strongest regime of quantum advantage in our framework, it remains of practical relevance in financial modelling and therefore serves as a meaningful benchmark.
To the best of our knowledge, this work provides the first analog encoding of the exact fBM path together with its exponentiated process, as well as extensions to more general Gaussian processes.
This essentially addresses a task previously identified as an open problem in quantum finance, as highlighted in \cite[Section 6.1]{BDP23}.
We further present a complete end-to-end complexity analysis, from state preparation to classical output, in which the dependence of $\norm{\Sigma}_F/\lambda_{\max}$ and $\kappa$ on $N$ is made explicit for fBMs and stationary fOU processes. This analysis is carried out within the scope of currently available quantum techniques, whose further development is still to be expected. 
The algorithms presented are formulated in a modular manner, so that future improvements in individual subroutines directly translate into reductions in the overall complexity.
Taken together, our results establish a concrete analytic foundation for future developments in quantum finance.

\subsection{Comparison to related methods}
\subsubsection{Classical methods}

Our simulation is exact by construction and mirrors classical covariance square-root (Cholesky-based) factorisation methods. 
It therefore has advantage over classical approximation and FFT-based methods in terms of preserving exactness and universality, including the capability for simulation on non-uniform grids. 
The magnitude of the speedup depends on the parameter regime,
as discussed in \autoref{application-ex-and-post-process-subsec}, in which the most favourable setting is the class of well-conditioned and dense covariance matrices.

\subsubsection{Quantum methods}
To the best of our knowledge, the only work addressing the analog encoding of fractional stochastic processes is \cite{BDP23}. 
Their approach achieves better asymptotic complexity ($O(\polylog N)$), through a QFT-based implementation and enables coherent analog encoding suitable for quantum Monte-Carlo (QMC) output.
However, this construction entails several structural restrictions. The reliance on a discrete sine transform implemented via QFT imposes boundary constraints $Y_0=Y_T= 0$ and restricts simulations to uniform grids. Besides, their stochastic integral discretisation approach requires state-selection measurement, so the state preparation is not unitary. 
As a result, the amplitude exponentiation technique of \cite{RR23} cannot be directly applied, since it requires a state-preparation unitary or an SPBE building block.
In addition, both path-generation approaches in \cite{BDP23} require either truncated Karhunen-Lo\`{e}ve expansion or stochastic integral discretisation, therefore producing approximations rather than exact simulations.

While their coherent analog encoding allows superposition over sample paths, which further allows QMC applications, the normalisation factor becomes entangled across branches of the superposition. 
This prevents branch-wise compensation in the exponentiation stage, where the non-linear transformation depends explicitly on $\norm{\vec{x}}$ for each individual path.
For this reason, constructing a coherent analog encoding that supports $\norm{\vec{x}}$-compensated exponentiation is not straightforward and requires further methodological development.

\subsection{Possible extensions and future works}
Exploring applications (post-processing procedures) that can make use of the prepared states embedding correlated Gaussian vectors or their exponentiated counterparts to produce quantities of practical interest remains an important direction.
A natural example is the simulation of rough Bergomi price processes, $(S_t)_{t\geq 0}$, as defined in \autoref{rBer-S-V-defn}, on a quantum computer, which extends to the simulation of the associated price processes for more general volatility models within the log-normal fractional stochastic volatility family.
This could potentially provide significant progress in numerical methods for financial mathematics by bypassing the long-standing complexity bottleneck of Cholesky decomposition under appropriate conditions, and could enable new practical possibilities for quantum-enhanced financial derivative pricing, market analysis, and risk management, among others.
However, it can be more complicated than just simulating the variance/volatility process, since the drivers $W$ and $B$ (from \autoref{rBer-S-V-defn}) are correlated and hence the Cholesky decomposition of a larger covariance matrix involving the joint law of $W$ and $W^H$ (driven by $B$) is required, cf.\ \cite[Section~4]{BFG16}.
The square-root decomposition itself is not particularly problematic, since our method can sample Gaussian vectors from such a covariance matrix, provided that the condition number behaves sufficiently well. The actual difficulty is that the sampled enlarged vector is normalised, and using its components separately to reconstruct the price process is not a trivial task. We therefore leave this problem for future investigation.

\section*{Acknowledgment}

TT was supported by JST, the establishment of university fellowships towards the creation of science technology innovation, Grant No.\ JPMJSP2138.
KM is supported by MEXT Quantum Leap Flagship Program (MEXT QLEAP) Grant No.\ JPMXS0120319794, JST Moonshot R\&D Grant No.\ JPMJMS256J, and JST COI-NEXT Program Grant No.\ JPMJPF2014.

\bibliographystyle{alpha}
\bibliography{refs}

\clearpage
\appendix
\newpage
\renewcommand{\thesection}{\Alph{section}}% For Alpha numeric number
\section{Appendix}
\subsection{Proof of \autoref{prop-for-LN-times-unit-vec2} (Bounds on $\big\lVert\mathcal{L}_N\big\rVert_2$)}\label{appndx-proof-of-LN-property}
\begin{lemma}[{\cite[THEOREM 1.]{Yue05}}]\label{lem-tridiag-mat-almost-toeplitz}
Let $\mathcal{T}_N\in\mathbb{R}^{N \times N}$ be a tridiagonal matrix of the form
\[
\mathcal{T}_N
=
\begin{pmatrix*}[r]
b & c & 0 & 0 & \dots & 0 & 0 \\
a & b & c & 0 & \dots & 0 & 0 \\
0 & a & b & c & \dots & 0 & 0 \\
\dots & \dots & \dots & \dots & \dots & \dots & \dots \\
0 & 0 & 0 & 0 & \dots & a & -\sqrt{ac}+b
\end{pmatrix*},
\]
where $\sqrt{ac}\neq0$. Then, the eigenvalues $\lambda_1,\dots,\lambda_N$ of $\mathcal{T}_N$ are given by
\[
\lambda_k
=b+2\sqrt{ac}\cos\left(\frac{2k\pi}{2N+1}\right), \quad k=1,\dots,N.
\]
\end{lemma}
\begin{lemma}\label{lem-for-order-of-max-eigen-of-Lambda}
Let
\[
f(N)\coloneqq 2+2\cos\left(\frac{2N\pi}{2N+1}\right).
\]
Then, for $N\in\mathbb{N}$,
\[
\frac{2}{\pi}N
\leq
\frac{1}{\sqrt{f(N)}}
\leq
N.
\]
\end{lemma}
\begin{proof}
Let $\theta_N\coloneqq \pi/\big(2\cdot(2N+1)\big)$. It holds that $f(N)=2+2\cos\left(\pi-2\theta_N\right)=4\sin^2\left(\theta_N\right)$.
Since $N\in\mathbb{N}$, $\theta_N\in\left(0,\pi/6\right]$. Recall from elementary differential calculus that ${}^{\forall}\theta_N\in\left(0,\pi/6\right]$,
\[
\frac{3}{\pi}
\leq
\frac{\sin \theta_N}{\theta_N}
\leq
1. \tagaligneq \label{sint-t-relation}
\]
Therefore,
\[
\frac{1}{N^2}
\stackrel{(\ddagger)}{\leq}
\frac{9}{N^2}
\frac{N^2}{(2N+1)^2}
\equiv
\frac{9}{\pi^2}\cdot 4\theta_N^2
\stackrel{(\dagger)}{\leq}
f(N)
\stackrel{(\dagger)}{\leq}
4\theta_N^2
\equiv
\frac{\pi^2}{(2N+1)^2}
\leq
\frac{\pi^2}{(2N)^2},
\]
where $(\dagger)$'s hold from \eqref{sint-t-relation}, and $(\ddagger)$ holds since for $N\in\mathbb{N}$, we have ${N^2}/{(2N+1)^2}\geq 1/9$.
\end{proof}
\begin{lemma}\label{lem-2-norm-singular-val-bounds}
Let $A\in\mathbb{R}^{N\times N}$ be a positive definite matrix, admitting a singular value decomposition of the form
$A=UDV^T$,
for some orthogonal matrices $U,V$ and a positive definite diagonal matrix $D$. Then, for any $\widehat{u}\in\mathbb{R}^N$ with $\norm{\widehat{u}}=1$, we have
\[
\sigma_{\min}(A)
\leq
\norm{A\widehat{u}}
\leq \sigma_{\max}(A),
\]
where $\sigma_{\min}(A)$, $\sigma_{\max}(A)$ denote the minimal and the maximal singular values of $A$ respectively.
\end{lemma}
\begin{proof}
Observe that, since $U$ is orthogonal,
\[
\norm{A\widehat{u}}^2
=(\widehat{u}^TA^T)(A\widehat{u})
=(\widehat{u}^TVDU^T)(UDV^T\widehat{u})
=\widehat{u}^TVD^2V^T\widehat{u}
= \norm{DV^T\widehat{u}}^2.
\]
Now let us denote $\widehat{v}\coloneqq V^T\widehat{u}$.
Since $V$ is orthogonal, we also have $\norm{\widehat{v}}=(\widehat{u}^TV)(V^T\widehat{u})=\norm{\widehat{u}}=1$. Denoting the components of $\widehat{v}$ by $\widehat{v}=(v_1,\dots,v_N)$, we then have
\[
\norm{A\widehat{u}}^2
=\norm{D\widehat{v}}^2
=\sum_{i=1}^N \sigma_i^2 v_i^2
\geq (\sigma_{\min}(A))^2\sum_{i=1}^N v_i^2
=(\sigma_{\min}(A))^2 \norm{\widehat{v}}^2
=(\sigma_{\min}(A))^2,
\]
where $\sigma_i$, $i=1,\dots,N$, denote the singular values of $A$. The upper bound is proved in the same manner, by changing $\sigma_{\min}(A)$ in the above inequality with $\sigma_{\max}(A)$.
\end{proof}
\begin{lemma}[\autoref{prop-for-LN-times-unit-vec2}]
Let $N\in\mathbb{N}$ and $\mathcal{L}_N\in\mathbb{R}^{N\times N}$ be a lower triangular matrix as given in \autoref{defn-of-LN}.
Any $\widehat{u}\in\mathbb{R}^N$ with $\norm{\widehat{u}}=1$ satisfies
\[
\frac12 
\leq
\sigma_{\min}^N{(\mathcal{L}_N)}
\leq
\big\lVert\mathcal{L}_N\widehat{u}\big\rVert
\leq 
\sigma_{\max}^N{(\mathcal{L}_N)}
\leq
N,
\]
where $\sigma_{\min}^N{(\mathcal{L}_N)}$ and $\sigma_{\max}^N{(\mathcal{L}_N)}$ denote the minimal and the maximal singular values of $\mathcal{L}_N$ respectively.
\end{lemma}
\begin{proof}
As $\det(\mathcal{L}_N)=1$,
$\mathcal{L}_N$ is clearly invertible.
By basic computations, we know that
its inverse $\mathcal{L}_N^{-1}$ and the matrix $\Lambda\coloneqq\left(\mathcal{L}_N^{-1}\right)^T\mathcal{L}_N^{-1}$ are of the forms
\[\small
\mathcal{L}_N^{-1}=
\begin{pmatrix}
1 & 0 & 0 & \dots & 0 & 0 \\
-1 & 1 & 0 & \dots & 0 & 0 \\
0 & -1 & 1 & \dots & 0 & 0 \\
\vdots & \vdots & \vdots & \ddots & \vdots & \vdots \\
0 & 0 & 0 & \dots & 1 & 0 \\
0 & 0 & 0 & \dots & -1 & 1 \\
\end{pmatrix},\quad\text{and}\quad
\Lambda
=
\begin{pmatrix}
2 & -1 & 0 & 0 & \dots & 0 & 0 & 0 \\
-1 & 2 & -1 & 0 & \dots & 0 & 0 & 0 \\
0 & -1 & 2 & -1 & \dots & 0 & 0 & 0 \\
\vdots & \vdots & \vdots & \vdots & \dots & \vdots & \vdots & \vdots \\
0 & 0 & 0 & 0 & \dots & -1 & 2 & -1 \\
0 & 0 & 0 & 0 & \dots & 0 & -1 & 1 \\
\end{pmatrix},
\]
where the latter is a tridiagonal matrix of the form that matches $\mathcal{T}_N$ from \autoref{lem-tridiag-mat-almost-toeplitz} with $a=c=-1$, $b=2$, and $\sqrt{ac}=1\neq 0$. Therefore, the eigenvalues $\lambda_1^N,\dots,\lambda_N^N$ of $\Lambda$ are given by
\[
\lambda_k^N
=2+2\cos\left(\frac{2k\pi}{2N+1}\right), \quad k=1,\dots,N. \tagaligneq \label{lem-for-LN-eigens-of-Lambda}
\]
It is obvious from the behaviour of cosine function that $k=1$ and $k=N$ give the maximal and the minimal eigenvalues of $\Lambda$ respectively.
When the singular value decomposition of $\mathcal{L}_N$ is given by $\mathcal{L}_N=UDV^T$, where $U,V\in\mathbb{R}^{N\times N}$ are some orthogonal matrices and $D\in\mathbb{R}^{N\times N}$ is a diagonal matrix whose diagonal entries are the singular values of $\mathcal{L}_N$,
we have that $\mathcal{L}_N^{-1}=VD^{-1}U^T$ and $\Lambda=U\left(D^{-1}\right)^2 U^T$.
This implies that $\sigma_{\min}^N{(\mathcal{L}_N)}=(\lambda_1^N)^{-1/2}$ and $\sigma_{\max}^N{(\mathcal{L}_N)}=(\lambda_N^N)^{-1/2}$ hold.
As $N\in\mathbb{N}$, we have that the argument of the cosine function for $\lambda_1^N$ lies in the interval $[0,2\pi/3]$.
From the monotonicity of the cosine function over $[0,2\pi/3]$, we have that $\lambda_1^N \uparrow 4$ as $N\to\infty$. Therefore, it holds that $\sigma_{\min}^N{(\mathcal{L}_N)}=(\lambda_1^N)^{-1/2}\geq 1/2$ for all $N\in\mathbb{N}$ with $\sigma_{\min}^N{(\mathcal{L}_N)}\downarrow 1/2$ as $N\to\infty$. 
On the other hand, by \autoref{lem-for-order-of-max-eigen-of-Lambda}, we have that $\sigma_{\max}^N{(\mathcal{L}_N)}=(\lambda_N^N)^{-1/2}\leq N$ holds.
Finally, \autoref{lem-2-norm-singular-val-bounds} completes the proof.
\end{proof}

\subsection{Correction to \cite{BFG16} covariance structure of RL-fBM}\label{proof-of-cov-RL-fbm}
Here, we provide a mathematical derivation of \eqref{eq-for-G-in-cov-of-RL-fBM}.
Let us denote by $B$ the Beta function. Then the following integral form holds for $c>b>0$:
\[
{}_2F_1(a,b;c;z)=
\frac{1}{B(b,c-b)}\int_0^1 t^{b-1}(1-t)^{c-b-1}(1-zt)^{-a} \diff t. \tagaligneq \label{int-form-2f1-euler-rep}
\]
This is a standard result, usually referred to as the Euler integral representation of ${}_2F_1$, see e.g.\ \cite[Eq.~(1.2)]{HK74}.
Note that, even though \eqref{int-form-2f1-euler-rep} does not look to have symmetry between the first and the second arguments, i.e.\ $a,b$ of ${}_2F_1(a,b;c;z)$, they are actually symmetric (interchangeable) by definition using infinite series, cf.\ \eqref{defn-of-2f1-inf-sum}.
\begin{proof}[Proof of \autoref{cov-of-RL-fbm}.]
Recall the definition of RL-fBM, \autoref{defn-of-rl-fbm-bfg16-ver}. For $v\geq u$, we have
\[
\ee{\widetilde{W}^H_v\widetilde{W}^H_u}
=\int_0^{u}\frac{(2H)\diff s}{(v-s)^{\frac{1}{2}-H}(u-s)^{\frac{1}{2}-H}}, \tag{$\star$}\label{cov-eq-in-proof-1}
\]
because $\diff \langle B,B\rangle_s=\diff s$, $\int_0^v=\int_0^u+\int_u^v$, and Brownian increments on $[u,v]$ are independent of increments on $[0,u]$, so the term $$\ee{\left(\sqrt{2H}\int_u^v(v-s)^{H-1/2}\diff B_s\right)\widetilde{W}^H_u}=\ee{\sqrt{2H}\int_u^v(v-s)^{H-1/2}\diff B_s}\cdot\ee{\widetilde{W}^H_u}=0,$$ 
and hence \eqref{cov-eq-in-proof-1} holds.
We can further simplify \eqref{cov-eq-in-proof-1} to
\[
\eqref{cov-eq-in-proof-1}=
\frac{1}{u^{1-2H}}\int_0^u 
\frac{(2H)\diff s}
{\left(
\frac{v}{u}-\frac{s}{u}
\right)^{\frac12-H}
\left(
1-\frac{s}{u}
\right)^{\frac12-H}
}
\stackrel{(\dagger)}{=}
u^{2H}
\int_0^1
\frac{(2H)\diff t}
{\left(
\frac{v}{u}-t
\right)^{\frac12-H}
\left(
1-t
\right)^{\frac12-H}
},
\]
where $(\dagger)$ holds by the change of variables $t\coloneqq s/u$.
Here, by letting
\[
G(x)\coloneqq \int_0^1
\frac{(2H)\diff t}
{\left(
x-t
\right)^{\frac12-H}
\left(
1-t
\right)^{\frac12-H}
}, \tagaligneq \label{defn-of-G-int-form}
\]
it then holds that
\[
\ee{\widetilde{W}^H_v\widetilde{W}^H_u}
=\eqref{cov-eq-in-proof-1}= u^{2H}G\left(\frac{v}{u}\right).
\]
The rest is to show that $G$ is indeed as given in \eqref{eq-for-G-in-cov-of-RL-fBM}. That is, for $x\geq 1$,
\[
G(x)=\frac{2H}{\frac12+H}\left(\frac1x\right)^{\frac12-H}
{}_2F_1\left(\frac12-H,1;\frac32+H;\frac1x\right).
\tagaligneq \label{eq-to-prove-for-G}
\] 
Observe that, by \eqref{int-form-2f1-euler-rep},
\begin{align*}
{}_2F_1\left(\frac12-H,1;\frac32+H;\frac1x\right)
&=\frac{1}{B\left(1,\frac12+H\right)}
\int_0^1 t^0(1-t)^{-\frac12+H}\left(1-\frac{t}{x}\right)^{H-\frac12} \diff t \\
&\stackrel{(\ddagger)}{=}\left(\frac12+H\right)
\int_0^1 \frac{\diff t}
{
(1-t)^{\frac12-H}
\left(\frac{x-t}{x}\right)^{\frac12-H}
}\\
&=\left(\frac12+H\right)
\cdot
x^{\frac12-H}
\cdot
\int_0^1 \frac{\diff t}
{
(1-t)^{\frac12-H}
\left(x-t\right)^{\frac12-H}
}\\
&\stackrel{\eqref{defn-of-G-int-form}}{=}
\left(\frac12+H\right)
\cdot
x^{\frac12-H}
\cdot
\frac{1}{2H}\cdot G(x), \tagaligneq \label{eq-for-G-2f1-vs-int-form}
\end{align*}
where $(\ddagger)$ follows from $B(1,x)=1/x$, for all $x>0$\footnote{This follows easily from $B(x,y)=\frac{\Gamma(x)\Gamma(y)}{\Gamma(x+y)}$ and $\Gamma(x+1)=x\Gamma(x)$.}.
Rearranging terms in \eqref{eq-for-G-2f1-vs-int-form}, we can see that \eqref{eq-to-prove-for-G} holds. This completes the proof.
\end{proof}
\subsection{Expression of stationary fOU covariance structure using ${}_1F_2$} \label{appndx-stn-fOU-cov-1f2}
\begin{definition}[Complex argument, real part, imaginary part]
For $z\in\mathbb{C}$ whose polar coordinate form admits $z=re^{i\theta}$ for $r>0$ and $\theta \in (-\pi,\pi]$, denote $\arg z \coloneqq \theta$, $\mathrm{Re}(z)\coloneqq r\cos \theta$, and $\mathrm{Im}(z)\coloneqq r\sin \theta$.
\end{definition}
The Mellin transform is an integral transform closely related to the bilateral Laplace transform: it can be shown that the Mellin transform can be obtained by a variable substitution of the Laplace transform.
For a more detailed treatment of this subject, we refer to e.g.\ \cite{PK01}.
The following result of the Mellin transform of the function $\cos(bx)(a^2+x^2)^{-\nu}$ will be used to prove the expression of the stationary fOU covariance that involves ${}_1F_2$.
\begin{lemma}[{\cite[p.43 item 5.10]{Obe74}}]\label{lem-mellin-cosbx-reciprocal-a2-x2}
Let $\nu,z\in\mathbb{C}$ be such that $0<\mathrm{Re}(z)<1+2\mathrm{Re}(\nu)$. Then, for\footnote{The condition $a,b>0$ is specified in the preface of \cite{Obe74}.} $a,b>0$,
\begin{align*}
\int_0^\infty x^{z-1}\frac{\cos(bx)}{(a^2+x^2)^{\nu}} \diff x
=&
\frac{a^{z-2\nu}}{2}B\left(\frac{z}{2},\nu-\frac{z}{2}\right)
{}_1F_2\left(\frac{z}{2};1+\frac{z}{2}-\nu,\frac12;\frac{a^2b^2}{4}\right)\\
&+
\frac{\sqrt{\pi}\, 2^{z-2\nu-1}\, b^{2\nu-z}}{\Gamma\left(\frac12+\nu-\frac{z}{2}\right)}\Gamma\left(\frac{z}{2}-\nu\right){}_1F_2\left(\nu;\frac12+\nu-\frac{z}{2},1+\nu-\frac{z}{2};\frac{a^2b^2}{4}\right),
\end{align*}
where $B$ is the Beta function, and ${}_1F_2$ is the generalised hypergeometric function defined as
\[
{}_1F_2(a;b_1,b_2;z)\coloneqq \sum_{n=0}^\infty \frac{(a)_n}{(b_1)_n(b_2)_n}\frac{z^n}{n!},
\]
where for any $a\in\mathbb{C}$ and $n\in\mathbb{N}\cup\{0\}$, $(a)_n\coloneqq {\Gamma(a+n)}/{\Gamma(a)}$. 
\end{lemma}
\begin{remark}
\cite{Obe74} also provides the Mellin transform of the case where $\nu=1$, but there was an error in their expression ($b^{-z}$ instead of $b^{2-z}$), which can be checked easily from dimensional homogeneity of intregrals. Therefore, we use a more general form as suggested in \autoref{lem-mellin-cosbx-reciprocal-a2-x2}.
\end{remark}
\begin{lemma}[Properties of Gamma functions]\label{lem-gamma-fn-relations}
For all $z\in\mathbb{C}$, the following relations hold (including $\infty=\infty$ if either side diverges resulting from a pole from the value of $z$):
\begin{enumerate}[label=(\roman*)]
\item $\displaystyle \Gamma(z)\Gamma(1-z)=\frac{\pi}{\sin(\pi z)}$.
\item $\displaystyle \Gamma(z+1/2)=\sqrt{\pi}\,2^{-2z}\frac{\Gamma(2z+1)}{\Gamma(z+1)}$.
\end{enumerate}
\end{lemma}
\begin{proof}
The first property is what is usually referred to as the reflection formula; see e.g.\ \cite[p.164 Chapter 6 Theorem 1.4]{SS03}.
The second property is a result from another well-known formula $\Gamma(s)\Gamma(s+1/2)=\sqrt{\pi}\,2^{1-2s}\Gamma(2s)$, usually referred to as the duplication formula; see e.g.\ \cite[p.175 Chapter 6 Exercise 3]{SS03}. By letting $s\coloneqq z+1/2$, we get the relation of the second claim.
\end{proof}
\begin{lemma}[{\cite[p.164 Chapter 6 Lemma 1.5]{SS03}}]
\label{appndx-lem-beta-to-gamma-reflection}
For $0<a<1$, $\displaystyle \int_0^\infty \frac{v^{a-1}}{1+v}\diff v=\frac{\pi}{\sin(\pi a)}$.
\end{lemma}
\begin{lemma}[{\autoref{Fourier-trans-of-fOU-cov-int}}]
\label{Fourier-trans-of-fOU-cov-int-appndx}
For any $\lambda>0$, $H\in(0,1)$, and $s\geq 0$, the following holds
\begin{align*}
\widehat{f}_{H,\lambda}(s)
&\coloneqq\int_{-\infty}^{\infty} e^{isx}\frac{|x|^{1-2H}}{\lambda^2+x^2} \diff x\\
&=
\frac{\pi}{\sin(\pi H)}
\left\{\lambda^{-2H} \cosh(\lambda s)
-\frac{s^{2H}}{\Gamma(2H+1)} {}_1F_2\Bigg(1;H+\frac{1}{2},H+1;\frac{\lambda^2 s^2}{4}\Bigg)\right\}. \tagaligneq \label{lem-fOU-cov-claim-statement-eq-appndx}
\end{align*}
In particular, if $H=1/2$, $\widehat{f}_{1/2,\lambda}$ reduces to the usual stationary Ornstein-Uhlenbeck covariance structure:
\[
\widehat{f}_{1/2,\lambda}(s)=\frac{\pi}{\lambda}e^{-\lambda s}. \tagaligneq \label{lem-fOU-cov-claim-statement-eq-H-is-12-appndx}
\]
\end{lemma}
\begin{proof}
By the property of a Fourier transform of an even function, it holds that 
\[
\widehat{f}_{H,\lambda}(s)
=\int_{-\infty}^{\infty}\cos(sx) \frac{|x|^{1-2H}}{\lambda^2+x^2} \diff x
=2\int_{0}^{\infty}\cos(sx) \frac{x^{1-2H}}{\lambda^2+x^2} \diff x,
\]
whose expression coincides with the Mellin transform of $\cos(sx)(\lambda^2+x^2)^{-1}$.
For $s>0$, we can calculate this integral by letting $a\coloneqq \lambda>0$, $b\coloneqq s>0$, $\nu\coloneqq 1$, and $z\coloneqq 2-2H$ in \autoref{lem-mellin-cosbx-reciprocal-a2-x2}.
Note that the condition $0<\mathrm{Re}(z)<1+2\mathrm{Re}(\nu)=3$ required in \autoref{lem-mellin-cosbx-reciprocal-a2-x2} is satisfied, since $H\in(0,1)$ implies $z\in(0,2)$.
This gives, for $s>0$,
\begin{equation}
\begin{aligned}
\widehat{f}_{H,\lambda}(s)
=&\,
\lambda^{-2H}B\left(1-H,H\right)\,
{}_1F_2\left(1-H;1-H,\frac12;\frac{\lambda^2 s^2}{4}\right)\\
&+\frac{\sqrt{\pi}\, 2^{-2H} s^{2H}}{\Gamma\left(\frac12+H\right)}\Gamma\left(-H\right)\,
{}_1F_2\left(1;\frac12+H,1+H;\frac{\lambda^2 s^2}{4}\right).
\end{aligned}\label{f-hat-after-mellin-1}
\end{equation}
Here, note that $B(1-H,H)=\Gamma(1-H)\Gamma(H)=\pi/\sin(\pi H)$ holds by \textit{(i)} of \autoref{lem-gamma-fn-relations}.
In addition,
\begin{align*}
{}_1F_2\left(1-H;1-H,\frac12;\frac{\lambda^2 s^2}{4}\right)
&=\sum_{n=0}^{\infty}\frac{\Gamma\left(1/2\right)}{\Gamma\left(n+1/2\right)}\left(\lambda s\right)^{2n}\frac{2^{-2n}}{n!}\\
&=\sum_{n=0}^{\infty}\frac{\Gamma\left(1/2\right)}{\Gamma\left(2n+1\right)}\frac{2^{2n}\Gamma(n+1)}{\sqrt{\pi}}\left(\lambda s\right)^{2n}\frac{2^{-2n}}{n!}\\
&=\sum_{n=0}^{\infty}\frac{\left(\lambda s\right)^{2n}}{\left(2n\right)!}\\
&=\cosh\left(\lambda s\right), \tagaligneq \label{appndx-fOU-cov-lem-cosh-expansion-eq}
\end{align*}
holds from the relation $\Gamma(n+1/2)=\sqrt{\pi}\,2^{-2n}\Gamma(2n+1)/\Gamma(n+1)$ given in \textit{(ii)} of \autoref{lem-gamma-fn-relations}, and the properties $\Gamma(1/2)=\sqrt{\pi}$ and $\Gamma(n+1)=n!$.
This proves that the first term of \eqref{f-hat-after-mellin-1} matches the first term of \eqref{lem-fOU-cov-claim-statement-eq-appndx} as claimed.
For the second term of \eqref{f-hat-after-mellin-1}, again, \textit{(i)} and \textit{(ii)} of \autoref{lem-gamma-fn-relations} imply
\begin{align*}
\frac{\sqrt{\pi}\, 2^{-2H} s^{2H}}{\Gamma\left(\frac12+H\right)}\Gamma\left(-H\right)
&=\frac{s^{2H}}{\Gamma\left(2H+1\right)}\Gamma(1+H)\Gamma(-H)
=\frac{-\pi s^{2H}}{\sin(\pi H)\Gamma\left(2H+1\right)},
\end{align*}
where, in the last equality, the reflection formula $\Gamma\left(1-(-H)\right)\Gamma(-H)=\pi/\sin(-\pi H)=-\pi/\sin(\pi H)$ is applied.
Therefore, we have proved that, for $s>0$, $\widehat{f}_{H,\lambda}(s)$ equals \eqref{lem-fOU-cov-claim-statement-eq-appndx} as claimed.
When $s=0$, $\widehat{f}_{H,\lambda}(0)$ reduces to
\[
\widehat{f}_{H,\lambda}(0)
=
2\int_{0}^{\infty}\frac{x^{1-2H}}{\lambda^2+x^2} \diff x
=
\lambda^{-2H}
\int_0^{\infty} \frac{t^{-H}}{1+t}\diff t
=
\lambda^{-2H}
\frac{\pi}{\sin(\pi H)}, \tagaligneq \label{lem-appndx-fou-cov-when-s-is-0-eq}
\]
which holds from the variable substitution $t\coloneqq x^2/\lambda^2$ giving $\diff t =(2x\diff x)/\lambda^2$, applying \autoref{appndx-lem-beta-to-gamma-reflection} with $a=1-H$ (the condition $a\in(0,1)$ required in \autoref{appndx-lem-beta-to-gamma-reflection} is satisfied since $H\in(0,1)$),
and the fact that $\sin\left(\pi(1-H)\right)=\sin(\pi H)$.
It easily follows that as $s\downarrow 0$, the quantity \eqref{lem-fOU-cov-claim-statement-eq-appndx} converges to \eqref{lem-appndx-fou-cov-when-s-is-0-eq}, since $\cosh(\lambda s)$ in the first term and $s^{2H}$ in the second term of \eqref{lem-fOU-cov-claim-statement-eq-appndx} converge to $1$ and $0$, respectively. 
Therefore, the expression \eqref{lem-fOU-cov-claim-statement-eq-appndx} holds for all $s\geq 0$.
To show how the form \eqref{lem-fOU-cov-claim-statement-eq-H-is-12-appndx} for the special case $H=1/2$ can be obtained from \eqref{lem-fOU-cov-claim-statement-eq-appndx}, observe that since $\Gamma(n+1)=n\Gamma(n)$ implies $\Gamma(3/2)=\Gamma(1/2)/2$ and $\Gamma(n+3/2)=(n+1/2)\Gamma(n+1/2)$, we have that
the ${}_1F_2$ in the second term of \eqref{lem-fOU-cov-claim-statement-eq-appndx} becomes
\begin{align*}
{}_1F_2\left(1;1,\frac32;\frac{\lambda^2 s^2}{4}\right)
&=\sum_{n=0}^{\infty}\frac{\Gamma(3/2)}{\Gamma(n+3/2)}(\lambda s)^{2n}\frac{2^{-2n}}{n!}\\
&=\sum_{n=0}^{\infty}
\left(\frac{1}{2n+1}\right)\cdot
\frac{\Gamma\left(1/2\right)}{\Gamma\left(n+1/2\right)}\left(\lambda s\right)^{2n}\frac{2^{-2n}}{n!}\\
&=\sum_{n=0}^{\infty}
\left(\frac{1}{2n+1}\right)\cdot
\frac{\left(\lambda s\right)^{2n}}{\left(2n\right)!}\\
&=\sum_{n=0}^{\infty}
\frac{\left(\lambda s\right)^{2n+1}}{\left(2n+1\right)!}
\cdot \frac{1}{\lambda s}
\\
&=\frac{\sinh\left(\lambda s\right)}{\lambda s},
\end{align*}
similarly to the derivation of \eqref{appndx-fOU-cov-lem-cosh-expansion-eq}.
The rest follows since $\exp(-\lambda s)=\cosh\left(\lambda s\right)-\sinh\left(\lambda s\right)$.
\end{proof}
\subsection{Supplementary discussion on infinity-norm bound assumption} \label{appndx-inf-norm-bnd}
\begin{proposition}[{\cite[Theorem 4.2]{Nou12}}] \label{tail-bnd-gauss-proc-ivan-nourdin}
Let $G=(G_t)_{t\in[0,1]}$ be a centered and continuous Gaussian process. Let $\widehat{\sigma}^2\coloneqq \sup_{t\in[0,1]}\ee{G_t^2}$. Then $\widehat{m}\coloneqq\ee{\sup_{u\in[0,1]}G_u}$ is finite and we have, for all $\Xi>\widehat{m}$,
\[
\pp{\sup_{u\in[0,1]}G_u\geq\Xi}\leq e^{-\frac{(\Xi-\widehat{m})^2}{2\widehat{\sigma}^2}}.
\]
\end{proposition}
\autoref{tail-bnd-gauss-proc-ivan-nourdin} directly translates to the following.
\begin{proposition} \label{bound-xi-on-norm-x-infty}
Let $G,\widehat{\sigma},\widehat{m}$ be as defined in \autoref{tail-bnd-gauss-proc-ivan-nourdin}. For $\beta\in[0,1)$, and
\[
\Xi\coloneqq \widehat{m}+\widehat{\sigma}\sqrt{2\log\left(\frac{2}{1-\beta}\right)},\quad \text{we have }\quad \pp{\sup_{u\in[0,1]}\lvert G_u\rvert \leq \Xi }
\geq \beta.
\]
\end{proposition}
\begin{proof}
For $S$ being a stochastic process, denote $S^{*}\coloneqq \sup_{u\in[0,1]}S_u$ and $S_{*}\coloneqq \inf_{u\in[0,1]}S_u$.
Since $G$ is centered, $G\stackrel{d}{=}(-G)$ as both produce the same covariance matrix $\Sigma$ on arbitrary finite grids, i.e.\ $\ee{(-G_{t_i})(-G_{t_j})}=\ee{G_{t_i}G_{t_j}}$. By subadditivity of probability measures, $\pp{|G|^*>\Xi}\leq \pp{G^*>\Xi}+\pp{G_*<-\Xi} = 2 \pp{G^*>\Xi}$, since $\pp{G_*<-\Xi}=\pp{(-G)^*>\Xi}=\pp{G^*>\Xi}$.
Therefore,
$\pp{|G|^*\leq \Xi}=1-\pp{|G|^*> \Xi}\geq 1-2\pp{G^*>\Xi} \geq 1-2\exp(-(\Xi-\widehat{m})^2/2\widehat{\sigma}^2)=\beta$, where the last inequality follows from \autoref{tail-bnd-gauss-proc-ivan-nourdin}.
\end{proof}
The dependence of $\Xi$ on $\beta$ is asymptotically the same as $F_{Z_1}^{-1}(\beta)\approx \sqrt{-2\log(2\sqrt{\pi}(1-\beta))}$ when $F_{Z_1}^{-1}$ denotes the inverse of a one-dimensional standard normal CDF (see e.g.\ \cite[equation (2)]{BEJ76}\footnote{The asymptotic expansion in \cite{{BEJ76}} uses the inverse error function $\erf^{-1}$, whose relation to the inverse of a standard normal CDF is: $F_{Z_1}^{-1}(\beta)=\sqrt{2}\erf^{-1}(2\beta-1)$.}), 
which means the familiar sigma-multiple rules for $1$-dimensional standard Gaussian distribution apply. 
To accurately model even the occurrence of extreme events, largest simulations require a Gaussian random number generator to be able to sample $\pm 10 \sigma$ events which corresponds to tail probability $1-\beta \approx 1.52\times 10^{-23}$ \cite[Section 1]{TLLV07}.
Even so, the contribution of $\beta$ to $\Xi$, the term $\sqrt{2\log(2/(1-\beta))}$, is roughly $10$. 
Given that Gaussian (pseudo-)random number generators cannot generate samples that are beyond the tail probability of machine precision limits, samples that can exceed $\Xi$ given such a $\beta$ are unobservable in practice.

The factor
$\widehat{\sigma}=\sup_{t\in[0,1]}\ee{G_t^2}$, which is the supremum of variance, usually attains its supremum at $t=1$. When $G$ is std-fBM or RL-fBM, $\widehat{\sigma}=1$\footnote{The RL-fBM has to be defined with the coefficient $\sqrt{2H}$ in the fractional integral for this to hold, cf.\ \autoref{defn-of-rl-fbm-bfg16-ver}.}. 
For a stationary fOU process, the supremum of variance is $\sigma^2\lambda^{-2H}\Gamma(2H+1)/2$ when $H\in(0,1)$.
If $\lambda=\sigma=1$, then $\widehat{\sigma}=\Gamma(2H+1)/2\leq 1$ for all $H\in(0,1)$.

The factor $\widehat{m}$ is guaranteed to be of finite value as proven in the original proof of \autoref{tail-bnd-gauss-proc-ivan-nourdin} (\cite[Theorem 4.2]{Nou12}). However, the following result might be useful for getting an idea about the size of $\widehat{m}$.
\begin{proposition}[{\cite[Theorem 1 (ii)]{BMNZ15}}]
Let $G=(G_t)_{t\in[0,1]}$ be a zero-mean, continuous Gaussian process. If there exist $\mathcal{C}>0$ and $\mathcal{H}\in(0,1)$ such that $\sqrt{\ee{(G_t-G_s)^2}}\leq \mathcal{C}\lvert t-s\rvert^{\mathcal{H}}$ for any $t,s\in[0,1]$ then
\[
\ee{\sup_{u\in[0,1]}G_u}<L\mathcal{C}\sqrt{\frac{2\pi}{\mathcal{H}\log^32}} \mathrm{erfc}\sqrt{\frac{\log 2}{2}H}<16.3\frac{\mathcal{C}}{\sqrt{\mathcal{H}}},
\]
where $L<3.75$ is a constant from generic chaining bound and $\mathrm{erfc}$ is the complementary error function.
\end{proposition}
\begin{proof}
We refer the reader to equation (7) of \cite{BMNZ15} and the proof following Remark 1\&2 therein. 
\end{proof}
For the case of std-fBM, the condition $\ee{(G_t-G_s)^2}\leq \mathcal{C}^2\lvert t-s\rvert^{2\mathcal{H}}$ follows immediately from its covariance structure and stationary increments, with $\mathcal{C}=1$ and $\mathcal{H}=H$, where $H$ is the Hurst index. 
Therefore, $\widehat{m}=\ee{\sup_{u\in[0,1]}B^H_u}<16.3/\sqrt{H}$. 
For the case of RL-fBM, although it is not as straightforward, one can still find a constant $\mathcal{C}$ and have the condition satisfied similarly with $\mathcal{H}=H$, see e.g.\ \cite[equations (18),(19)]{Lim01} which derives $\mathcal{C}^2$ up to a multiplicative constant (the difference lies in the choice of the coefficient we make to define the fractional integral representation for RL-fBMs).
For the case of stationary fOU, we can expand $\ee{(Y^H_t-Y^H_s)^2}=\ee{(Y^H_t)^2}+\ee{(Y^H_s)^2}-2\ee{Y^H_tY^H_s}=2\left(\ee{(Y^H_0)^2}-\ee{Y^H_0Y^H_{|t-s|}}\right)\eqqcolon r(|t-s|)$, which follows from stationarity. 
By \autoref{cov-of-fou} and \autoref{Fourier-trans-of-fOU-cov-int},
this quantity reduces to a function of the form $r(|t-s|)= C_{\sigma,H}\left\{ \lambda^{-2H}(1-\cosh(\lambda |t-s|))+\frac{|t-s|^{2H}}{\Gamma(2H+1)}{}_1F_2\left(1;H+\frac12,H+1;\frac{\lambda^2 |t-s|^2}{4}\right) \right\}$, where $C_{\sigma,H}=\sigma^2\,\Gamma(2H+1)>0$ is a constant.
The series expansion (as in the proof of \autoref{Fourier-trans-of-fOU-cov-int-appndx}) implies that the first term of $r(|t-s|)$ is of order $O(|t-s|^2)$ and the second term is of order $O(|t-s|^{2H})$.
Since $|t-s|\in[0,1]$ and $H\in(0,1)$, the latter always dominates, and hence $r(|t-s|)=O(|t-s|^{2H})$ on this domain.
Then, we can maximise the quotient $\widetilde{r}(|t-s|)\coloneqq r(|t-s|)/|t-s|^{2H}$ over $|t-s|\in[0,1]$ and let $\mathcal{C}^2\coloneqq \sup_{\tau\in[0,1]}\widetilde{r}(\tau)$ (a standard argument). This implies that
the bound $\ee{(Y^H_t-Y^H_s)^2}\leq \mathcal{C}^2\lvert t-s\rvert^{2\mathcal{H}}$ holds with, again, $\mathcal{H}=H$, similarly to the fBM cases. Here, $\sup_{\tau\in[0,1]}\widetilde{r}(\tau)$ is guaranteed to exist, since $\widetilde{r}(\tau)\to\sigma^2$ as $\tau\downarrow 0$, ensuring finiteness, and hence $\widetilde{r}$ is a continuous function on the compact interval $[0,1]$. A detailed calculation is omitted for brevity.
\end{document}